\documentclass[11pt]{article}

\usepackage{fullpage}
\usepackage{graphicx}
\usepackage{mathtools}
\usepackage{amsmath}
\usepackage{amsthm}
\usepackage{amssymb}
\usepackage{natbib}
\usepackage{mathtools}
\usepackage{hyperref}
\usepackage{bm}
\usepackage{xfrac}
\usepackage{todonotes}

\usepackage[noabbrev]{cleveref}
\usepackage{pgfplots}
\usetikzlibrary{patterns}
\usepgfplotslibrary{fillbetween}

\usepackage{amsmath,amsfonts,amssymb,amsthm,latexsym}
\usepackage{graphicx}
\usepackage{color}
\usepackage{bm}
\newtheorem{theorem}{Theorem}
\newtheorem{definition}{Definition}
\newtheorem{lemma}{Lemma}
\newtheorem{prop}{Proposition}

\newtheorem{corollary}{Corollary}

\newtheorem{claim}{Claim}

\newcommand{\agind}[1][i]{_{#1}}

\newcommand{\ironed}{\bar}
\newcommand{\constrained}{\hat}
\newcommand{\optconstrained}{\composed{\optimized}{\constrained}}
\newcommand{\optimized}[1]{#1\opt}
\newcommand{\differentiated}[1]{#1'}

\newcommand{\tagged}[2]{{#2}^{#1}}

\newcommand{\primedarg}[1]{#1\primed}
\newcommand{\noaccents}[1]{#1}
\newcommand{\composed}[3]{#1{#2{#3}}}

\newcommand{\newagentvar}[3][\noaccents]{%
\expandafter\newcommand\expandafter{\csname #2\endcsname}{#1{#3}}%
\expandafter\newcommand\expandafter{\csname #2s\endcsname}{#1{\boldsymbol{#3}}}%
\expandafter\newcommand\expandafter{\csname #2smi\endcsname}[1][i]{#1{\boldsymbol{#3}}_{-##1}}%
\expandafter\newcommand\expandafter{\csname #2i\endcsname}[1][i]{#1{#3}\agind[##1]}%
\expandafter\newcommand\expandafter{\csname #2ith\endcsname}[1][i]{#1{#3}_{(##1)}}%
}

\newcommand{\newitemvar}[3][\noaccents]{%
\expandafter\newcommand\expandafter{\csname #2\endcsname}{#1{#3}}%
\expandafter\newcommand\expandafter{\csname #2s\endcsname}{#1{\boldsymbol{#3}}}%
\expandafter\newcommand\expandafter{\csname #2smj\endcsname}[1][j]{#1{\boldsymbol{#3}}_{-##1}}%
\expandafter\newcommand\expandafter{\csname #2j\endcsname}[1][j]{#1{#3}_{##1}}%
\expandafter\newcommand\expandafter{\csname #2jth\endcsname}[1][j]{#1{#3}_{(##1)}}%
}

\newcommand{\forrezs}[1]{{#1}^{\rezs}}
\newagentvar[\forrezs]{imbal}{\phi}
\newagentvar[\forrezs]{cumimbal}{\Phi}

\newagentvar{alloc}{x}

\newagentvar{falloc}{z}

\newagentvar{quant}{q}
\newagentvar[\constrained]{exquant}{\quant}
\newagentvar[\constrained]{critquant}{\quant}  %remove?
\newagentvar[\optconstrained]{monoq}{\quant}
\newagentvar{Val}{\nu}
\newagentvar{toquant}{Q}

\newagentvar{qprice}{\price}
\newagentvar{qrev}{R}
\newagentvar[\ironed]{iqrev}{\qrev}

\newagentvar[\bar]{bstrat}{\strat}
\newagentvar[\bar]{bpay}{\pay}
\newagentvar[\bar]{bwal}{\wal}
\newagentvar[\bar]{bbid}{\bid}
\newagentvar{zpay}{\pay^{\rezs}}
\newagentvar{zwal}{\wal^{\rezs}}
\newagentvar{zdelt}{\Delta^{\rezs}}
\newagentvar{bomega}{\bm{\omega}}

\newagentvar{qalloc}{y}
\newagentvar{cumalloc}{Y}
\newagentvar[\constrained]{calloc}{\qalloc}
\newagentvar[\constrained]{cumcalloc}{\cumalloc}
\newagentvar[\ironed]{ialloc}{\qalloc}
\newagentvar[\ironed]{icumalloc}{\cumalloc}

\newcommand{\exposted}[1]{#1^{\text{\it EP}}}
\newagentvar[\composed{\exposted}{\constrained}]{excalloc}{\qalloc}
\newagentvar[\exposted]{exalloc}{\qalloc}
\newagentvar[\exposted]{extalloc}{\talloc}
\newagentvar[\exposted]{exfalloc}{\falloc}

\newagentvar{rev}{R}
\newagentvar[\differentiated]{marg}{\rev}
\newagentvar{rawrev}{P}
\newagentvar[\differentiated]{rawmarg}{\rawrev}

\newagentvar[\primedarg]{inducedrev}{\rev}

\newagentvar[\tilde]{pseudorev}{\rev}
\newagentvar[\tilde]{pseudorawrev}{\rawrev}

\newagentvar{typespace}{{\cal T}}
\newagentvar{typesubspace}{S}

\newagentvar{type}{t}
\newagentvar{othertype}{s}
\newagentvar{val}{v}
\newagentvar{outcome}{w}
\newagentvar{outcomespace}{{\cal W}}

\newcommand{\served}[1]{#1^1}
\newcommand{\nonserved}[1]{#1^0}
\newcommand{\alloced}[1]{#1^{\alloc}}
\newcommand{\allocedi}[1]{#1^{\alloci}}
\newagentvar[\alloced]{xoutcome}{\outcome}
\newagentvar[\allocedi]{xioutcome}{\outcome}
\newagentvar[\served]{soutcome}{\outcome}

\newagentvar[\nonserved]{nsoutcome}{\outcome}

\newagentvar{price}{p}
\newagentvar{talloc}{\alloc}
\newagentvar[\tilde]{eval}{\val}
\newagentvar[\tilde]{balloc}{\alloc}
\newagentvar[\tilde]{dalloc}{y}
\newagentvar[\tilde]{eballoc}{y}
\newagentvar{ealloc}{y}
\newagentvar[\tilde]{bprice}{\price}

\newagentvar{act}{a}
\newagentvar{bidspace}{A}
\newagentvar{actspace}{A}

\newagentvar[\constrained]{critval}{\val}
\newagentvar[\constrained]{crittype}{\type}
\newagentvar[\constrained]{critvirt}{\virt}
\newagentvar[\constrained]{reserve}{\val} % should this be used???
\newagentvar[\constrained]{bidreserve}{\bid} % should this be used???
\newagentvar[\optconstrained]{monop}{\val}
\newagentvar[\constrained]{monot}{\type}
\newagentvar[\optimized]{monorev}{\rev}

\newagentvar{rcalloc}{y}
\newagentvar[\optimized]{optrcalloc}{\rcalloc}
\newagentvar{biddist}{G}

\newagentvar[\constrained]{critbid}{\bid}
\newagentvar[\constrained]{cbid}{B}
\newagentvar[\primedarg]{wbid}{\bid}
\newagentvar{gfunc}{\vartheta}
\newagentvar{block}{C}

\newagentvar{util}{u}

\newagentvar{strat}{\beta}
\newagentvar{bid}{b}

\newagentvar{pay}{\pi}
\newagentvar{rez}{\rho}
\newagentvar[\tilde]{trez}{\rez}

\newagentvar{virt}{\phi}
\newagentvar{cumvirt}{\Phi}
\newagentvar{qvirt}{\phi}
\newagentvar[\ironed]{ivirt}{\virt}
\newagentvar[\ironed]{icumvirt}{\cumvirt}

\newagentvar[\tagged{\text{SD}}]{sdvirt}{\virt}
\newagentvar[\composed{\tagged{\text{SD}}}{\ironed}]{sdivirt}{\virt}
\newagentvar[\tagged{\text{MD}}]{mdvirt}{\virt}
\newagentvar[\composed{\tagged{\text{MD}}}{\ironed}]{mdivirt}{\virt}

\newagentvar{dist}{F}
\newagentvar{dens}{f}
\newagentvar{hazard}{h}
\newagentvar{cumhazard}{H}

\newagentvar[\ironed]{iprice}{\price}  %% what for?
\newagentvar[\ironed]{ival}{\val}  %% what for?

\newagentvar{ints}{{\cal I}}

\newagentvar{wal}{w}
\newagentvar{Util}{U}

\newitemvar{pos}{j}
\newitemvar{weight}{w}
\newitemvar[\differentiated]{mweight}{\weight}
\newitemvar{udtype}{\type}
\newitemvar[\constrained]{udcrittype}{\type}
\newitemvar{udalloc}{\alloc}
\newitemvar{udprice}{\price}
\newitemvar[\constrained]{udcalloc}{\qalloc}
\newitemvar[\constrained]{udcumcalloc}{\cumalloc}

\newagentvar{mech}{M}
\newagentvar[\skew{5}{\hat}]{cmech}{{\cal M}}
\newagentvar{alg}{{\cal A}}

\DeclareMathOperator{\OPT}{OPT}

\DeclareMathOperator{\APX}{APX}

\newcommand{\mecha}{{M}}

\newagentvar{trans}{\sigma}

\newagentvar{demandset}{S}

\newcommand{\opt}{^{\star}}
\newcommand{\primed}{^\dagger}

\newagentvar{distout}{w}
\newagentvar{idistout}{\bar{w}}

\newagentvar{gap}{\delta}

\newcommand{\valf}{V}
\newcommand{\ratio}{r}
\newcommand{\ScaledQuant}{\hat{Q}}
\DeclareMathOperator{\regu}{Reg}
\newcommand{\regular}{\feasibleDist^{\regu}}
\DeclareMathOperator{\tri}{Tri}
\DeclareMathOperator{\Trunc}{Trunc}
\newcommand{\TRIF}{\feasibleDist^{\tri}}
\newcommand{\TRUNCF}{\feasibleDist^{\Trunc}}
\newcommand{\QRF}{\feasibleDist^{\Qr}}
\DeclareMathOperator{\Qr}{Qr}
\DeclareMathOperator{\SI}{SI}
\DeclareMathOperator{\SMKUP}{SMKUP}
\newcommand{\scalei}{\feasibleMecha^{\SI}}
\newcommand{\stocmarkup}{\feasibleMecha^{\SMKUP}}

\newcommand{\distTri}{\dist^{\tri}}
\newcommand{\distTrunc}{\dist^{\Trunc}}
\newcommand{\distQr}{\dist^{\Qr}}

\newcommand{\vsec}{\val_{(2)}}

\newcommand{\mechaAR}{\mecha_{\alpha, r}}

\newcommand{\normBenchmark}{\mathcal{B}}
\newcommand{\feasibleDist}{\mathcal{F}}
\newcommand{\feasibleVals}{\mathcal{V}}
\newcommand{\benchmark}{B}

\newcommand{\feasibleMecha}{\mathcal{M}}
\newcommand{\qo}{\bar{\quant}}
\newcommand{\qt}{\bar{\quant}'\!}
\newcommand{\revq}{P_{\qt}}
\newcommand{\valq}{\valf_{\qt}}
\newcommand{\revqt}{P_{\qt'}}
\newcommand{\valqt}{\valf_{\qt'}}
\newcommand{\quantf}{Q}

\newcommand{\mechaARO}{\mecha_{\optweight, \optratio}}
\newcommand{\aval}{0.80564048}
\newcommand{\rval}{2.4469452}
\newcommand{\qval}{0.0931057}
\newcommand{\apxval}{1.9068943}
\newcommand{\asimp}{0.806}
\newcommand{\rsimp}{2.447}
\newcommand{\qsimp}{0.093}
\newcommand{\apxsimp}{1.907}
\newcommand{\revv}{\bar{\rev}}

\newcommand{\secprob}{\alpha}
\newcommand{\optf}[1]{\OPT_{#1}(#1)}
\newcommand{\optweight}{\secprob^*}
\newcommand{\optratio}{\ratio^*}
\newcommand{\optqo}{\qo^*}

\newcommand{\piratio}{\beta}
\newcommand{\resolution}{\gamma}
\newcommand{\res}{\rho}
\newcommand{\heuristic}{\bar{\resolution}}
\newcommand{\pfratio}{\rho}
\newcommand{\pflowerbound}{\bar{\pfratio}}

\newcommand{\valscale}{\alpha}

\DeclareMathOperator{\theBIH}{BIH}
\DeclareMathOperator{\IS}{BIS}
\DeclareMathOperator{\OL}{OL}

\DeclareMathOperator{\theFTL}{FTL}

\newcommand{\round}{t}
\newcommand{\horizon}{n}
\newcommand{\numexpert}{k}

\newcommand{\permute}{\sigma}
\newcommand{\ISD}{\feasibleDist^{\IS}}
\newcommand{\BIH}{\benchmark^{\theBIH}}
\newcommand{\FTL}{\mecha^{\theFTL}}
\newcommand{\OLA}{\feasibleMecha^{\OL}}

\newcommand{\alphaL}{0.81}
\newcommand{\alphaS}{0.8}
\newcommand{\ratioL}{2.449}
\newcommand{\ratioS}{2.445}
\newcommand{\qoL}{0.094}
\newcommand{\qoS}{0.093}

\newcommand{\APXSPA}{\APX_1}
\newcommand{\APXR}{\APX_{*}}

\DeclareMathOperator{\MKUP}{MKUP}
\DeclareMathOperator{\SPA}{SPA}

\newcommand{\pseudomech}{\tilde{\mech}}
\newcommand{\gapbm}{\benchmark}

\newcommand{\discreteQ}{Q_d}
\newcommand{\hatdiscreteQ}{\hat{Q}_d}
\newcommand{\precision}{\epsilon}
\newcommand{\hatprecision}{\hat{\epsilon}}
\newcommand{\discreteR}{R_d}
\newcommand{\precisionR}{\precision_r}

\newcommand{\hsact}{j^c}
\newcommand{\postact}{j^p}

\newcommand{\cumval}{V}

\newcommand{\mean}{f}

\DeclareMathOperator*{\argmax}{argmax}

\newcommand{\given}{\,\mid\,}

\newcommand{\prob}[2][]{\text{\bf Pr}\ifthenelse{\not\equal{}{#1}}{_{#1}}{}\!\left[{\def\givenn{\middle|}#2}\right]}
\newcommand{\expect}[2][]{\text{\bf E}\ifthenelse{\not\equal{}{#1}}{_{#1}}{}\!\left[{\def\givenn{\middle|}#2}\right]}

\newcommand{\tparen}{\big}
\newcommand{\tprob}[2][]{\text{\bf Pr}\ifthenelse{\not\equal{}{#1}}{_{#1}}{}\tparen[{\def\given{\tparen|}#2}\tparen]}
\newcommand{\texpect}[2][]{\text{\bf E}\ifthenelse{\not\equal{}{#1}}{_{#1}}{}\tparen[{\def\given{\tparen|}#2}\tparen]}

\newcommand{\sprob}[2][]{\text{\bf Pr}\ifthenelse{\not\equal{}{#1}}{_{#1}}{}[#2]}
\newcommand{\sexpect}[2][]{\text{\bf E}\ifthenelse{\not\equal{}{#1}}{_{#1}}{}[#2]}

\newif\iffocs
\focsfalse

\begin{document}

%\begin{frontmatter}


% \begin{titlepage}

\title{Benchmark Design and Prior-independent Optimization}

\newcommand{\email}[1]{\href{mailto:#1}{#1}}

%\author{Jason Hartline\thanks{Northwestern U., Evanston IL.  Work done in part while supported by NSF CCF 1618502. Email: \email{hartline@northwestern.edu}} \and Aleck Johnsen\thanks{Northwestern U., Evanston IL.  Work done in part while supported by NSF CCF 1618502. Email: \email{aleckjohnsen2012@u.northwestern.edu}}}

\author{
Jason Hartline\thanks{Northwestern U., Evanston, IL. Email: \email{hartline@u.northwestern.edu}.  Work done in part while supported by NSF CCF 1618502.} \and
Aleck Johnsen\thanks{Northwestern U., Evanston IL.  Email: \email{aleckjohnsen@u.northwestern.edu}.  Work done in part while supported by NSF CCF 1618502.} \and 
Yingkai Li\thanks{Northwestern U., Evanston IL. Email: \email{yingkai.li@u.northwestern.edu}.  Work done in part while supported by NSF CCF 1618502.}
}

\begin{titlepage}
\maketitle
\thispagestyle{empty}

\begin{abstract}
  This paper compares two leading approaches for robust optimization
  in the models of online algorithms and mechanism design.
  Competitive analysis compares the performance of an online algorithm
  to an offline benchmark in worst-case over inputs, and prior-independent
  mechanism design compares the expected performance of a mechanism on
  an unknown distribution (of inputs, i.e., agent values) to the
  optimal mechanism for the distribution in worst case over
  distributions.  For competitive analysis, a critical concern is the
  choice of benchmark.  This paper gives a method for selecting a good
  benchmark.  We show that optimal algorithm/mechanism for the
  optimal benchmark is equal to the prior-independent optimal
  algorithm/mechanism.

  We solve a central open question in prior-independent mechanism
  design, namely we identify the prior-independent revenue-optimal
  mechanism for selling a single item to two agents with i.i.d.\ and
  regularly distributed values. 
  We use this solution to solve the corresponding benchmark design problem. 
  Via this solution and the above
  equivalence of prior-independent mechanism design and competitive
  analysis (a.k.a.\ prior-free mechanism design) we show that the
  standard method for lower bounds of prior-free mechanisms is not
  generally tight for the benchmark design program.
\end{abstract}

\end{titlepage}

\section{Introduction}
\label{s:intro}

%%
%% summary of problem and paper.
%%

% (Possible Title? Benchmark Design and Prior-independent Optimization)

There are two leading approaches for robust optimization in the models
of mechanism design and online algorithms \iffocs(see detailed historical context deferred to the end of this section).  \else(see detailed historical discussion in \Cref{sec:related-work}).  \fi
A key property of these
environments is that with the given constraints (incentive
compatibility or online arrivals) the best outcome is unachievable
pointwise.  The first approach considers a benchmark and looks for a
mechanism (resp.\ algorithm) that pointwise approximates the
benchmark, a.k.a., prior-free approximation (resp.\ competitive
analysis).  The choice of benchmark is an important variable of this
approach.  The second approach assumes that the input is drawn at
random from an unknown distribution in a family and looks for a
mechanism that approximates, in worst-case over distributions in the
family, the performance of the Bayesian optimal mechanism for the
distribution, a.k.a., prior-independent approximation.  (Respectively,
this approach could also be applied to online algorithms.)  This paper
formalizes the problem of designing a good benchmark for prior-free
approximation (resp.\ competitive analysis) and connects this
benchmark design problem to prior-independent optimization.

Terminology and concrete examples from mechanism design are henceforth
adopted for the majority of the paper.  However, the main result
relating benchmark design to prior-independent optimization does not
rely on any specifics from mechanism design.  Instead it applies to
families of mechanisms represented by a family of functions from inputs
to performances.  In \Cref{sec:online-learning} we instantiate the
framework for online learning and give results that parallel the
specific developments for mechanism design.

%% There are two leading approaches for robust optimization in the models
%% of online algorithms and mechanism design.  A key property of these
%% environments is that with the given constraints (online arrivals or
%% incentive compatibility) the best outcome is unachievable pointwise.
%% The first approach considers a benchmark and looks for an
%% algorithm/mechanism that pointwise approximates the benchmark, a.k.a.,
%% competitive analysis.  The choice of benchmark is an important
%% variable of this approach.  The second approach assumes that the input
%% is drawn at random from an unknown distribution in a family and looks
%% for an algorithm/mechanism that approximates, in worst-case over
%% distributions in the family, the performance of the Bayesian optimal
%% algorithm/mechanism for the distribution, a.k.a., prior-independent
%% approximation.  This paper formalizes the problem of designing a good
%% benchmark for competitive analysis and connects this benchmark design
%% problem to prior-independent optimization.  See
%% \Cref{sec:related-work} for a detailed survey of the historical
%% context from online algorithms and mechanism design.

The choice of benchmark in the first approach impacts the ability of
approximation with respect to the benchmark to distinguish
between good and bad mechanisms.  On one hand a benchmark should be
an upper bound on what is achievable by a mechanism; otherwise,
approximating it does not necessarily mean that a mechanism is good.
On the other hand it should not be too loose an upper bound;
otherwise, neither good nor bad mechanisms can obtain good
approximations and the degree to which we can distinguish good and bad
mechanisms via the benchmark is limited.

As illustration, consider comparing the revenue of mechanisms for
selling a digital good to $n$ agents with values $\vals =
(\vali[1],\ldots,\vali[n])$ bounded on $[1,H]$ %\markup{should this be [1,h]?} 
to one of two
benchmarks, the {\em sum-of-values} benchmark $\sum_i \vali$ and the
{\em price-posting-revenue} benchmark $\max_i i\,\valith$ where
$\valith$ is the $i$th highest value.  To ensure benchmarks give an
upper bound on revenue, \citet{HR-08} suggested that the benchmark 
satisfies the property that, 
% \markup{what is the beginning of this sentence supposed to be?} 
for any distribution on inputs from a given
family of distributions, the expected benchmark (over the same
distribution of inputs) be at least the expected performance of the
optimal mechanism that knows the distribution.  Thus approximation of
the benchmark implies approximation of the optimal mechanism for any
of the distributions in the family.  We refer to this constraint as
{\em normalization}.  Both the sum-of-values and price-posting-revenue
benchmarks are normalized \citep{HR-08}.

Not all normalized benchmarks are equally good at discriminating
between good and bad mechanisms.  Consider the following two
mechanisms.  The {\em random-sampling} mechanism partitions the agents
at random and offers the optimal price from each part to the other
part \citep{GHKSW-06}. The {\em random-power-pricing} mechanism posts
a take-it-or-leave-it price drawn from the uniform distribution on
powers of two in $[1,H]$ \citep{GH-03}.  These mechanisms and the
benchmarks of the preceding paragraph are related as follows.  On all
inputs $\vals$, the sum-of-values benchmark exceeds the
random-power-pricing mechanism by a $\Theta(\log H)$ factor.  On all
inputs $\vals$, the random sampling mechanism and the
price-posting-revenue benchmark are $\Theta(1)$.  Moreover, the
performance of the latter benchmark and mechanism are always
sandwiched between the former benchmark and mechanism and can equal
either of them up to $\Theta(1)$.  From this analysis, we see that the
loose benchmark of sum-of-values does not discriminate between good
mechanisms like random-sampling and bad mechanisms like
random-power-pricing.

The preceding discussion suggests a benchmark design problem of
identifying the normalized benchmark to which the tightest
approximation is possible.  We refer to the tightest approximation
possible for a benchmark as its {\em resolution} (see
\Cref{def:resolution}), as up to this factor the benchmark cannot
distinguish between good and bad mechanisms.  The sum-of-values
benchmark has logarithmic resolution while the price-posting-revenue
benchmark has constant resolution.

In summary: A family of distributions over inputs induces a class of
normalized benchmarks.  The benchmark admitting the tightest
approximation, i.e.\ having smallest resolution, is optimal.  The
mechanism achieving this approximation is the prior-free optimal
mechanism for the benchmark.  Normalization of the benchmark implies
that the prior-free approximation factor of a mechanism (to the
benchmark) is at least its prior-independent approximation factor (for
the normalizing family of distributions).  A natural question is how
this prior-free approach, and its optimal mechanism, compares to
directly identifying the prior-independent optimal mechanism, i.e.,
the one with the best worst-case over distributions approximation to
the Bayesian optimal mechanism.

The first main result of the paper is the general result that optimal
benchmark design is equivalent to prior-independent optimization
(\Cref{sec:prior-free}).  Described in the context of mechanism
design, this result shows that the prior-free optimal mechanism for
the optimal benchmark is the prior-independent optimal mechanism and
that the optimal benchmark is simply the prior-independent optimal
mechanism scaled up by its approximation factor (so as to satisfy the
normalization constraint).\footnote{The analysis that proves this
  equivalence is straightforward and, perhaps, obvious in hindsight.
  As we will describe below, there is an alternative benchmark design
  program which we view, in hindsight, as a relaxation of our
  benchmark design program.  The prior literature strongly suggested
  that this alternative program and our program are equivalent; our
  third result shows that in fact the relaxation is lossy.}
Consequently, it is not possible to identify optimal benchmarks and
their corresponding optimal mechanisms for problems in which we are
unable to solve the prior-independent optimization problem.

%  \markup{This result may seem obvious in hindsight.  It should be noted that it establishes a first-of-its-type, strong connection between prior-free and prior-independent analysis, and does so for a very broad class of problems.  Further a related statement -- with constraints on one side of the equivalence relaxed in comparison to our result -- was conjectured prior to our work.  Our third main result, described further below, disproves the original conjecture generally, and leaves open the question of special problem classes for which it might be true.  See \Cref{sec:related-work} for context of the conjecture from past literature.  }

Our second main result is to solve the benchmark optimization problem
(equivalently: solve the prior-independent mechanism design problem)
for the problem of maximizing revenue from the sale of a single item
to two agents for the family of i.i.d.\ regular value distributions
(\Cref{sec:prior-independent}).  This result answers a major question
left open from \citet{DRY-15}, \citet{FILS-15}, and \citet{AB-18}.
The optimal mechanism is a mixture between the second-price auction,
where each agent is offered a price equal to the highest of the other
agents' values, and the auction where these prices are scaled up by a
factor of about 2.5.  Our solution to this central open question is
the first example of a prior-independent optimal mechanism that arises
as the solution to a non-trivial optimization problem and is not a
standard mechanism from the literature (though the mechanism does fall
into the lookahead family of mechanisms described by \citet{ron-01}
and has the same form as mechanisms used to prove bounds in
\citet{FILS-15} and \citet{AB-18}\textbf).  Our construction pins down
the worst-case family of distributions.  To solve the equivalent of
benchmark design and prior-independent optimization, the approach of
this paper is to directly solve the prior-independent optimization
problem and use that solution to solve the benchmark design problem.

A key component of optimal benchmark design is tight lower bounds on the prior-free approximation of a benchmark. 
% A key part of identifying the optimal benchmark is evaluating the
% tight approximation factor of a mechanism with respect to a benchmark.
There is a standard method for lower bounding the approximation ratio
of the best mechanism for a given benchmark \citep{GHKSW-06}.
Consider the distribution over inputs for which all mechanisms achieve
the same performance, e.g., for revenue maximization in mechanism
design this distribution is the so-called {\em equal-revenue}
distribution.  The ratio of the expected value of the benchmark on
this distribution to the revenue of any mechanism (all ``undominated" mechanisms are
the same) gives a lower bound on the approximation factor of any
mechanism to the benchmark. \citet{HM-05} proved that this approach is
tight for a large family of benchmarks and $n=3$ agents.
\citet{CGL-14} proved that this approach is tight for a large family
of benchmarks and a general number $n$ of agents.  Thus, a natural
approach to solve the benchmark optimization problem is to relax the
program to optimize, not the approximation factor of the best
mechanism, but the lower bound on the best approximation factor from
the above approach.  %(More historical context for this relaxation is discussed in \Cref{sec:related-work}.)  
Our third main result -- without
being able to explicitly solve this relaxed program -- is that the relaxation is not
without loss, i.e., it gives a normalized benchmark 
for which the lower bound is
not achievable by any mechanism (\Cref{sec:gap}).

Our proof of the third main result follows from the second main result
and a counter example that shows that there is a benchmark %for the relaxed program 
that achieves a lower objective value (for the relaxed
benchmark program) than the approximation ratio of the optimal
prior-independent auction (which our first result shows to be the optimal
objective value of the original benchmark program).
%See \Cref{sec:related-work} for a more detailed discussion why this relaxed program is important and was conjectured to be tight in the literature. 
%\markup{added to stress the importance of this result.}\markup{do we need both this comment here and footnote 1?  Perhaps we should move (the idea of) this sentence into footnote 1 (which I have tried to do in my re-write of footnote 1 above)}
%  I.e., in comparison to the benchmark set by scaled-up optimal prior-independent auction performance, the counter-example benchmark has smaller ``relaxed" resolution but is more difficult to approximate.

Our second and third results are proved under the restriction of
mechanisms that are (a) dominant strategy incentive compatible (DSIC), i.e.,
where truthtelling is a good strategy for each agent regardless of the
strategies of other agents, and (b) scale invariant.  It is known that
there are environments for prior-independent mechanism design where
DSIC is not without loss \citep{FH-18}.  The restricted family of DSIC mechanisms is interesting even if it is with loss; however,
far more study of non-incentive-compatible mechanisms is warranted.
It is not known whether scale invariance is without loss or not, though
\citet{AB-18} conjecture that it is without loss.  This question
is important and remains open.

There is an important negative interpretation of our first result,
that the prior-free benchmark optimization problem and the
prior-independent optimization problem give the same answer (which has
consequences for both mechanism design and online algorithms).  One
reason to prefer prior-free analysis over prior-independent analysis
is that it could be more robust.  Of course to be more robust, as
\citet{HR-08} have recommended, we need to choose a normalized
benchmark, i.e., one for which prior-free approximation implies
prior-independent approximation.  With many possible normalized
benchmarks, we need a method for selecting one.  We have adopted a
natural formal method for selecting one, namely, the one that admits
the tightest approximation is the best one.  However, the results of
the paper show that the answer we will then get from studying
prior-free approximation of this optimal benchmark is the same as the
answer we will get from the original prior-independent question.
Thus, there is no added robustness from the prior-free approximation
of the optimal benchmark (over prior-independent analysis).  Of
course, there are environments where increased robustness can
informally be observed from prior-free approximation of ad hoc, i.e.,
non-optimal, benchmarks.  This observations suggests the main open
question of this paper which is to identify a rigorous framework for
evaluating prior-free benchmarks which leads to an additional desired
robustness over the prior-independent framework.  We will formally
describe this shortcoming of the framework with the example
environment of no-regret online learning in
\Cref{sec:online-learning}.

\section{Benchmark Optimization Is Prior-independent Optimization}
\label{sec:prior-free}

We formulate the benchmark optimization problem in abstract terms and
prove that it is equivalent to prior-independent optimization.  Our
framework and results in this section hold generally for algorithm
design, however, we will adopt notation and terminology for mechanism
design to maintain consistency with the discussion in subsequent
sections.  Notably, the development of this section makes no
assumptions on the families of distributions over the input that are
considered.

Denote the space of inputs by $\feasibleVals$ and an input in this
space by $\vals$.  Denote a family of distributions over input space
by $\feasibleDist \subset \Delta(\feasibleVals)$ and a distribution in
the family by $\dist$.  Denote a family of feasible mechanisms by
$\feasibleMecha$ and a mechanism in this family by $\mecha$.  Denote a
family of benchmarks by $\normBenchmark$ and a benchmark in this
family by $\benchmark$.  For our purposes we will view both a
mechanism and a benchmark as a function that maps the input space to
an expected performance (e.g., in the case the mechanism is
randomized), denoted respectively by $\mecha(\vals)$ and
$\benchmark(\vals)$.  When evaluating the performance of a mechanism
or benchmark in expectation over the distribution we adopt the
short-hand notation $\mecha(\dist) = \expect[\vals \sim
  \dist]{\mecha(\vals)}$ and $\benchmark(\dist) = \expect[\vals \sim
  \dist]{\benchmark(\vals)}$.

In these abstract terms we formally define the Bayesian,
prior-independent, and prior-free optimization problems.

\begin{definition}
  The {\em Bayesian optimal mechanism design} problem is given by a distribution $\dist$ and family of mechanisms $\feasibleMecha$ and asks for the mechanism $\OPT_{\dist}$ with the maximum expected performance:
  \begin{align}
    \tag{$\OPT_{\dist}$}
    \OPT_{\dist} &= \argmax_{\mecha \in \feasibleMecha} \mecha(\dist).
  \end{align}
\end{definition}

\begin{definition}
\label{def:pimd}
  The {\em prior-independent mechanism design} problem is given by a
  family of mechanisms $\feasibleMecha$ and a family of distributions
  $\feasibleDist$ and solves the program
\begin{align}
  \label{eq:pi}
  \tag{$\piratio$}
\piratio &= \min_{\mecha \in \feasibleMecha} \max_{\dist \in \feasibleDist}
\frac{\OPT_{\dist}(\dist)}{\mecha(\dist)}.
\end{align}
\end{definition}

\noindent Next, $\pfratio^{\benchmark}$ is the approximation ratio of the optimal prior-free mechanism for benchmark $\benchmark$; in the subsequent discussion of benchmark design, we reinterpret $\pfratio^{\benchmark}$ as the {\em resolution} of benchmark $\benchmark$.%, where benchmarks with smaller resolution

\begin{definition}
\label{def:resolution}
  The {\em prior-free mechanism design} problem is given by a
  family of mechanisms $\feasibleMecha$ and a benchmark $\benchmark$
  and solves the program
  \begin{align}
    \label{eq:pfratio}
    \tag{$\pfratio^{\benchmark}$}
    \pfratio^{\benchmark} &= \min_{\mecha \in \feasibleMecha} \max_{\vals \in \feasibleVals}
    \frac{\benchmark(\vals)}{\mecha(\vals)}.
  \end{align}
% Further, the tightest approximation factor of any mechanism $\res^{\benchmark}$ identifies the {\em resolution} of the benchmark.
\end{definition}

Recall from the introduction, benchmarks with small resolution are
better at differentiating good mechanisms from bad one.  Both
prior-independent and prior-free mechanism design problems are
searching for mechanisms with robust performance guarantees.  In
principle, prior-free guarantees can provide more robustness than
prior-independent guarantees as the guarantee is required to hold
pointwise on all inputs rather than in expectation according to the
distribution.  Whether or not a prior-free guarantee is meaningful
depends on the choice of benchmark.  The possibility that some
benchmarks might be better than others for meaningfully quantifying
the performance of a mechanism suggests that benchmarks themselves can
be optimized.

\citet{HR-08} recommend restricting attention to benchmarks that
satisfy the following normalization property which requires that the
benchmark is an upper bound on the optimal performance of a mechanism.
Rather than measure this optimal performance pointwise as is common
with online algorithms, \citet{HR-08} recommend measuring this optimal
performance in expectation with respect to any distribution in a
family of distributions.  This way of measuring the optimal
performance can take into account the constraints on the mechanism, i.e.,
that $\mecha \in \feasibleMecha$.\footnote{For mechanism design, these
  constraints will be incentive compatibility and individual
  rationality. For online algorithms these constraints are that the
  current decision must be made before the future input is known.}
The normalization constraint implies a strong guarantee: A mechanism
that is a $\pfratio$ approximation to a normalized benchmark
guarantees a $\pfratio$ prior-independent approximation.

\begin{definition}
  A benchmark $\benchmark$ is {\em normalized} for a family of distributions
  $\feasibleDist$ and family of mechanisms if for every distribution
  in the family the expected benchmark is at least the optimal
  expected performance, i.e.,
  \begin{align*}
    \benchmark(\dist) &\geq \optf{\dist}, \quad \forall \dist \in \feasibleDist.
  \end{align*}
  Denote the normalized benchmarks for $\feasibleDist$ by
  $\normBenchmark(\feasibleDist)$.
  \end{definition}

\begin{prop}[\citealp{HR-08}]
  \label{prop:npf=>pi}
  If mechanism $\mecha$ is a prior-free $\pfratio$ approximation of a
  benchmark $\benchmark$ normalized to distributions $\feasibleDist$
  then its prior-independent approximation for
  distributions $\feasibleDist$ is at most $\pfratio$.
\end{prop}

The point of mechanism design is a principled method for choosing one
mechanism over another.  The prior-free framework described above
gives such a method only in so far as approximation of the benchmark
distinguished between good mechanisms and bad ones.  For example, no
mechanism will approximate a prior-free benchmark that is too large,
thus, good mechanisms will not necessarily be separated from bad ones.
Recall the discussion of the sum-of-values and posted-price-revenue
benchmarks in the introduction.  One way to quantify the inability of
a benchmark to discriminate is by considering the approximation factor
of the best mechanisms for the benchmark, i.e., the benchmark's
resolution (\Cref{def:resolution}).  Our benchmark design program aims
to identify the benchmark with the finest resolution.

\begin{definition}\label{def:optresolution}
  The {\em benchmark design} problem for family of distributions
  $\feasibleDist$, family of mechanisms $\feasibleMecha$, and space of
  inputs $\feasibleVals$ solves the program%following program, which identifies the normalized benchmark with optimal resolution:
\begin{align}\tag{$\resolution$}
\label{eq:resolution}
\resolution &= \min_{\benchmark\in \normBenchmark(\feasibleDist)} \res^{\benchmark} = \min_{\benchmark\in \normBenchmark(\feasibleDist)}
\min_{\mecha \in \feasibleMecha} \max_{\vals \in \feasibleVals}
\frac{\benchmark(\vals)}{\mecha(\vals)}.%,\quad\quad \res^{\benchmark} =\min_{\mecha \in \feasibleMecha} \max_{\vals \in \feasibleVals} \frac{\benchmark(\vals)}{\mecha(\vals)}.
\end{align}
\end{definition}

We are now ready to state and prove the main result of this section,
that the benchmark design problem and the prior-independent mechanism
design problem are equivalent.  Before doing so it should be noted
that it is not generally understood how to solve these problems.  (We
will, however, give a solution to a paradigmatic prior-independent
mechanism design problem in the next section.)

\begin{theorem}\label{thm:equiv}
For any family of distributions $\feasibleDist$, 
any set of mechanisms $\feasibleMecha$, benchmark design is equivalent to prior-independent mechanism design, i.e., $\resolution = \piratio$, and the optimal benchmark is given by the performance of the prior-independent optimal mechanism scaled up by its approximation ratio~$\piratio$.
\end{theorem}

This theorem follows from a corollary of \Cref{prop:npf=>pi} which
shows that $\piratio \leq \resolution$ and the following lemma which
shows that $\resolution \leq \piratio$.

\begin{corollary}
  For any families of distributions and mechanisms, the
  prior-independent optimal ratio $\piratio$ is at most the optimal
  benchmark ratio $\resolution$, i.e., $\piratio \leq \resolution$.
\end{corollary}
\begin{proof}
  By the definition of program~\eqref{eq:resolution}, the optimal mechanism for the optimal benchmark is a
  prior-free $\resolution$ approximation.  By \Cref{prop:npf=>pi} and
  the normalization of the optimal benchmark, this mechanism is a
  prior-independent $\resolution$ approximation.  The optimal
  prior-independent mechanism is no worse, i.e., $\piratio \leq \resolution$.
\end{proof}

\begin{lemma}
  For any families of distributions and mechanisms, the optimal
  benchmark ratio $\resolution$ is at most the prior-independent
  optimal ratio $\piratio$, i.e., $\resolution \leq \piratio$.
\end{lemma}

\begin{proof}
  Consider the prior-independent optimal mechanism $\mecha^*$ with
  approximation $\piratio$.  Define the benchmark
  \begin{align}
    \label{eq:scaled-up-benchmark}
    \benchmark^*(\vals) &=
    \piratio\,\mecha^*(\vals),\\
    \intertext{i.e., the benchmark is the
      performance of the prior-independent optimal mechanism scaled up
      by its approximation factor.  Taking the expectation of $\vals$
      drawn from any distribution $\dist$, we have}
    \label{eq:scaled-up-benchmark-expectation}
    \benchmark^*(\dist) &= \piratio\,\mecha^*(\dist).
  \end{align}

  First, notice that $\benchmark^*$ is normalized.  Since $\mecha^*$
  is a prior-independent $\piratio$-approximation, $\mecha^*(\dist)
  \geq \frac{1}{\piratio} \optf{\dist}$ for all $\dist$ in the family
  of distributions.  Multiplying through by $\piratio$ and applying
  equation~\eqref{eq:scaled-up-benchmark-expectation} shows that the benchmark meets the definition of normalization.

  Second, equation~\eqref{eq:scaled-up-benchmark} implies that
  $\mecha^*$ is a prior-free $\piratio$-approximation of
  $\benchmark^*$.  Thus, $(\mecha^*,\benchmark^*)$ is a solution to
  the benchmark design program \eqref{eq:resolution} with ratio
  $\piratio$.  The optimal solution to the program is no larger.
  Thus, $\resolution \leq \piratio$.
\end{proof}

The benchmark design problem presented above asks for both an optimal
benchmark and the optimal mechanism for this benchmark.  Fixing the
benchmark, the problem of identifying the optimal mechanism is not
well understood.  There is a canonical method for identifying lower
bounds on the approximation ratio of the optimal mechanism.
\citet{GHKSW-06} suggest that the prior-free approximation any
mechanism for a benchmark can be lower bounded by identifying a
distribution over inputs for which all non-dominated mechanisms obtain
the same performance.  For revenue maximizing mechanism design, this
distribution is the so-called {\em equal revenue} distribution.
%In
%fact, the ratio of the expected benchmark to the performance of any
%mechanism for this distribution gives a lower bound on the prior-free
%approximation of the benchmark, i.e., for the optimal prior-free
%mechanism there must exist an input where the ratio of benchmark to
%mechanism performance is at least the ratio of their expectations for
%the distribution.
%For digital goods revenue maximization and a large
%family of benchmarks, \citet{CGL-14} give a non-constructive proof
%that this lower bound is tight.
The following lemma and definition
generalize this lower bound to environments where a mechanism
neutralizing distribution may not exist.

\begin{lemma}
  \label{l:pflowerboundperdist}
  For any benchmark $\benchmark$, distribution $\dist$, and family of mechanisms which induce $\OPT_{\dist}$, the optimal prior-free approximation $\pfratio^{\mecha}$ is at least $ \frac{\benchmark(\dist)}{\optf{\dist}}$.
\end{lemma}

\begin{proof}
  Let $\mecha^{\benchmark}$ be the prior-free optimal mechanism for benchmark $\benchmark$, i.e., that optimizes the program~\eqref{eq:pfratio}, then
  \begin{align*}
  \pfratio^{\benchmark}
  = \max_{\vals \in \feasibleVals}\frac{\benchmark(\vals)}{\mecha^{\benchmark}(\vals)}
  \geq \frac{\benchmark(\dist)}{\mecha^{\benchmark}(\dist)}
  \geq \frac{\benchmark(\dist)}{\optf{\dist}}. \qedhere
  \end{align*}
\end{proof}

%\noindent Our goal is to use \Cref{l:pflowerboundperdist} to connect
%prior-free benchmark design to a prior-independent design setting
%given a specific family $\feasibleDist$.  Therefore next we identify
%$\pflowerbound^{\benchmark}$ as the largest (critical) lower bound
%from \Cref{l:pflowerboundperdist} resulting from a distribution in
%$\feasibleDist$.

\begin{definition}
  \label{def:pflowerbound}
  For any benchmark $\benchmark$, family of distributions
  $\feasibleDist$, and family of mechanisms which induce
  $\OPT_{\dist}$, the {\em canonical lower-bound on the optimal prior-free
  approximation of $\benchmark$} is
  \begin{align}
    \tag{$\pflowerbound^{\benchmark}$}
%    \label{eq:pflowerbound}
    \label{eq:pflowerboundfromdef}
    \pflowerbound^{\benchmark} &= \max_{\dist \in \feasibleDist} \frac{\benchmark(\dist)}{\optf{\dist}}.
  \end{align}
\end{definition}

This lower bound holds for any family of distributions; e.g.,
contrasting to standard assumptions in mechanism design, it allows
distributions that are irregular and correlated.  Moreover, the lower
bound is tight for a number of interesting benchmarks.  By
\Cref{thm:equiv}, the optimal benchmark $\benchmark^*$ satisfies
$\pfratio^{\benchmark^*} = \pflowerbound^{\benchmark^*}$\!.  For the
digital goods revenue maximization problem, a large family of
benchmarks were shown by \citet{CGL-14} to also satisfy this equality.
Thus, a natural relaxation of the benchmark design program is, instead
of optimizing benchmarks that admit the best prior-free approximation,
to optimize benchmarks that admit the best lower bound of
\Cref{def:pflowerbound}.

\begin{definition}
  \label{d:heuristic}
  The {\em relaxed benchmark design} problem for distribution family
  $\feasibleDist$ and family of mechanisms that defines the Bayesian
  optimal mechanism $\OPT_\dist$ for and distribution $\dist
  \in\feasibleDist$ solves the program:
\begin{align}\tag{$\heuristic$}
\label{eq:heuristic}
\heuristic &= \min_{\benchmark\in \normBenchmark(\feasibleDist)} \max_{\dist \in \feasibleDist} 
\frac{\benchmark(\dist)}{\optf{\dist}}.
\end{align}
\end{definition}

\begin{prop}
\label{prop:heuristiclessthanresolution} The value of the relaxed program lower bounds optimal resolution, i.e., $\bar{\gamma}\leq \gamma$.
\end{prop}

We will see in \Cref{sec:gap} that the relaxation is not generally
without loss and, in particular, the optimal benchmark for the relaxed
program can have $\heuristic < \resolution$.  We will show the
strictness of this inequality by example for the paradigmatic problem
of maximizing revenue from the sale of an item to one of two agents
and benchmarks that are normalized for value distributions that are
i.i.d.\ and regular.  We show this inequality is strict for online
learning as well.

\section{Prior-Independent Optimal Mechanisms}
\label{sec:prior-independent}

In this section we consider the problem of maximizing revenue from the
sale of a single item to one of two agents with values distributed
independently and identically from a distribution that satisfies a
natural convexity property (to be defined formally).  We identify the
prior-independent optimal mechanism and thus, by the equivalence
between benchmark design and prior-independent optimization, the
optimal benchmark.  The section begins with preliminary discussion of
mechanism design.

\subsection{Mechanism Design Preliminaries}

Consider $n$ agents with private values $\vals =
(\vali[1],\ldots,\vali[n])$.  The agents have linear utility given,
e.g, agent $i$'s utility is $\vali\,\alloci - \pricei$ for allocation
probability $\alloci$ and expected payment $\pricei$.  Agents' values
are drawn independently and identically from a product distribution
$\dists = \dist \times \cdots \times \dist$ where $\dist$ will denote
the cumulative distribution function of each agent's value.

A mechanism $\mecha$ is defined by an ex post allocation and payment
rule $\allocs^{\mecha}$ and $\prices^{\mecha}$ which map the profile
of values $\vals$ to a profile of allocation probabilities and a
profile of payments, respectively.  We focus on mechanisms that are
feasible, dominant strategy incentive compatible, and individually
rational:
\begin{itemize}
  \item For selling a single item, a mechanism is {\em feasible} if for all
    valuation profiles, the allocation probabilities sum to at most
    one, i.e., $\forall \vals,\ \sum_i\alloci^{\mecha}(\vals) \leq 1$.
  \item A mechanism is {\em dominant strategy incentive compatible} if
    no agent $i$ with value $\vali$ prefers to misreport some value
    $z$: $\forall \vals,i,z,\ \vali\,\alloci^{\mecha}(\vals) -
    \pricei^{\mecha}(\vals) \geq \vali\,\alloci^{\mecha}(z,\valsmi) -
    \pricei^{\mecha}(z,\valsmi)$ where $(z,\valsmi)$ denotes the
    valuation profile with $\vali$ replaced with $z$.
  \item A mechanism is {\em individually rational} if truthful
    reporting always leads to non-negative utility: $\forall
    \vals,i,\ \vali\,\alloci^{\mecha}(\vals) - \pricei^{\mecha}(\vals)
    \geq 0$.
\end{itemize}

A mechanism's revenue can be easily and geometrically understood via the
marginal revenue approach of \citet{mye-81} and \citet{BR-89}.
For distribution $\dist$, the \emph{quantile} $\quant$ of an agent
with value value $\val$ denotes how strong that agent is relative to
the distribution $\dist$.  Quantiles are defined by the mapping
$\quantf_{\dist}(\val) = 1 - \dist(\val)$.  Denote the mapping back to
value space by $\valf_{\dist}$, i.e., $\valf_{\dist}(\quant)$ is the
value of the agent with quantile~$\quant$.  A single agent {\em
  price-posting revenue curve} gives the revenue of posting a price as
a function of the probability that the agent accepts the price.  For
an agent with value distribution $\dist$, price $\valf_{\dist}(\quant) =
\dist^{-1}(1-\quant)$ is accepted with probability $\quant$; its
revenue is $\quant\,\valf_{\dist}(\quant)$.  A single agent {\em revenue
  curve} gives the optimal revenue from selling to a single agent
$\rev_{\dist}(\quant)$ as a function of ex ante sale probability
$\quant$. 
Note that the revenue curve $\rev$ is always concave. 
The agent is {\em regular} if her revenue curve equals her
price-posting revenue curve. 
% The revenue curve $\rev$ is concave if the agent is regular. \markup{added the concavity based on suggestions from Fu.}
The optimal mechanism for a single agent
posts the {\em monopoly price} $\valf_{\dist}(\monoq)$ which corresponds to
the monopoly quantile $\monoq = \argmax_\quant \rev_{\dist}(\quant)$.
The expected revenue of a multi-agent mechanism $\mech$ is equal to
its surplus of marginal revenue.

\begin{theorem}[\citealp{mye-81}]
  \label{thm:myerson}
  Given any incentive-compatible mechanism $\mecha$ 
with allocation rule $\allocs^{\mecha}(\vals)$, 
the expected revenue of mechanism $\mecha$
for agents with regular distribution $\dists$ is 
equal to its expected surplus of marginal revenue, 
i.e., 
\begin{equation*}
  \mecha(\dists) = \sum\nolimits_i \expect[\vals \sim \dists]{\pricei^{\mecha}(\vals)} = \sum\nolimits_i \expect[\vals \sim \dists]{
\rev'_{\dist}(\quantf_{\dist}(\vali))\,\alloci^{\mecha}(\vals) 
}.
\end{equation*}
\end{theorem}

\begin{corollary}[\citealp{mye-81}]
  For i.i.d., regular, single-item environments, the optimal mechanism
  $\OPT_{\dist}$ is the second-price auction with reserve equal to the
  monopoly price.
\end{corollary}

\begin{proof}
  Optimizing marginal revenue pointwise the item is assigned to the
  agent with the highest non-negative marginal revenue.  Since agents
  are i.i.d.\ and the marginal revenue curves are monotonically
  non-increasing, this winning agent is the one with the highest value
  that exceeds the monopoly price.
\end{proof}

The following lemma from \citet{DRY-15} follows from
\Cref{thm:myerson} and gives a geometric understanding of revenue in
two-agent auctions.

\begin{lemma}[\citealp{DRY-15}]
  \label{l:DRY-15}
  In i.i.d.\ two-agent single-item environments, the expected revenue
  of the second price auction is twice the area under the revenue
  curve and the expected revenue of the optimal mechanism is twice the
  area under the smallest monotone concave upper bound of the revenue
  curve.
\end{lemma}

\subsection{Prior-independent Optimization}

%% ORIGINAL
% 0. setting
% 1. assumptions
% 2. (defn) scale invariance
% 3. (theorem) main
% 4. (defn) lookahead
% 5. outline of proof
% 6. (defn)  triangle
% 7. (defn) quadrilateral
% 8. normalize distribution [max R(q) = 1]
% 9. (lemma) truncated optimal revenue = 2 - qbar
% 10. (lemma) characterization of lookahead mechanisms [this seems not to be needed?]
% 11. tiebreaking [can we not skip this?]
% 12. (lemma) optimality of lookahead mechanisms for triangles.
% 13. (theorem) PI for triangles.

%% REWRITE

% 0. setting
% 1. assumptions
% 2. (defn) scale invariance
% (defn) r-markup mechanism
% 3. (theorem) main
% 5. outline of proof
%    - pi auction for triangles
%    - best response to stochastic markup mechanisms is triangle.

%% subsection: PI for triangle distributions
% 8. normalize distribution [max R(q) = 1]
% 6. (defn)  triangle
% (defn) truncated distributions
% 9. (lemma) truncated optimal revenue = 2 - qbar
% 12. (lemma) optimality of stochastic markup mechanisms for triangles.
% 13. (theorem) PI for triangles.

% 7. (defn) quadrilateral

%% SKIP
% 4. (defn) lookahead (distributions over markup mechanisms)
% 10. (lemma) characterization of lookahead mechanisms [this seems not to be needed?]
% 11. tiebreaking [can we not skip this?]

In the remainder of this section we solve for the prior-independent
optimal mechanism for the revenue objective with the restriction to
\begin{itemize}
\item single-item, two-agent environments, i.e., $n=2$ (implicit);
\item the family of i.i.d.\ regular value distribution $\regular$; and
\item the family of feasible, incentive compatible, individually
  rational, and scale-invariant mechanisms $\scalei$.
\end{itemize}

The following discussion motivates these restrictions.  The
single-item two-agent environment is canonical for prior-independent
revenue maximization.  There do not exist good prior-independent
mechanisms for general asymmetric and irregularly distributed agent
values.  Almost all papers on prior-independent mechanism design
restrict to i.i.d.\ agents.  Almost all papers on revenue
maximization for prior-independent mechanism design restrict to
regular distributions.  The restriction to feasible and individually
rational mechanisms is required to have a sensible optimization
problem.  The restriction to incentive compatible mechanisms is made
in almost all papers on prior-independent mechanism design, an
exception is \citet{FH-18} where it is shown that the restriction can
be lossy.  The remaining condition which we formally define below is
scale invariance.

\begin{definition}
  Given any incentive-compatible mechanism $\mecha$ with allocation rule
  $\alloc^{\mecha}(\vals)$, mechanism $\mecha$ is \emph{scale invariant}
  if for each agent $i$, valuation profile $\vals$ and any constant
  $\valscale > 0$, $\alloci^{\mecha}(\valscale \cdot \vals) =
  \alloci^{\mecha}(\vals)$.  Scale invariance further implies $\mecha(a\cdot\vals)=a\cdot\mecha(\vals)$.
\end{definition}

%
% SCALE INVARIANCE.
%
\citet{AB-18} prove that the optimal prior-independent mechanism among
a broad family of mechanisms is scale invariant. They show that if
$\lim_{\valscale \to 0} \alloci(\valscale \cdot \vals)$ always exists
for mechanisms in the family, then the optimal prior-independent
mechanism is scale invariant.  They conjecture that this weaker
assumption is without loss; if true, the mechanism we identify as the
optimal mechanism among scale-invariant mechanisms is also
prior-independent optimal among all mechanisms.

Given the restriction
to scale-invariant mechanisms, it will be sufficient to consider
distributions that are normalized so that the single-agent optimal
revenue is $\max_\quant \rev(\quant) = 1$.

%Given that the mechanism is scale-invariant, both the revenue of the
%mechanism and the optimal revenue increase linearly with respect to
%the scale of the distribution.  Thus, it is sufficient for us to
%consider the distribution such that the single agent optimal revenue
%is normalized to $\max_{\quant} \rev(\quant) = 1$.

The following family of (stochastic) markup mechanisms is (essentially,
in $n=2$ agent environments) the restriction of the family of
lookahead mechanisms \citep{ron-01} to those that are scale invariant.
Notice that the second-price auction is the $1$-markup mechanism $\mecha_1$.

\begin{definition}
  \label{d:markup}
  The {\em $\ratio$-markup mechanism} $\mecha_{\ratio}$ identifies the
  agent with the highest-value (and ties broken uniformly at random)
  and offers this agent $\ratio$ times the second-highest value. A
  {\em stochastic markup mechanism} draws $\ratio$ from a given
  distribution on $[1,\infty)$.  The family of stochastic markup mechanisms is $\stocmarkup$.
\end{definition}

% \markup{my instinct is that the following statement should state that the worst case distribution against the given mechanism is ``worst-case regular" rather than ``worst-case triangle."  This is the ``global" result, isn't it?  Theorem 4 below restricts to triangles versus the setting here is ``regular"}

\begin{theorem}\label{thm:pi optimal mechanism}
  For i.i.d., regular, two-agent, single-item environments, the
  optimal scale-invariant, incentive-compatible mechanism for
  prior-independent optimization program \eqref{eq:pi} is $\mechaARO$
  which randomizes over the second-price auction $\mecha_1$ with
  probability $\optweight$ and $\optratio$-markup mechanism
  $\mecha_{\optratio}$ with probability $1-\optweight$, where
  $\optweight \approx \asimp$ and $\optratio\approx \rsimp$. The
  worst-case regular distribution for this mechanism is triangle distribution
  $\tri_{\optqo}$ with $\optqo \approx \qsimp$ and its approximation
  ratio is $\beta \approx \apxsimp$.
\end{theorem}

In the two sections below we prove this theorem with the following
main steps.  First, we characterize the prior-independent optimal
mechanism under the restriction to stochastic markup mechanisms and
triangle distributions, cf.\ \citet{AHNPY-18}.  This restricted
program has the same solution as is given in \Cref{thm:pi optimal
  mechanism}.  Second we show that the stochastic markup
mechanisms and triangle distributions are mutual best responses among
the more general families of scale-invariant mechanisms and regular
distributions.  Combining these results gives the theorem.

\subsection{Stochastic Markup Mechanisms versus Triangle Distributions}

In this section we characterize the solution to the prior-independent
optimization program restricted to stochastic markup mechanisms and
triangle distributions.  We first define triangle distributions, which have
revenue curves shaped like triangles (\Cref{f:quad}), as well as a
more general family of truncated distributions, which will be
important subsequently in the proof.  Recall that for scale-invariant
mechanisms, it is without loss to normalize the distributions to have
monopoly revenue one.

% \markup{does Tri subscript to Q need the $\vals$ subscript in this definition?}

\begin{definition}\label{def:triangle}
A \emph{normalized triangle distribution with monopoly quantile $\qo$},
denoted $\tri_{\qo}$, is defined by the quantile function
\begin{equation*}
\quantf_{\tri_{\qo}}(\val) =
\begin{cases}
\frac{1}{1 + \val(1-\qo)} & \val \leq
\sfrac{1}{\qo}\\ 0 & \text{otherwise.}
\end{cases}
\end{equation*} The {\em
    triangulation} of a normalized distribution with monopoly quantile
$\qo$ is $\tri_{\qo}$.  The family of normalized triangle
distributions is $\TRIF = \{\tri_{\qo} : \qo \in [0,1]\}$.
\end{definition}

\begin{definition}
  \label{def:truncate}
  A distribution is {\em truncated} if the highest-point in its
  support is the monopoly price (typically a point mass).  The {\em
    truncation} of a distribution is the distribution that replaces
  every point above the monopoly price with the monopoly price.  The
  family of truncated distributions is denoted $\TRUNCF$.
\end{definition}

The three lemmas below give formulae for the revenue of
the optimal mechanism, the second-price auction, and non-trivial
markup mechanisms for triangle distributions.  The formula for
revenue of markup mechanisms is discontinuous at $\ratio = 1$.  Thus,
in our discussion we will distinguish between the second-price auction
$\mech_1$ and non-trivial markup mechanism $\mech_\ratio$ for $\ratio
> 1$.

\begin{lemma}\label{lem:opt rev}
  For i.i.d., normalized truncated, two-agent, single-item
  environments, the optimal mechanism posts the monopoly price and
  obtains revenue $2-\qo$ where $\qo$ is the probability that an
  agent's value equals the monopoly price.
\end{lemma}

\begin{proof}
  The smallest monotone concave function that upper bounds the revenue
  curve is a trapezoid; its area is $\sfrac{\qo}{2} + 1-\qo$.  The
  optimal revenue from two agents, by \Cref{l:DRY-15}, is twice this
  area, i.e., $2-\qo$.
\end{proof}

\begin{lemma}
  \label{lem:spa rev}
The revenue of the second-price auction $\mech_1$ for distribution
$\tri_{\qo}$ is 1, i.e., $\mech_1(\tri_{\qo}) = 1$.
\end{lemma}

\begin{proof}
  By \Cref{l:DRY-15}, the revenue is twice the area under the revenue
  curve.  That area is $\sfrac 1 2$; thus, the revenue is 1.
\end{proof}

\begin{lemma}\label{clm:rev of r}
  The revenue of the $\ratio$-markup mechanisms $\mech_{\ratio}$ on
  triangle distribution $\tri_{\qo}$, for $\ratio \in (1,\infty)$ and
  $\qo \in [0,1)$, is
$$
\mecha_{\ratio}(\tri_{\qo}) 
= 
\frac{2\ratio}{(1-\qo)(\ratio-1)} 
\left(
\frac{1-\qo}{1-\qo+\qo\ratio}
+
\frac{\ln \left(
\frac{\ratio}{1-\qo+\qo\ratio}
\right)}{1-\ratio}
\right). 
$$
\end{lemma}

The proof of \Cref{clm:rev of r} is straightforward and given in
\Cref{app:prior-independent}.

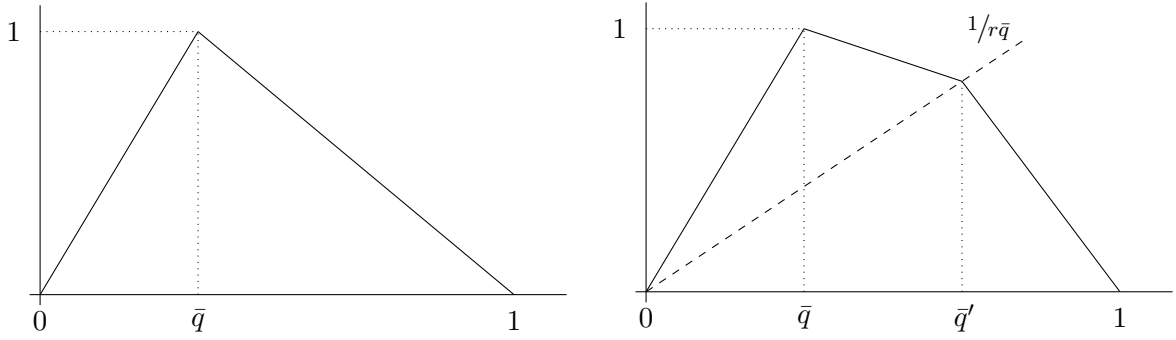
\begin{figure}[t]
\begin{flushleft}
\begin{minipage}[t]{0.48\textwidth}
\centering
\iffocs
\newcommand{\TRISCALE}{0.35}
\else
\newcommand{\TRISCALE}{0.7}
\fi

\begin{tikzpicture}[scale = \TRISCALE]

\draw (-0.2,0) -- (10, 0);
\draw (0, -0.2) -- (0, 5.5);

\draw (0, 0) -- (3, 5);
\draw (3, 5) -- (9, 0);

\draw [dotted] (0, 5) -- (3, 5);
\draw [dotted] (3, 0) -- (3, 5);

\draw (-0.5, 5) node {$1$};

\iffocs
\draw (0, -0.7) node {$0$};
\draw (9, -0.7) node {$1$};
\draw (3, -0.7) node {$\qo$};
\else
\draw (0, -0.5) node {$0$};
\draw (9, -0.5) node {$1$};
\draw (3, -0.5) node {$\qo$};
\fi

\end{tikzpicture}
\end{minipage}
\begin{minipage}[t]{0.48\textwidth}
\centering
\iffocs
\newcommand{\QUADSCALE}{0.35}
\else
\newcommand{\QUADSCALE}{0.7}
\fi

\begin{tikzpicture}[scale = \QUADSCALE]

\draw (-0.2,0) -- (10, 0);
\draw (0, -0.2) -- (0, 5.5);

\draw (0, 0) -- (3, 5);
\draw (3, 5) -- (6, 4);
\draw (6, 4) -- (9, 0);

\draw [dotted] (0, 5) -- (3, 5);
\draw [dashed] (0, 0) -- (7.2, 4.8);
\draw [dotted] (3, 0) -- (3, 5);

\draw [dotted] (6, 0) -- (6, 4);

% \draw (-0.2, 4.8) -- (0.2, 4.8);
\draw (-0.5, 5) node {$1$};

\iffocs
\draw (0, -0.7) node {$0$};
\draw (9, -0.7) node {$1$};
\draw (3, -0.7) node {$\qo$};
\draw (6, -0.7) node {$\qt$};
\else
\draw (0, -0.5) node {$0$};
\draw (9, -0.5) node {$1$};
\draw (3, -0.5) node {$\qo$};
\draw (6, -0.5) node {$\qt$};
\fi

\draw (6.5, 5) node {$\sfrac{1}{r\qo}$};

\end{tikzpicture}
\end{minipage}
\vspace{-18pt}
\end{flushleft}
\caption{\label{f:quad}
The left hand side is the revenue curve for triangle distribution $\tri_{\qo}$
and the right hand side is the revenue curve for quadrilateral distribution $\Qr_{\qo, \qt, \ratio}$. 
The definition of quadrilateral distribution $\Qr_{\qo, \qt, \ratio}$ will be formally introduced later in \Cref{sec:best response}.
}
\end{figure}

The following theorem characterizes the prior-independent optimal
stochastic markup mechanism against triangle distributions.  The
parameters of this optimal mechanism are the solution to an algebraic
expression (cf.\ \Cref{clm:rev of r}) that we are unable to solve
analytically.  Our proof will instead combine numeric calculations of
select points in parameter space with theoretical analysis to rule out
most of the parameter space.  For the remaining parameter space, we
can show that the expression is well-behaved and, thus, numeric
calculation can identify near optimal parameters.  Discussion of this
hybrid numerical and theoretical analysis can be found in
\Cref{app:prior-independent}.

Before giving \Cref{thm:lower}, we give context for its proof.  In abstract terms, the prior-independent optimization program \eqref{eq:pi} can be viewed as a zero sum game between the designer and an adversary, 
where the designer chooses a prior-independent mechanism $\mecha$, 
the adversary chooses a worst-case distribution $\dist$ (and its induced revenue curve), 
and the payoff of the designer is the approximation ratio $\OPT_{\dist}(\dist)/\mecha(\dist)$ (see \Cref{def:pimd}).

\begin{theorem}\label{thm:triangle}\label{thm:lower}
    For i.i.d., triangle distribution, two-agent, single-item
    environments, the optimal stochastic markup mechanism for
    prior-independent optimization program \eqref{eq:pi} is
    $\mechaARO$ which randomizes over the second-price auction
    $\mecha_1$ with probability $\optweight$ and $\optratio$-markup
    mechanism $\mecha_{\optratio}$ with probability $1-\optweight$,
    where $\optweight \approx \asimp$ and $\optratio\approx
    \rsimp$. The worst-case distribution for this mechanism is the
    triangle distribution $\tri_{\optqo}$ with $\optqo \approx \qsimp$
    and its approximation ratio is $\beta \approx \apxsimp$.
\end{theorem}
\begin{figure}[t]
  \begin{flushleft}
    \begin{minipage}[t]{0.48\textwidth}
      \centering
      \iffocs
\newcommand{\APXSCALE}{0.35}
\else
\newcommand{\APXSCALE}{0.7}
\fi

\begin{tikzpicture}[scale = \APXSCALE]

\draw (-0.2,0) -- (9.5, 0);
\draw (0, -0.2) -- (0, 4.5);

\draw [dotted] (0, 0.9) -- (9, 0.9);
\draw [dotted] (9, 0) -- (9, 0.9);

\draw [dotted] (0.838, 0) -- (0.838, 1.72);

\begin{scope}[thick]
\draw (0, 1.8) -- (9, 0.9);

%price posting rev curve
\draw plot [smooth, tension=0.8] coordinates {
(0, 0.9)
(0.18, 1.3136282419623115)
(0.36, 1.4640131820450633)
(0.54, 1.5748078501498688)
(0.72, 1.664835789425897)
(0.9, 1.7413094575111476)
(1.08, 1.8079002132369764)
(1.2600000000000002, 1.8667921201394844)
(1.44, 1.9194246152090024)
(1.6199999999999999, 1.9667793471684476)
(1.8, 2.0095760109323155)
(1.98, 2.04837544425269)
(2.16, 2.083615029815598)
(2.34, 2.116216216216216)
(2.5200000000000005, 2.15)
(2.6999999999999997, 2.185714285714286)
(2.88, 2.223529411764706)
(3.06, 2.2636363636363637)
(3.2399999999999998, 2.30625)
(3.42, 2.351612903225807)
(3.6, 2.4000000000000004)
(3.78, 2.4517241379310346)
(3.96, 2.507142857142857)
(4.140000000000001, 2.5666666666666664)
(4.32, 2.6307692307692307)
(4.5, 2.7)
(4.68, 2.7750000000000004)
(4.86, 2.856521739130435)
(5.040000000000001, 2.9454545454545458)
(5.22, 3.0428571428571427)
(5.3999999999999995, 3.1499999999999995)
(5.58, 3.2684210526315787)
(5.76, 3.4)
(5.94, 3.5470588235294116)
(6.12, 3.7125)
(6.3, 3.9)
(6.4799999999999995, 4.114285714285714)
(6.66, 4.361538461538461)
(6.84, 4.65)
(7.0200000000000005, 4.990909090909091)
};

\end{scope}

\iffocs
%\draw (0, -0.7) node {$0$};
\draw (9, -0.7) node {$1$};
\draw (1, -0.7) node {$.093$};
\draw (-0.6, 0) node {$0$};
\draw (-0.6, 0.9) node {$1$};
\draw (-0.6, 1.8) node {$2$};
\else
\draw (0, -0.5) node {$0$};
\draw (9, -0.5) node {$1$};
\draw (1, -0.5) node {$0.093$};
\draw (-0.4, 0) node {$0$};
\draw (-0.4, 0.9) node {$1$};
\draw (-0.4, 1.8) node {$2$};
\fi

\end{tikzpicture}
    \end{minipage}    
    \begin{minipage}[t]{0.48\textwidth}
      \centering
      \iffocs
\newcommand{\MKUPSCALE}{0.35}
\else
\newcommand{\MKUPSCALE}{0.7}
\fi

\begin{tikzpicture}[scale = \MKUPSCALE]

\draw (-0.2,0) -- (9.5, 0);
\draw (0, -0.2) -- (0, 4.5);

\draw [dotted] (9, 0) -- (9, 1.8318570061324486);
\draw [dotted] (0, 4) -- (2.52, 4);
\draw [dotted] (2.52, 0) -- (2.52, 4);

\draw[thick] (0,2.1378860000000017) circle (0.1cm);
\fill[black] (0,4) circle (0.1cm);

\begin{scope}[thick]

%price posting rev curve
\draw plot [smooth, tension=0.8] coordinates {
(0, 2.1378860000000017)
(0.18000000000000016, 2.4858546236856434)
(0.35999999999999993, 2.7780923762522782)
(0.5400000000000001, 3.0236057552038424)
(0.7199999999999999, 3.2295397242542947)
(0.9, 3.401652011704517)
(1.0800000000000003, 3.54464523357451)
(1.26, 3.662405252282408)
(1.4400000000000002, 3.758175921484643)
(1.6199999999999999, 3.8346895991501135)
(1.8, 3.894266236323606)
(1.9800000000000002, 3.9388897117643253)
(2.1600000000000006, 3.9702674089368877)
(2.34, 3.9898772626667665)
(2.52, 3.9990053071619727)
(2.7, 3.9987759336185498)
(2.8800000000000003, 3.9901764886222963)
(3.0600000000000005, 3.974077433896767)
(3.2399999999999998, 3.951248991475193)
(3.42, 3.922374981523067)
(3.6, 3.8880643995024613)
(3.7800000000000002, 3.8488611591964172)
(3.9600000000000004, 3.8052523372212086)
(4.14, 3.7576751852487114)
(4.32, 3.706523122688221)
(4.5, 3.6521508810329983)
(4.680000000000001, 3.594878938548881)
(4.86, 3.534997358325196)
(5.04, 3.472769122332881)
(5.22, 3.408433037845761)
(5.4, 3.3422062794818572)
(5.579999999999999, 3.2742866195242932)
(5.760000000000001, 3.204854390562332)
(5.9399999999999995, 3.1340742174459386)
(6.120000000000001, 3.0620965497569976)
(6.3, 2.989059021220733)
(6.4799999999999995, 2.915087658517148)
(6.66, 2.8402979586524637)
(6.84, 2.76479585129049)
(7.0200000000000005, 2.688678560128462)
(7.2, 2.6120353754500805)
(7.38, 2.534948348339359)
(7.5600000000000005, 2.4574929156391434)
(7.74, 2.3797384635475645)
(7.920000000000001, 2.3017488367286383)
(8.1, 2.223582798942793)
(8.28, 2.1452944504553884)
(8.46, 2.0669336068375266)
(8.64, 1.9885461432174623)
(8.82, 1.9101743075598314)
(9.0, 1.8318570061324486)
};

\end{scope}

\iffocs
\draw (0, -0.7) node {$1$};
\draw (9, -0.7) node {$5$};
\draw (2.52, -0.7) node {$2.45$};
\draw (-0.9, 0) node {$.8$};
\draw (-0.9, 2.137) node {$.91$};
\draw (-0.6, 4) node {$1$};
\else
\draw (0, -0.5) node {$1$};
\draw (9, -0.5) node {$5$};
\draw (2.52, -0.5) node {$2.45$};
\draw (-0.7, 0) node {$0.8$};
\draw (-0.7, 2.137) node {$0.91$};
\draw (-0.4, 4) node {$1$};
\fi

\end{tikzpicture}
    \end{minipage}
    \vspace{-18pt}
  \end{flushleft}
  \caption{\label{f:rev-approx} The figure on the left plots, as a function of $\qo$, the
    approximation ratio $\APXSPA(\qo)$ of the second-price auction
    $\mecha_1$ against triangle distribution $\tri_{\qo}$ (straight line), and the approximation ratio $\APXR(\qo)$ of the optimal non-trivial
    markup mechanism against triangle distribution $\tri_{\qo}$ (curved line).  These functions cross at $\optqo = \qval$.  The figure on the
    right plots the revenue of the $\ratio$ markup mechanism
    $\mech_{\ratio}$ on triangle distribution $\tri_{\optqo}$ as a
    function of markup $\ratio$, i.e.,
    $\mech_{\ratio}(\tri_{\optqo})$.  Notice that, by choice of $\optqo$, the optimal
    non-trivial markup mechanism has the same revenue as the
    second-price auction.}
\end{figure}
\begin{proof}
The approach of this proof is to identify the triangle
  $\tri_{\optqo}$ for which the designer is indifferent between the
  second price auction $\mech_1$ and the optimal (non-trivial) markup
  mechanism, denoted $\mech_{\optratio}$.  For such a distribution
  $\tri_{\optqo}$, the designer is also indifferent (in minimizing the
  approximation ratio) between any mixture over $\mech_1$ (with
  probability~$\secprob$) and $\mech_{\optratio}$ (with probability
  $1-\secprob$), and all other $\ratio$-markup mechanisms for $\ratio
  \not \in \{1,\optratio\}$ are inferior (\Cref{f:rev-approx}).  We
  then identify the $\optweight$ for which the adversary's best
  response (in maximizing the approximation ratio) to $\mechaARO$ is
  the distribution $\tri_{\optqo}$.  This solution $\mechaARO$ and
  $\tri_{\optqo}$ is a Nash equilibrium between the designer and
  adversary and, thus, it solves the prior-independent optimization
  problem.  The parameters can be numerically identified as
  $\optweight \approx \aval$, $\optratio\approx\rval$, $\optqo \approx
  \qval$, and the approximation ratio is $\piratio \approx \apxval$.

  First, we identify the triangle distribution $\optqo$ and the
  $\optratio$ for which $\mech_1$ and $\mech_{\optratio}$ obtain the
  same ratio.  Denote the approximation ratio for the second-price
  auction $\mech_1$ as $\APXSPA(\qo) = 2-\qo$ (the ratio of
  \Cref{lem:opt rev} to \Cref{lem:spa rev}), which is continuous in
  $\qo$.  Denote the approximation ratio of the optimal markup
  mechanism against distribution $\tri_{\qo}$ by $\APXR(\qo) =
  \sup_{\ratio > 1}
  \frac{\OPT_{\tri_{\qo}}(\tri_{\qo})}{\mech_{\ratio}(\tri_{\qo})}$.
  By \Cref{clm:rev of r}, the approximation ratio $\APXR(\qo)$ is
  continuous in $\qo$ as well.  It is easy to verify that $\APXSPA(0)
  = 2 > \APXR(0) = 1$ while $\APXSPA(1) = 1 < \APXR(1) = \infty$.  By
  continuity, there exists a $\optqo$ where these two functions cross,
  i.e., $\APXR(\optqo) = \APXSPA(\optqo)$.  See \Cref{f:rev-approx}. By numerical calculation,
  $\optqo \approx \qval$, and
  $$ \optratio = \argmax_{\ratio > 1}
  \frac{\optf{\tri_{\optqo}}}{\mech_{\ratio}(\tri_{\optqo})}
  \approx \rval.
  $$
  
  Now, fixing $\optratio$, we search for $\optweight$ for which the
  adversary maximizes the approximation ratio of mechanism $\mechaARO$
  by selecting triangle distribution $\tri_{\optqo}$.  Denote by
  $\qo_{\ratio}(\alpha)$ the monopoly quantile as a function of
  $\alpha$ for the triangle distribution that maximizes the
  approximation ratio of mechanism $\mechaAR$, i.e.,
$$\qo_{\ratio}(\alpha) = \argmax_{\qo}
  \frac{\optf{\tri_{\qo}}}{\mechaAR(\tri_{\qo})}.$$

  By numerical calculation, for any $\ratio \in [\ratioS, \ratioL]$,
  $\qo_{\ratio}(\alphaL) < \optqo < \qo_{\ratio}(\alphaS)$.
  Continuity of $\qo_{\ratio}(\cdot)$ for $\ratio \in [\ratioS, \ratioL]$ and
  $\alpha \in [\alphaS, \alphaL]$ (formally proved in
  \Cref{app:prior-independent}), then implies that there exists $\optweight$
  such that $\qo_{\optratio}(\optweight) = \optqo$.  By numerical
  calculation, $\optweight \approx \aval$.
  \end{proof}

\subsection{Mutual best-response of Stochastic Markup Mechanisms
  and Triangle Distributions}
  \label{sec:best response}

In this section we show that stochastic markup mechanisms are a best
response (for the designer) to truncated distributions and that
truncated distributions are a best response (for the adversary) to
stochastic markup mechanisms.  Moreover, we show that among truncated
distributions, triangle distributions are the best for the adversary.
Triangle distributions are known to be worst case for other questions
of interest in mechanism design, e.g., approximation by anonymous
reserves and anonymous pricings \citep{AHNPY-18}.  The proof that
triangle distributions are worst-case for two-agent prior-independent
revenue maximization is significantly more involved than these
previous results.

\begin{theorem}\label{lem:lookahead}
  For i.i.d., two-agent, single-item environments and any
  scale-invariant incentive-compatible mechanism $\mech$, there is a
  stochastic markup mechanism $\mech'$ with (weakly) higher revenue
  (and weakly lower approximation ratio) on every truncated
  distribution $\dist$.  I.e., $\mech'(\dist) \geq \mech(\dist)$.
\end{theorem}

\begin{proof}
  In a stochastic markup mechanism the price of the higher agent is a
  stochastic multiplicative factor $\ratio \geq 1$ of the value of the
  lower agent (with ties broken randomly).  To prove this theorem we
  must argue that (a) if the agents are not tied, then revenue
  improves if the lower agent loses, (b) if the agents are tied, then
  revenue is unaffected by random tie-breaking, and (c) any such
  scale-invariant mechanism looks to the higher-valued agent like a
  stochastic posted pricing with price that is a multiplicative factor
  (at least one) of the lower-valued agent's value.

  To see (a), note that the revenue of the mechanism is equal to its
  virtual surplus (\Cref{thm:myerson}) and for triangle distributions
  only the highest value in the support of the distribution has
  positive virtual value.  Thus, any mechanism that sells to a
  strictly-lower-valued agent can be improved by not selling to such
  an agent.

  To see (b), note that for any i.i.d.\ distribution the revenue of any
  mechanism is invariant to randomly permuting the identities of the
  agents.  Thus, we can assume random tie-breaking.

  To see (c), recall that the family of incentive-compatible
  single-agent mechanisms is equivalent to the family of random price
  postings.  Once we have ruled out selling to the lower-valued agent,
  the mechanism is a single-agent mechanism for the higher-valued
  agent (with price at least the lower-valued agent's value.  By the
  assumption that the mechanism is scale invariant, the distribution
  of prices offered to the higher-valued agent must be multiplicative
  scalings of the lower-valued agent's value.
\end{proof}

Next we will give a sequence of results that culminate in the
observation that for any regular distribution and any stochastic
markup mechanism with probability $\secprob$ at least $\sfrac 2 3$ on
the second-price auction (which includes the optimal mechanism from
\Cref{thm:triangle}) either the triangulation of the distribution or
the point mass $\tri_{1}$ has (weakly) higher approximation ratio.  As
the notation indicates, the point mass distribution $\tri_{1}$ is a
triangle distribution.

\begin{theorem}\label{thm:tri}
  For i.i.d., two-agent, single-item environments and any regular
  distribution $\dist$ and any stochastic markup mechanism $\mech$
  that places probability $\alpha \in [\sfrac 2 3,1]$ on the
  second-price auction, either the triangulation of the distribution
  $\distTri$ or the point mass $\tri_{1}$ has (weakly) higher
  approximation ratio.  I.e.,
  $\max\left\{\frac{\optf{\distTri}}{\mech(\distTri)},\frac{\optf{\tri_1}}{\mech(\tri_1)}\right\}
  \geq \frac{\optf{\dist}}{\mech(\dist)}$.
  \end{theorem}

To prove this theorem we give a sequence of results showing that for
any regular distribution, a corresponding truncated distribution is only
worse; for any truncated distribution and a fixed stochastic markup
mechanism (that mixes over $\mech_1$ and some $\mech_{\ratio}$), a
corresponding quadrilateral distribution (based on $\ratio$) is only
worse; and for any quadrilateral distribution, a
corresponding triangle distribution (independent of $\ratio$) is only
worse.  The theorem follows from combining these results.  The first
step assumes that the probability that the stochastic markup mechanism
places on the second price auction is $\secprob \in [\sfrac 1 2,1]$;
the last step further assumes that $\secprob \in [\sfrac 2 3, 1]$.

To begin, the following lemma shows that the best response of the
adversary to a relevant stochastic markup mechanism is a truncated
distribution.  Recall that by \citet{FILS-15} the prior-independent
optimal mechanism is strictly better than a 2-approximation.  On the
other hand, any stochastic markup mechanism that places probability
$\secprob$ on the second-price auction $\mech_1$ has prior-independent
approximation at least $\sfrac 1 {\secprob}$. Specifically, on the
(degenerate) distribution that places all probability mass on 1,
a.k.a.\ $\tri_{1}$, the approximation factor of such a stochastic
markup mechanism is exactly $\sfrac 1 {\secprob}$.  We conclude that
all relevant stochastic markup mechanisms place probability $\secprob
> \sfrac 1 2$ on the second-price auction.  Thus, this lemma applies
to all relevant mechanisms.

\begin{lemma}\label{lem:truncate}
  For i.i.d., two-agent, single-item environments, any regular
  distribution $\dist$, and any stochastic markup mechanism $\mech$
  that places probability $\alpha \in [\sfrac 1 2,1]$ on the
  second-price auction; either the truncation of the distribution
  $\dist'$ or the point mass distribution $\tri_1$ has (weakly) higher
  approximation ratio.  I.e.,
  $\max\left\{\frac{\optf{\dist'}}{\mech(\dist')},\frac{\optf{\tri_1}}{\mech(\tri_1)}\right\}
  \geq \frac{\optf{\dist}}{\mech(\dist)}$.
\end{lemma}

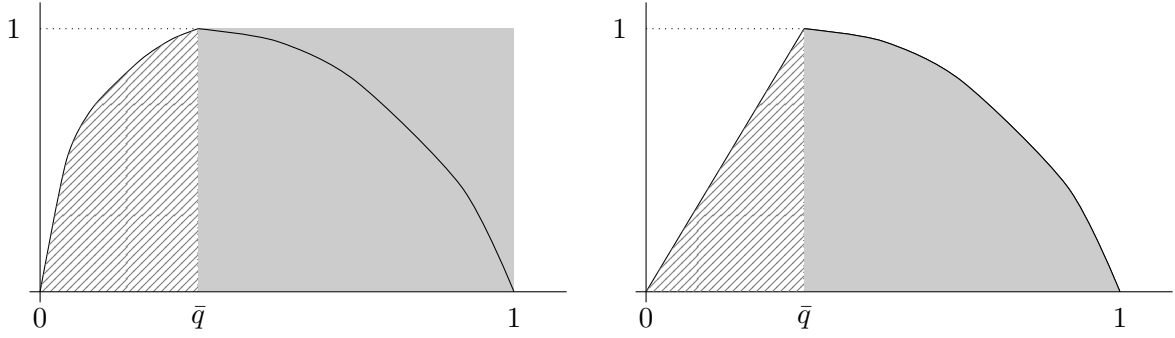
\begin{figure}[t]
\begin{flushleft}
\begin{minipage}[t]{0.48\textwidth}
\centering
\iffocs
\newcommand{\REVONESCALE}{0.35}
\else
\newcommand{\REVONESCALE}{0.7}
\fi

\begin{tikzpicture}[scale = \REVONESCALE]

\draw [name path global = A] plot [smooth, tension=0.45] coordinates {(0, 0) (0.5, 2.5) (1,3.5) (1.9, 4.4) (2.5, 4.8) (3, 5)};

\draw [dotted] (0, 5) -- (3, 5);
\draw [name path global = B, dotted] (3, 0) -- (3, 5);

\tikzfillbetween[of=A and B]{pattern=north east lines, pattern color=gray};

\fill[color=gray!40!white] (3, 0.01) -- (3, 5) -- (9, 5) -- (9, 0.01);

\draw plot [smooth, tension=0.5] coordinates {(3, 5) (4.5, 4.76) (6, 4) (8, 2) (9, 0)};

\draw (-0.5, 5) node {$1$};

\iffocs
\draw (0, -0.7) node {$0$};
\draw (9, -0.7) node {$1$};
\draw (3, -0.7) node {$\qo$};
\else
\draw (0, -0.5) node {$0$};
\draw (9, -0.5) node {$1$};
\draw (3, -0.5) node {$\qo$};
\fi

\draw (-0.2,0) -- (10, 0);
\draw (0, -0.2) -- (0, 5.5);

\end{tikzpicture}
\end{minipage}
\begin{minipage}[t]{0.48\textwidth}
\centering
\iffocs
\newcommand{\REVTWOSCALE}{0.35}
\else
\newcommand{\REVTWOSCALE}{0.7}
\fi

\begin{tikzpicture}[scale = \REVTWOSCALE]

\draw (-0.2,0) -- (10, 0);
\draw (0, -0.2) -- (0, 5.5);

\draw (0, 0) -- (3, 5);

\draw [dotted] (0, 5) -- (3, 5);
\draw [dotted] (3, 0) -- (3, 5);

\draw [name path global = B] (3, 0.01) -- (9, 0.01);

\draw [name path global = A] plot [smooth, tension=0.5] coordinates {(3, 5) (4.5, 4.76) (6, 4) (8, 2) (9, 0)};

\fill[pattern=north east lines, pattern color=gray] (0, 0) -- (3, 5) -- (3, 0);

\tikzfillbetween[of=A and B]{gray!40!white};

\draw plot [smooth, tension=0.5] coordinates {(3, 5) (4.5, 4.76) (6, 4) (8, 2) (9, 0)};

\draw (-0.5, 5) node {$1$};

\iffocs
\draw (0, -0.7) node {$0$};
\draw (9, -0.7) node {$1$};
\draw (3, -0.7) node {$\qo$};
\else
\draw (0, -0.5) node {$0$};
\draw (9, -0.5) node {$1$};
\draw (3, -0.5) node {$\qo$};
\fi

\end{tikzpicture}
\end{minipage}
\vspace{-18pt}
\end{flushleft}
\caption{\label{f:rev-decompose} The illustration of the revenue
  decomposition of \Cref{lem:truncate} for $\mecha$ on distribution
  $\dist$ and truncation $\dist'$ for the optimal mechanism and
  second-price auction.  The thin black line on the left and right
  figures are the revenue curves corresponding to $\dist$ and
  $\dist'$, respectively.  The dashed area on the left represents
  $\OPT_+ = \SPA_+$ and the gray area on the left represents $\OPT_- =
  \OPT'_-$.  The dashed area on the right represents $\OPT'_+ = \SPA'_+$
  and the
  gray area on the right represents $\SPA'_- = \SPA_-$.}
\end{figure}

\begin{proof}
  It can be assumed that the approximation of stochastic markup
  mechanism $\mecha$ on distribution $\dist$ is at least $\sfrac 1
  {\secprob}$ (where $\secprob$ denotes the probability that $\mecha$
  places on the second-price auction).  Notice that the revenue
  $\mecha$ on the point mass on 1 (a truncated distribution) is
  $\secprob$ and the optimal revenue on this distribution is 1.  If
  the approximation factor $\sfrac{\optf{\dist}}{\mecha(\dist)}$ is
  less than $\sfrac{1}{\secprob}$ then the point mass on 1 (a truncated
  distribution) achieves a higher approximation than $\dist$ and the
  lemma follows.  For the remainder of the proof, assume that the
  approximation factor of mechanism $\mecha$ on distribution $\dist$
  is more than $\sfrac{1}{\secprob}$.

  View the stochastic markup mechanism $\mecha$ as a distribution over two
  mechanisms: the second-price auction $\mech_1$ with probability
  $\alpha$, and $\mech_{*}$, a distribution over non-trivial markup
  mechanisms $\mech_{\ratio}$ with $\ratio > 1$, with probability
  $1-\alpha$.  The optimal mechanism is $\OPT_{\dist}$.  Decompose the
  revenue from distribution $\dist$ across these three mechanisms as
  follows.  Denote the monopoly quantile of $\dist$ by $\qo$.
  See \Cref{f:rev-decompose}.

  \begin{itemize}
  \item $\OPT_+$ and $\OPT_-$ give the expected revenue of the
    optimal mechanism from agents with values above and below the
    monopoly price (below and above the monopoly quantile $\qo$).
    
  \item $\SPA_+ = \OPT_+$ and $\SPA_-$ give the expected revenue of
    the second-price auction $\mecha_1$ from agents with values above
    and below the monopoly price.
    
  \item $\MKUP_+$ and $\MKUP_-$ give the expected revenue of the
    stochastic markup mechanism $\mecha_*$ from prices (strictly)
    above and (weakly) below the monopoly price.
  \end{itemize}
Consider truncating the distribution $\dist$ at the monopoly quantile
$\qo$ to obtain $\dist' \in \TRUNCF$.  Define analogous quantities (with identities):
\begin{itemize}
\item $\OPT'_+ < \OPT_+$ and $\OPT'_- = \OPT_-$.

  Identities follow from the geometric analysis of \Cref{l:DRY-15}.
  
\item $\SPA'_+ = \OPT'_+$ and $\SPA'_- = \SPA_-$.

  Identities follow from the geometric analysis of \Cref{l:DRY-15}.
  
\item $\MKUP'_+ = 0$ and $\MKUP'_- = \MKUP_-$.

  Values above the monopoly price are not supported by the truncated
  distribution, so the revenue from those prices is zero.  On the
  other hand, prices (weakly) below the monopoly price are bought with
  the exact same probability as the cumulative distribution function
  $\dist'$ and $\dist$ are the same for these prices.
 
\end{itemize}
The remainder of the proof follows a straightforward calculation.
Write the approximation ratio of $\mecha$ on distribution $\dist$
(using the given identities) and rearrange:
\begin{align*}
\frac{\optf{\dist}}{\mecha(\dist)}
&= \frac{\OPT_+ + \OPT_-}{\secprob\, (\OPT_+ + \SPA_-) + (1-\secprob)\,(\MKUP_+ + \MKUP_-)}\\
&= \frac{\OPT_+ + \left[\OPT_-\right]}{\secprob\, \OPT_+  + \left[\secprob \SPA_- + (1-\secprob)\,(\MKUP_+ + \MKUP_-)\right]}\\
\intertext{Since the approximation ratio on $\dist$ is at least $\sfrac 1 \secprob$, the ratio of the first term in the numerator and denominator is at most the ratio of the remaining terms [in brackets]:}
\frac{1}{\secprob} &= \frac{\OPT_+}{\secprob\,\OPT_+} \leq \frac{\left[\OPT_-\right]}{\left[\secprob \SPA_- + (1-\secprob)\,(\MKUP_+ + \MKUP_-)\right]}\\
\intertext{Now write the approximation ratio of $\mecha$ on truncation $\dist'$ (using the given identities) and bound:}
\frac{\optf{\dist'}}{\mecha(\dist')}
&= \frac{\OPT'_+ + \left[\OPT_-\right]}{\secprob\, \OPT'_+ + \left[\secprob\, \SPA_- + (1-\secprob)\,\MKUP_-\right]}\\
&\geq \frac{\OPT'_+ + \left[\OPT_-\right]}{\secprob\, \OPT'_+ + \left[\secprob\, \SPA_- + (1-\secprob)\,(\MKUP_+ + \MKUP_-)\right]}\\
&\geq \frac{\OPT_+ + \left[\OPT_-\right]}{\secprob\, \OPT_+ + \left[\secprob\, \SPA_- + (1-\secprob)\,(\MKUP_+ + \MKUP_-)\right]}\\
&= \frac{\optf{\dist}}{\mecha(\dist)}.
\end{align*}
The calculation shows that, for any distribution $\dist$, the
truncated distribution $\dist'$ increases the approximation factor of
the stochastic markup mechanism.  Thus, the worst-case distribution is
truncated.
\end{proof}

\begin{figure}[t]
\begin{flushleft}
\begin{minipage}[t]{0.48\textwidth}
\centering
\iffocs
\newcommand{\QONESCALE}{0.35}
\else
\newcommand{\QONESCALE}{0.7}
\fi

\begin{tikzpicture}[scale = \QONESCALE]

\draw (-0.2,0) -- (10, 0);
\draw (0, -0.2) -- (0, 5.5);

\draw (0, 0) -- (3, 5);
\draw plot [smooth, tension=0.6] coordinates {(6, 4) (7, 3) (9, 0)};

\draw plot [smooth, tension=0.6] coordinates {(3, 5) (4.5, 4.7) (6, 4)};

% \draw (3, 5) -- (6, 4);
% \draw (6, 4) -- (9, 0);

\draw [dotted] (0, 5) -- (3, 5);
\draw [dotted] (0, 0) -- (7.2, 4.8);
\draw [dotted] (3, 0) -- (3, 5);

\draw [dotted] (6, 0) -- (6, 4);

\begin{scope}[very thick]
\draw [dashed, gray] (0, 0) -- (3, 5);
\draw [dashed, gray] (3, 5) -- (6, 4);
\draw [dashed, gray] plot [smooth, tension=0.6] coordinates {(6, 4) (7, 3) (9, 0)};
% \draw [dashed, gray] (6, 4) -- (9, 0);
\end{scope}

% \draw (-0.2, 4.8) -- (0.2, 4.8);
\draw (-0.5, 5) node {$1$};

\iffocs
\draw (0, -0.7) node {$0$};
\draw (9, -0.7) node {$1$};
\draw (3, -0.7) node {$\qo$};
\draw (6, -0.7) node {$\qt$};
\else
\draw (0, -0.5) node {$0$};
\draw (9, -0.5) node {$1$};
\draw (3, -0.5) node {$\qo$};
\draw (6, -0.5) node {$\qt$};
\fi
\draw (6.5, 5) node {$\sfrac{1}{r\qo}$};

\end{tikzpicture}
\end{minipage}
\begin{minipage}[t]{0.48\textwidth}
\centering
\iffocs
\newcommand{\QTWOSCALE}{0.35}
\else
\newcommand{\QTWOSCALE}{0.7}
\fi

\begin{tikzpicture}[scale = \QTWOSCALE]

\draw (-0.2,0) -- (10, 0);
\draw (0, -0.2) -- (0, 5.5);

\draw (0, 0) -- (3, 5);
\draw plot [smooth, tension=0.6] coordinates {(6, 4) (7, 3) (9, 0)};

\draw (3, 5) -- (6, 4);
% \draw (6, 4) -- (9, 0);

\draw [dotted] (0, 5) -- (3, 5);
\draw [dotted] (0, 0) -- (7.2, 4.8);
\draw [dotted] (3, 0) -- (3, 5);

\draw [dotted] (6, 0) -- (6, 4);

\begin{scope}[very thick]
\draw [dashed, gray] (0, 0) -- (3, 5);
\draw [dashed, gray] (3, 5) -- (6, 4);
% \draw [dashed, gray] plot [smooth, tension=0.6] coordinates {(6, 4) (7, 3) (9, 0)};
\draw [dashed, gray] (6, 4) -- (9, 0);
\end{scope}

% \draw (-0.2, 4.8) -- (0.2, 4.8);
\draw (-0.5, 5) node {$1$};

\iffocs
\draw (0, -0.7) node {$0$};
\draw (9, -0.7) node {$1$};
\draw (3, -0.7) node {$\qo$};
\draw (6, -0.7) node {$\qt$};
\else
\draw (0, -0.5) node {$0$};
\draw (9, -0.5) node {$1$};
\draw (3, -0.5) node {$\qo$};
\draw (6, -0.5) node {$\qt$};
\fi
\draw (6.5, 5) node {$\sfrac{1}{r\qo}$};

\end{tikzpicture}
\end{minipage}
\vspace{-18pt}
\end{flushleft}
\caption{\label{f:proof-quad} The main two steps of \Cref{lem:quad}
  are illustrated.  In the first step (right-hand side), the revenue
  curves of distributions $\distTrunc$ (thin, solid, black) and
  $\dist\primed$ (thick, dashed, gray) are depicted.  In the second
  step, the revenue curves of the distributions $\dist\primed$ (thin,
  solid, black) and $\distQr$ (thick, dashed, gray) are depicted.  In
  both cases the revenue of the $\ratio$-markup mechanism is is higher
  on the thin, solid, black curve than the thick, dashed, gray curve.}
\end{figure}
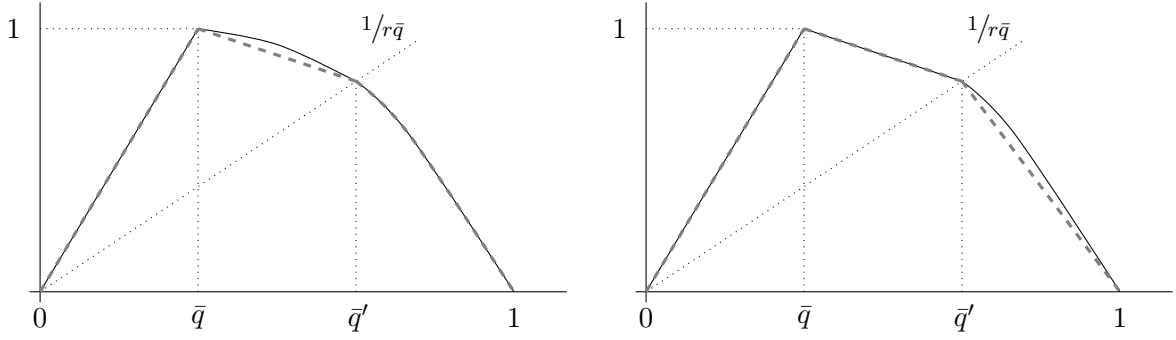

The next step is to show that, among truncated distributions, the
worst-case distribution for stochastic markup mechanisms are those
with quadrilateral-shaped revenue curves, i.e., ones that are
piecewise linear with three pieces (see \Cref{f:quad}).  Recall that
for a truncated distribution at monopoly quantile $\qo$, the upper
bound of the support is a point mass on $\sfrac 1 \qo$.

% calculation of case 2 for the following definition for Quadrilaterals
%\begin{align*}
% \text{height of Quad at}~\qt: & \quad \qt/(r\qo)\\
% \text{let quantile function output be}~Q &\\
% Qv &= \frac{\qt}{r\qo}+\frac{\qt-Q}{\qt-\qo}\cdot\left(1-\frac{\qt}{r\qo} \right)\\
% Qvr\qo(\qt-\qo)&= \qt(\qt-\qo)+(\qt-Q)(r\qo-\qt)\\
% Q &= \frac{\qt\qo(r-1)}{vr\qo(\qt-\qo)+(r\qo-\qt)}
%\end{align*}

\begin{definition}
\label{def:quadrilateral}
A \emph{normalized quadrilateral distribution} with parameters
$\qo, \qt$ and $\ratio$ with $\ratio \geq 1$ and $\frac{\qo r}{\qo r + (1-\qo)}
\leq \qt \leq \min\{r \qo,1\}$, denoted by $\Qr_{\qo,
  \qt, \ratio}$ is defined by quantile function as:
\begin{equation*}
\quantf_{\Qr_{\qo, \qt, \ratio}}(\val) = 
\begin{cases}
\frac{\qt}{\qt + \val \ratio\qo(1-\qt)} & \val < \sfrac{1}{r\qo}\\
%\frac{\sfrac{\qt}{r} 
%- \sfrac{\qt(\qo - \sfrac{\qt}{\ratio})}
%{(\qo - \qt)}}
%{\val\qo - \sfrac{(\qo - \sfrac{\qt}{\ratio})}
%{(\qo - \qt)}} 
\frac{\qt\qo(r-1)}{vr\qo(\qt-\qo)+(r\qo-\qt)}
& \sfrac{1}{r\qo} \leq \val \leq \sfrac{1}{\qo}\\
0 & \sfrac{1}{\qo} < \val
\end{cases}
\end{equation*}
\end{definition}

\noindent The following lemma summarizes an analysis from \citet{AB-18} and is
useful in bounding the revenue from markup mechanisms.

\begin{lemma}[\citealp{AB-18}]
  \label{l:AB-18}
Consider the $\ratio$-markup mechanism, two i.i.d.\ regular agents
with value distribution $\dist$, quantile $\qt$ corresponding to the
monopoly price divided by $\ratio$, and the distribution
$\Tilde{\dist}$ that corresponds to $\dist$ ironed on $[\qt,1]$: the
virtual surplus from quantiles $[\qt,1]$ is higher for $\dist$ than
for $\Tilde{\dist}$.
\end{lemma}

\begin{proof}
  The proof of this lemma is technical and non-trivial.  It is given in 
  the proof of Proposition~4 of \citet{AB-18}.
\end{proof}

The next lemma reduces the worst case distribution from the family of truncated distributions to the family of quadrilateral distributions.   The reduction is illustrated in \Cref{f:proof-quad}, by showing that ironing the revenue curves sequentially within $[\qo,\qt]$ and $[\qt, 1]$ decreases the revenue of the stochastic markup mechanism.  The optimal revenue is not affected because it is obtained using a reserve price corresponding to the monopoly quantile $\qo$ and it is agnostic to the shape of the revenue curve for $\quant> \qo$.

\begin{lemma}\label{lem:quad}
    For i.i.d., two-agent, single-item environments, any truncated
    distribution $\distTrunc$, and any stochastic markup mechanism
    $\mechaAR$ with probability $\secprob$ on the second-price auction
    $\mech_1$ and probability $1-\secprob$ on non-trivial markup
    mechanism $\mech_r$; there is a quadrilateral distribution
    $\distQr$ with the same optimal revenue and (weakly) lower revenue
    in $\mechaAR$.  I.e., $\optf{\distQr} = \optf{\distTrunc}$ and
    $\mechaAR(\distQr) \leq \mechaAR(\distTrunc)$.
\end{lemma}

\begin{proof}
  On any normalized truncated distribution with monopoly quantile
  $\qo$, the optimal revenue is $2-\qo$ (\Cref{lem:opt rev}).  Thus,
  to prove the lemma it is sufficient to show that for any truncated
  distribution $\distTrunc \in \TRUNCF$ with monopoly quantile $\qo$
  there is a normalized quadrilateral distribution $\distQr \in \QRF
  \subset \TRUNCF$ with monopoly quantile $\qo$ and lower revenue in
  $\mechaAR$.

  The quadrilateral distribution $\distQr$ is obtained by ironing
  $\distTrunc$ on $[\qo,\qt]$ and $[\qt,1]$ where quantile $\qt$
  satisfies $\valf_{\distTrunc}(\qo) =
  \ratio\,\valf_{\distTrunc}(\qt)$.  We consider an intermediary
  distribution $\dist\primed$ that is $\distTrunc$ ironed only on
  $[\qo,\qt]$.  See \Cref{f:proof-quad}.  The proof approach is to
  show that $\mechaAR(\distTrunc) > \mechaAR(\dist\primed) >
  \mechaAR(\distQr)$.

  As $\mechaAR$ is a convex combination of the second-price auction
  $\mecha_1$ and the $\ratio$-markup mechanism $\mecha_{\ratio}$.  It
  suffices to show the inequalities above hold for both auctions.  In fact,
  the result holds for the second-price auction from the geometric
  analysis of revenue of \Cref{l:DRY-15}.  The revenue of the
  second-price auction for two i.i.d.\ agents is twice the area under
  the revenue curve.  As the revenue curve has strictly smaller area
  from $\distTrunc$ to $\dist\primed$ to $\distQr$, we have
  $\mecha_1(\distTrunc) > \mecha_1(\dist\primed) > \mecha_1(\distQr)$.
  Below, we analyze the $\ratio$-markup mechanism $\mecha_{\ratio}$.

  The following price-based analysis shows that
  $\mecha_{\ratio}(\distTrunc) > \mech_{\ratio}(\dist\primed)$:
  \begin{itemize}
  \item The revenue from quantiles in $[0,\qo]$ is unchanged.
    
    These quantiles are offered prices from quantiles in $[\qt,1]$.
    The values of quantiles $[0,\qo]$ and $[\qt,1]$ are the same for
    both distributions; thus, the revenue is unchanged.
    
  \item The revenue from quantiles in $[\qo,\qt]$ decreases.

    These quantiles are offered prices from quantiles in $[\qt,1]$.
    For the distribution $\dist\primed$ relative to $\distTrunc$:
    Values are lower at any quantile $\quant \in [\qo,\qt]$; the
    distribution of prices (from quantiles in $[\qt,1]$) is the same.
    Thus, revenue is lower.

  \item The revenue from quantiles in $[\qt, 1]$ is unchanged.

    These quantiles are in $[\qt,1]$ and are offered prices from
    quantiles in $[\qt,1]$.  The distributions are the same for these
    quantiles; thus, the revenue is unchanged.
  \end{itemize}
  
  %% This argument is given formally as follows:
  %% \begin{align*}
  %%   \mecha_{\ratio}(\distTrunc) 
  %%   &= \int_0^1 r \cdot \valf_{\distTrunc}(\quant) 
  %%   \cdot \quantf_{\distTrunc}(r \cdot \valf_{\distTrunc}(\quant))\, dq \\
  %%   &= \int_{\qt}^1 r \cdot \valf_{\distTrunc}(\quant) 
  %%   \cdot \quantf_{\distTrunc}(r \cdot \valf_{\distTrunc}(\quant))\, dq \\
  %%   &\geq \int_{\qt}^1 r \cdot \valf_{\dist\primed}(\quant) 
  %%   \cdot \quantf_{\dist\primed}(r \cdot \valf_{\dist\primed}(\quant))\, dq \\
  %%   &= \mecha_{\ratio}(\dist\primed), 
  %% \end{align*}
  %% where the second equality holds because $\quantf_{\distTrunc}(r \cdot
  %% \valf_{\distTrunc}(\quant)) = 0$ for any quantile $\quant < \qt$.  The
  %% inequality holds because $\valf_{\distTrunc}(\quant) =
  %% \valf_{\dist\primed}(\quant)$ for $\quant \in [\qt, 1]$ and
  %% $\distTrunc(\val) \leq \dist\primed(\val)$ for any value $\val$.

  The following virtual-surplus-based analysis shows that
  $\mecha_{\ratio}(\dist\primed) > \mech_{\ratio}(\distQr)$:
  \begin{itemize}
  \item The virtual surplus of quantiles in $[0,\qo]$ is unchanged.
    
    These quantiles have the same virtual values under the two
    distributions and the same probability of winning, i.e., $1-\qt$
    (when the other agent's quantile is in $[\qt,1]$.
    
  \item The virtual surplus of quantiles in $[\qo,\qt]$ is decreased.

    Their prices come from quantiles in $[\qt,1]$ which are decreased;
    thus, their probabilities of winning are increased.  Their virtual
    values are negative, so these increased probabilities of winning
    result in decreased virtual surplus.

  \item The virtual surplus of quantiles in $[\qt,1]$ is decreased.

    This result is given by \Cref{l:AB-18}.  \qedhere
  \end{itemize}

\end{proof}

\begin{figure}[t]
\centering
\begin{tikzpicture}[scale = 0.7]

\draw (-0.2,0) -- (10, 0);
\draw (0, -0.2) -- (0, 5.5);

\draw (0, 0) -- (3, 5);
\draw (3, 5) -- (6, 4);
\draw (6, 4) -- (9, 0);

\draw [dotted] (0, 5) -- (3, 5);
\draw [dotted] (0, 0) -- (7.2, 4.8);
\draw [dotted] (3, 0) -- (3, 5);

\draw [dotted] (6, 0) -- (6, 4);
\draw [dotted] (5.4, 0) -- (5.4, 3.6);

\begin{scope}[very thick]
\draw [dashed, gray] (0, 0) -- (3, 5);
\draw [dashed, gray] (3, 5) -- (5.4, 3.6);
\draw [dashed, gray] (5.4, 3.6) -- (9, 0);
\end{scope}

\draw (0, -0.5) node {$0$};
% \draw (-0.2, 4.8) -- (0.2, 4.8);
\draw (-0.5, 5) node {$1$};

\draw (9, -0.5) node {$1$};
\draw (3, -0.5) node {$\qo$};
\draw (6.1, -0.5) node {$\qt$};
\draw (5.3, -0.5) node {$\qt'$};

\draw (8.3, 1.9) node {$\rev_{\qt}$};
\draw (7.5, 0.65) node {$\rev_{\qt'}$};

\draw (6.5, 5) node {$\sfrac{1}{r\qo}$};

\fill[color=gray!40!white] (3, 5) -- (5.4, 3.6) -- (9, 0) -- (6, 4);

\end{tikzpicture}
\caption{\label{f:proof-tri} Illustrating the proof of \Cref{lem:tri}, the difference of revenue for second
  price auction $\mecha_1$ on revenue curves $\rev_{\qt}$ and
  $\rev_{\qt'}$, which respectively correspond to quadrilateral
  distributions $\Qr_{\qo, \qt, \ratio}$ and $\Qr_{\qo, \qt',
    \ratio}$, is equal to twice of the gray area, which is at least
  $\qt' - \qt$.  Moreover, the difference of revenue for the
  $\ratio$-markup mechanism $\mecha_{r}$ on revenue curves
  $\rev_{\qt}$ and $\rev_{\qt'}$ is at most $2(\qt' - \qt)$.  }
\end{figure}
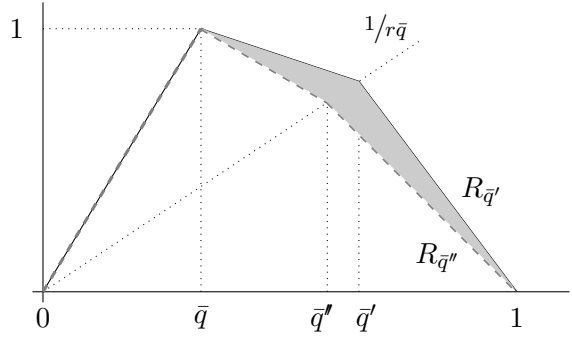

%We complete the proof of \Cref{thm:tri} by showing that triangle distributions lead to lower revenue than quadrilateral distributions.
We complete the proof of \Cref{thm:tri} by showing that triangle distributions lead to lower revenue than quadrilateral distributions. The intuition of the proof is illustrated in \Cref{f:proof-tri}.  For any $\ratio > 1$ and any stochastic markup mechanism $\mechaAR$ with probability $\secprob \in [\sfrac{2}{3},1]$,  consider a family of quadrilateral distributions $\Qr_{\qo, \qt, \ratio}$ parameterized by $\qt$.  The optimal revenue is again not affected by $\qt$ while the revenue of $\mechaAR$ is monotone increasing in $\qt$. Thus the approximation ratio of $\mechaAR$ is maximized by minimal $\qt$ for which the degenerate quadrilateral $\Qr_{\qo, \qt, \ratio}$ is a triangle.

\begin{lemma}\label{lem:tri}
  For i.i.d., two-agent, single-item environments, normalized
  quadrilateral distribution $\Qr_{\qo, \qt, \ratio}$, and stochastic
  markup mechanism $\mechaAR$ with probability $\secprob \in [\sfrac 2
    3,1]$ on the second-price auction $\mech_1$ and probability
  $1-\secprob$ on non-trivial markup mechanism $\mech_r$; the triangle
  distribution $\tri_{\qo}$ has the same optimal revenue and (weakly)
  lower revenue in $\mechaAR$.  I.e., $\optf{\tri_{\qo}} =
  \optf{\Qr_{\qo, \qt, \ratio}}$ and $\mechaAR(\tri_{\qo}) \leq
  \mechaAR(\Qr_{\qo, \qt, \ratio})$.
\end{lemma}

\begin{proof}
  By \Cref{lem:opt rev}, the optimal revenues for quadrilateral
  distribution $\Qr_{\qo, \qt, \ratio}$ and triangle distribution
  $\tri_{\qo}$ are the same (and equal to $2-\qo$).  To show that the
  revenue of $\mechaAR$ is worse on $\tri_{\qo}$ than $\Qr_{\qo, \qt,
    \ratio}$, it suffices to show that the revenue on $\Qr_{\qo, \qt,
    \ratio}$ is monotonically increasing in $\qt$.  Specifically the
  minimum revenue is when the quadrilateral distribution is
  degenerately equal to the triangle distribution.

The proof strategy is to lower bound the partial derivative with respect to
$\qt$ of the revenues of $\ratio$-markup mechanism and the
second-price auction for quadrilateral distributions $\Qr_{\qo, \qt, \ratio}$ as
\begin{align}
  \label{eq:r-markup-partial}
  \frac{\partial \mecha_\ratio(\Qr_{\qo, \qt, \ratio})}{\partial \qt}
  &\geq -2, \\
  \label{eq:spa-partial}
  \frac{\partial \mecha_1(\Qr_{\qo, \qt, \ratio})}{\partial \qt}
  &\geq 1. 
\intertext{Thus, for mechanism $\mechaAR$ with $\secprob \geq \sfrac{2}{3}$, 
  we have}
\notag
\frac{\partial \mechaAR(\Qr_{\qo, \qt, \ratio})}{\partial \qt}
&\geq \secprob - 2(1-\secprob) \geq 0
\end{align}
and revenue is minimized with the smallest choice of $\qt$ for which
quadrilateral distribution $\Qr_{\qo, \qt, \ratio}$ is degenerately a
triangle distribution.  It remains to prove bounds
\eqref{eq:r-markup-partial} and \eqref{eq:spa-partial}.

% The optimal revenue for distribution $\Qr_{\qo, \qt, \ratio}$
% is $2 - \qo$,
For simplicity, 
since in this section the only parameter we change in 
distribution $\Qr_{\qo, \qt, \ratio}$ is $\qt$, 
we introduce the notation $\revq(\val)$
to denote the revenue of posting price $\val$,
and $\valq(\quant)$ to denote the price $\val$ given quantile $\quant$ when the distribution is $\Qr_{\qo, \qt, \ratio}$.  The proof is illustrated in \Cref{f:proof-tri}.

We now prove bound \eqref{eq:r-markup-partial}.
For any pair of quadrilateral distributions 
$\Qr_{\qo, \qt, \ratio}$ and $\Qr_{\qo, \qt', \ratio}$ with $\qt' \geq \qt$, 
we analyze the difference in revenue for posting price $\ratio\cdot \vsec$. 
% Explanations for each step are given below.
\begin{align*}
\lefteqn{\mecha_{\ratio}(\Qr_{\qo, \qt', \ratio})
- \mecha_{\ratio}(\Qr_{\qo, \qt, \ratio})} \qquad\\
&= 2 \int_{\qt'}^1 \revqt(\ratio \cdot \valqt(\quant)) \, dq
- 2 \int_{\qt}^1 \revq(\ratio \cdot \valq(\quant)) \, dq \\
&\geq 2 \int_{\qt'}^1 \revqt(\ratio \cdot \valqt(\quant)) \, dq
- 2 \int_{\qt'}^1 \revq(\ratio \cdot \valq(\quant))\, dq 
- 2(\qt' - \qt)\\
&\geq 2 \int_{\qt'}^1 \revq(\ratio \cdot \valqt(\quant)) \, dq
- 2 \int_{\qt'}^1 \revq(\ratio \cdot \valq(\quant))\, dq 
- 2(\qt' - \qt)\\
% \geq\,& 2 \int_{\qt'}^1 \revq(\ratio \cdot \valq(\quant)) \, dq
% - 2 \int_{\qt'}^1 \revq(\ratio \cdot \valq(\quant))\, dq 
% - 2(\qt' - \qt)\\
&\geq - 2(\qt' - \qt).
\end{align*}
The first equality is constructed as follows: Both agents face a
random price that is $r$ times the value of the other agent who has
quantile $\quant$ drawn from $U[0,1]$.  The revenue from this price is
given by, e.g., $\revq(\ratio \cdot \valq(\quant))$ which is 0 when
$\quant \leq \qt$.  The first inequality holds because $\revq(\ratio
\cdot \valq(\quant)) \leq 1$ for any quantile $\quant$.  The second
inequality holds since the revenue from revenue curve $\revqt$ is
weakly higher than from revenue curve $\revq$ for any value $\val$.
The third inequality holds because (a) the prices of the first
integral are higher than the prices of the second integral, i.e., $\valqt(\quant) \geq \valq(\quant)$ for
every $\quant$, and (b) because
these prices are below the monopoly price for distribution $\Qr_{\qo,
  \qt', \ratio}$ and so higher prices give higher revenue.

Therefore, we have
\begin{align*}
\frac{\partial \mecha_{\ratio}(\Qr_{\qo, \qt, \ratio})}{\partial \qt}
= \lim_{\qt' \to \qt} \frac{\mecha_{\ratio}(\Qr_{\qo, \qt', \ratio}) - \mecha_{\ratio}(\Qr_{\qo, \qt, \ratio})}{\qt' - \qt}
\geq -2. 
\end{align*}
We now prove bound \eqref{eq:spa-partial}.
The revenue of the second price auction for two i.i.d.\ agents is twice the area under the revenue curve (\Cref{l:DRY-15}).  For 
quadrilateral distribution $\Qr_{\qo, \qt, \ratio}$ this revenue is calculated as:
\begin{align*}
\mecha_1(\Qr_{\qo, \qt, \ratio})
&= 2 \int_0^1 \rev_{\qt}(q) \, dq\\
&= 2 \int_0^{\qo} \rev_{\qt}(q) \, dq + 2 \int_{\qo}^{\qt} \rev_{\qt}(q) \, dq + 2 \int_{\qt}^1 \rev_{\qt}(q) \, dq\\
&= \qo + (\qt - \qo)(1+\frac{\qt}{\ratio\cdot \qo}) + (1-\qt)\frac{\qt}{\ratio\cdot \qo} \\
&= \qt + (1 - \qo)\frac{\qt}{\ratio\cdot \qo}. 
\end{align*}
% where the third line evaluates each of the integrals as the width of the interval times the average height of the revenue curve on the interval, as each is a line segment. 
Therefore, we have 
\begin{align*}
\frac{\partial \mecha_1(\Qr_{\qo, \qt, \ratio})}{\partial \qt}
& = 1 + \frac{1 - \qo}{\ratio\cdot \qo} \geq 1. \qedhere
\end{align*}
\end{proof}

\section{Sub-optimality of Relaxed Benchmark Design}
\label{sec:gap}

In this section, we will show that the relaxed benchmark design
program \eqref{eq:heuristic} is not generally equal to the benchmark
design program \eqref{eq:resolution} by considering the revenue
maximization problem for two agents with i.i.d.\ regular
distributions.  Since benchmark optimization and prior-independent
mechanism design are equivalent problems (see
\Cref{sec:prior-independent}, \Cref{thm:equiv}), a gap between the
objective values $\heuristic$ and $\resolution$ of
programs~\eqref{eq:heuristic} and~\eqref{eq:resolution} is implied by
exhibiting a benchmark with lower objective value in
program~\eqref{eq:heuristic} than the approximation achieved by the
optimal prior-independent mechanism, i.e., the solution to
program~\eqref{eq:pi}.

We first establish a gap in the value of these programs under the
restriction to triangle distributions.
\begin{theorem}\label{thm:gap triangle}
  For i.i.d.\ triangle distributions, two-agent, single-item environments and
  scale-invariant, incentive-compatible mechanisms the heuristic
  benchmark optimization program~\eqref{eq:heuristic} has a strictly
  smaller objective value than the benchmark optimization
  program~\eqref{eq:resolution}, i.e., $\heuristic < \resolution$.
\end{theorem}

This section follows the notations and definitions introduced in
\Cref{sec:prior-independent}.  Let $\mechaARO$ be the approximately
optimal prior-independent mechanism with numerical parameters
$\optweight = \aval, \optratio = \rval$ (\Cref{thm:pi optimal
  mechanism}).  We first define the pseudo-mechanism $\pseudomech$ as
the affine combination of mechanisms with revenue $\pseudomech(\vals)
= (1+\delta) \mechaARO(\vals) - \delta\, \mecha_{1.1}(\vals)$, where
$\delta > 0$ is a small constant.  Recall that $\TRIF$ is the family
of triangle distributions.  We show that
\begin{enumerate}
\item $\forall \dist \in \TRIF,\ \pseudomech(\dist) \leq \optf{\dist}$;

\item $\forall \dist \in \TRIF,\ \piratio'\,\pseudomech(\dist) \geq
  \optf{\dist}$, where $\piratio' < \piratio$.
\end{enumerate}
Letting $\gapbm(\vals) = \piratio' \, \pseudomech(\vals)$,
benchmark $\gapbm$ is normalized for all triangle distributions and
the approximation of optimal payoff against benchmark $\gapbm$ is
$\piratio' < \piratio = \resolution$.  Therefore, the approximation of
the optimal heuristic benchmark is $\heuristic \leq \piratio' <
\resolution$ and \Cref{thm:gap triangle} holds.

We first prove that the first statement holds by showing a more
general lemma for regular distributions with proof deferred at the end
of the section.
\begin{lemma}\label{lem:upper bound benchmark}
% For revenue maximization with 2 agents, 
For i.i.d., two-agent, regular, single-item environments,
any constants $1\leq \ratio_1 \leq \ratio_2 \leq \ratio_3$, 
and any constants $a, b, c \geq 0$ such that 
$a+b-c \leq 1$, $a\geq c, b\geq c$: 
\begin{align*}
\pseudomech(\dist) = a \, \mecha_{\ratio_1}(\dist)
+ b \, \mecha_{\ratio_3}(\dist)
- c \, \mecha_{\ratio_2}(\dist)
\leq \optf{\dist}.
\end{align*}
\end{lemma}

Our approach to show that the second statement holds is as follows.
Recall that the benchmark program \eqref{eq:resolution} has value
$\apxval$.  It will suffice then to show that on triangle
distributions defined by monopoly quantile $\qo$ close to $\optqo =
\qval$, where the prior-independent optimal mechanism $\mechaARO$
obtains revenue at least $\apxval$, that $\mech_{1.1}$ performs very
badly (specifically, it is worse than a 2-approximation).  Moreover,
on distributions with monopoly quantile $\qo$ far from $\optqo$, mechanism
$\mechaARO$ does much better than it does on worst case distributions
(specifically, it is better than a 1.9041 approximation).  The theorem
will follow by combining these observations.

The following lemmas establish the two main steps of the argument.
These lemmas are based on our formula for the revenue
$\mech_{\ratio}(\tri_{\qo})$ as given by \Cref{clm:rev of r}.  We
evaluate the formula for monopoly quantiles $\qo$ in a grid and then
use the smoothness of the revenue formula that is established in
\Cref{app:prior-independent} to show that the revenue for all monopoly
quantiles in the desired range satisfy the desired bound.  The formal
proof of these lemmas is given in \Cref{app:gap}.

\begin{lemma}\label{lem:lower bound ratio for m1}
For i.i.d., two-agent, single-item environments, and any quantile $\qo
\in [0.05, 0.25]$, the approximation of mechanism
$\mech_{1.1}$ on triangle distribution $\tri_{\qo}$ is at least 2.
\end{lemma}

\begin{lemma}\label{lem:upper bound ratio for opt mech}
  For i.i.d., two-agent, single-item environments, and any quantile
  $\qo \not\in [0.05, 0.25]$, the approximation of mechanism
  $\mechaARO$ on triangle distribution $\tri_{\qo}$ is at most 1.9041.
\end{lemma}

To complete the proof of \Cref{thm:gap triangle}, we combine these
lemmas as follows.  Setting $\delta = 0.00154$, for $\qo \in [0.05,
  0.25]$, by \Cref{lem:lower bound ratio for m1} and \Cref{thm:pi
  optimal mechanism}, we have
\begin{align*}
\frac{\OPT_{\tri_{\qo}}(\tri_{\qo})}{\expect[\vals\sim \dist]{\pseudomech(\vals)}}
&= \frac{\OPT_{\tri_{\qo}}(\tri_{\qo})}{(1+\delta) \mechaARO(\tri_{\qo}) - \delta\, \mecha_{1.1}(\tri_{\qo})} \\
&\leq\frac{\OPT_{\tri_{\qo}}(\tri_{\qo})}
{\frac{1+\delta}{1.9068943} \OPT_{\tri_{\qo}}(\tri_{\qo}) - \frac{\delta}{2}\OPT_{\tri_{\qo}}(\tri_{\qo})} 
\leq 1.90676.
\end{align*}
For $\qo \not\in [0.05, 0.25]$, 
we have $\mecha_{1.1}(\tri_{\qo}) \leq \OPT_{\tri_{\qo}}(\tri_{\qo})$, 
and by \Cref{lem:upper bound ratio for opt mech}, 
\begin{align*}
\frac{\OPT_{\tri_{\qo}}(\tri_{\qo})}{\expect[\vals\sim \dist]{\pseudomech(\vals)}}
&= \frac{\OPT_{\tri_{\qo}}(\tri_{\qo})}{(1+\delta) \mechaARO(\tri_{\qo}) - \delta\, \mecha_{1.1}(\tri_{\qo})} \\
&\leq\frac{\OPT_{\tri_{\qo}}(\tri_{\qo})}
{\frac{1+\delta}{1.9041} \OPT_{\tri_{\qo}}(\tri_{\qo}) - \delta\,\OPT_{\tri_{\qo}}(\tri_{\qo})} 
\leq 1.90676.
\end{align*}
Therefore, the gap between program~\eqref{eq:heuristic}
and~\eqref{eq:resolution}
is at least $10^{-4}$ for triangle distributions.

\Cref{thm:gap triangle} can be generalized, from the restriction to
triangle distributions, to regular distributions.  This result is
summarized below in \Cref{thm:gap} and its proof is given in
\Cref{app:gap}.  The argument generalizes \Cref{lem:lower bound ratio
  for m1} and \Cref{lem:upper bound ratio for opt mech} to give bounds
on markup mechanisms' revenues for regular distributions that are
close to or far from the worst-case triangle distribution
$\tri_{\optqo}$.  A key part of the proof is an appropriate definition
of closeness.

\begin{theorem}\label{thm:gap}
  For i.i.d., regular, two-agent, single-item environments and
  scale-invariant, incentive-compatible mechanisms the heuristic
  benchmark optimization program~\eqref{eq:heuristic} has a strictly
  smaller objective value than the benchmark optimization
  program~\eqref{eq:resolution}, i.e., $\heuristic < \resolution$.
\end{theorem}

Note that the parameters for the benchmark and the analysis of
\Cref{thm:gap triangle} and \Cref{thm:gap} are not optimized for the
relaxed benchmark program.  The main message of this section is just
to show that the gap exists and that the relaxed benchmark program is
inappropriate for optimizing benchmarks.

% \begin{lemma}
% For any $1\leq \ratio_1 \leq \ratio_2 \leq \ratio_3$, 
% and for any $a, b, c \geq 0$ such that 
% $a+b-c \leq 1$ and $a\geq c, b\geq c$, 
% for any regular distribution $\dist$
% with single agent revenue normalized to 1, 
% \begin{align*}
% a \, \mecha_{\ratio_1}(\dist)
% + b \, \mecha_{\ratio_3}(\dist)
% - c \, \mecha_{\ratio_2}(\dist)
% \leq \optf{\dist}.
% \end{align*}
% \end{lemma}

\begin{figure}[t]
    \centering
    \begin{tabular}{|c|c|c|c|}
    \hline
         & $[0, \qo]$ & $[\qo, \qt]$ & $[\qt, 1]$ \\
    \hline
    $\revv_{\dist}(\valf_{\dist}(\quant) \cdot \ratio_1)$ 
    & $\revv_{\dist}(\valf_{\dist}(\quant))$
    & 1 & 1 \\
    \hline
    $-\revv_{\dist}(\valf_{\dist}(\quant) \cdot \ratio_2)$ 
    & $-\revv_{\dist}(\valf_{\dist}(\quant) \cdot \ratio_3)$
    & $-\revv_{\dist}(\valf_{\dist}(\quant) \cdot \ratio_3)$ & $-\revv_{\dist}(\valf_{\dist}(\quant) \cdot \ratio_1)$ \\
    \hline
    $\revv_{\dist}(\valf_{\dist}(\quant) \cdot \ratio_3)$ 
    & $\revv_{\dist}(\valf_{\dist}(\quant))$
    & 1 & 1 \\
    \hline
    \end{tabular}
    \caption{Bounds on revenues for the three markup mechanisms in
      three quantile ranges are shown.  As $a$ and $b$ are both at least $c$, we can combine bounds $\mech_{\ratio_2}$ with appropriate bounds for $\mech_{\ratio_1}$ and $\mech_{\ratio_3}$ to obtain: $b\,\revv_{\dist}(\valf_{\dist}(\quant)
      \cdot \ratio_3) - c\,\revv_{\dist}(\valf_{\dist}(\quant) \cdot
      \ratio_2) \leq (b-c)\, \revv_{\dist}(\valf_{\dist}(\quant))$ for
    $q \in [0,\qo]$ and $a\,\revv_{\dist}(\valf_{\dist}(\quant) \cdot
    \ratio_1) + b\,\revv_{\dist}(\valf_{\dist}(\quant) \cdot
    \ratio_3) - c\,\revv_{\dist}(\valf_{\dist}(\quant) \cdot \ratio_2)
    \leq (a+b-c)$ for $q \in [\qo,1]$.
    \label{tab:bound}}
\end{figure}

\begin{proof}[Proof of \Cref{lem:upper bound benchmark}]
Let $\qo$ be the monopoly quantile of normalized distribution $\dist$
(i.e., with monopoly revenue 1).  For convenience, denote the revenue
as a function of price as $\revv_{\dist}(\val) =
\rev_{\dist}(\quantf_{\dist}(\val))$.  Let $\qt =
\quantf_{\dist}(\sfrac{\valf_{\dist}(\qo)}{\ratio_2})$ be the quantile
that corresponds to the monopoly reserve divided by $\ratio_2$.  The
proof follows from the bounds in \Cref{tab:bound} (derived below) and
the following analysis.  For $\quant \in [0,\qo]$, combine the bounds
in \Cref{tab:bound} for $\mech_{\ratio_2}$ and
$\mech_{\ratio_3}$ to obtain:
\begin{align*}
 b\,\revv_{\dist}(\valf_{\dist}(\quant)
      \cdot \ratio_3) - c\,\revv_{\dist}(\valf_{\dist}(\quant) \cdot
      \ratio_2) &\leq (b-c)\, \revv_{\dist}(\valf_{\dist}(\quant)).\\
      \intertext{Add the bound for $\mech_{\ratio_1}$:}
 a\,\revv_{\dist}(\valf_{\dist}(\quant)
      \cdot \ratio_1) + b\,\revv_{\dist}(\valf_{\dist}(\quant)
      \cdot \ratio_3) - c\,\revv_{\dist}(\valf_{\dist}(\quant) \cdot
      \ratio_2) &\leq (a+b-c)\, \revv_{\dist}(\valf_{\dist}(\quant)).\\
      \intertext{Similarly, for $\quant \in [\qo,1]$ combine bounds in \Cref{tab:bound} for $\mech_{\ratio_2}$ with those for $\mech_{\ratio_1}$ or $\mech_{\ratio_3}$ (depending on $\quant$) to obtain:}
      a\,\revv_{\dist}(\valf_{\dist}(\quant) \cdot
      \ratio_1) + b\,\revv_{\dist}(\valf_{\dist}(\quant) \cdot
    \ratio_3) - c\,\revv_{\dist}(\valf_{\dist}(\quant) \cdot \ratio_2)
    &\leq (a+b-c).
    \end{align*}
Thus, we see the revenue of the pseudo mechanism is,
\begin{align*}
    \pseudomech(\dist) & = 2\int_0^1 \left[
    a\,\revv_{\dist}(\valf_{\dist}(\quant) \cdot  \ratio_1)
    + b\,\revv_{\dist}(\valf_{\dist}(\quant) \cdot  \ratio_3)
    - c\,\revv_{\dist}(\valf_{\dist}(\quant) \cdot  \ratio_2)\right]\, d\quant\\
    & \leq 2\,(a+b-c)\, \left[\int_0^{\qo} \revv_{\dist}(\valf_{\dist}(\quant))\, d\quant + (1-\qo)\right]\\
    & = \optf{\dist}.
    \end{align*}

The proof then follows from the inequalities of \Cref{tab:bound} which
we now derive.  The distribution $\dist$ is regular and by definition
its revenue curve $\rev_{\dist}(\quant)$ is concave in $\quant$
(\Cref{f:regular imply monotone}).  Therefore, $\rev_{\dist}(\quant)$
is increasing in $\quant$ for any quantile smaller than the monopoly
quantile $\qo$, i.e., $\quant \leq \qo$, and $\rev_{\dist}(\quant)$ is
decreasing in $\quant$ for any quantile larger than the monopoly
quantile $\qo$, i.e., $\quant \geq \qo$.  Since the quantile is
decreasing in values, we have that $\revv_{\dist}(\val)$ is decreasing
in $\val$ for any value %larger than the monopoly price
$\valf_{\dist}(\qo)$, i.e., $\val \geq \valf_{\dist}(\qo)$, and
$\revv_{\dist}(\val)$ is increasing in $\val$ for any value
$\valf_{\dist}(\qo)$, i.e., $\val \leq \valf_{\dist}(\qo)$.  The
$[0,\qo]$ column of \Cref{tab:bound} follows from this monotonicity of
$\revv_{\dist}(\cdot)$.  The $\mech_{\ratio_2}$ row also follows from
this monotonicity with the observation that $\qt$ is defined so
$\valf_{\dist}(\qt)\cdot \ratio_2$ equals the monopoly price
$\valf_{\dist}(\qo)$.  Thus, for $\quant < \qt$, higher prices such as
$\valf_{\dist}(\quant) \cdot \ratio_3$ give lower revenue than
$\valf_{\dist}(\quant) \cdot \ratio_2$; and for $\quant \geq \qt$
lower prices such as $\valf_{\dist}(\quant) \cdot \ratio_1$ give lower
revenue than $\valf_{\dist}(\quant) \cdot \ratio_2$ (thus, the negated
revenues are upper bounds).  The remaining entries of 1 in
\Cref{tab:bound} follow from normalization which gives
$\revv_{\dist}(\val) \leq 1$ for all prices~$\val$.
\end{proof}

\begin{figure}[t]
\centering
\begin{tikzpicture}[scale = 0.7]

\draw (-0.2,0) -- (10, 0);
\draw (0, -0.2) -- (0, 5.5);

\draw plot [smooth, tension=0.6] coordinates {(0, 0) (1.5,3.6) (3, 5) (6, 3.8) (9, 0)};

\draw [dotted] (0, 5) -- (3, 5);
\draw [dotted] (3, 0) -- (3, 5);

\draw (0, -0.5) node {$0$};
\draw (-0.5, 5) node {$1$};

\draw (9, -0.5) node {$1$};
\draw (3, -0.5) node {$\qo$};

\end{tikzpicture}
\caption{\label{f:regular imply monotone} 
For regular distributions, 
$\rev_{\dist}(\quant)$ is increasing in $\quant$ for any quantile $\quant \leq \qo$, 
and $\rev_{\dist}(\quant)$ is decreasing in $\quant$ for any quantile $\quant \geq \qo$. 
}
\end{figure}

% where the second inequality holds by rearranging the terms and 
% (a) for any $\ratio \geq 1$ and any quantile $\quant \leq \qo$, 
% we have $\valf_{\dist}(\quant) \cdot \ratio \geq \valf_{\dist}(\quant) \geq  \valf_{\dist}(\qo)$
% and hence $\revv_{\dist}(\valf_{\dist}(\quant) \cdot \ratio) \leq \revv_{\dist}(\valf_{\dist}(\quant))$;
% (b) for any $\val$,  $\revv_{\dist}(\val)
% \leq 1$ since the optimal single agent revenue is normalized to 1;
% (c) $\qt \geq \qo$ since posting lower prices gives higher quantile. 
% The last inequality holds 
% because $a+b-c \leq 1$.  

\section{Prior-free versus Prior-independent Expert Learning}
\label{sec:online-learning}

A main result of the paper, given in \Cref{sec:prior-free}, is that
optimal benchmark design, as we have defined it, is equivalent to
prior-independent optimization.  Moreover, the optimal prior-free
algorithm for the optimal benchmark is the optimal prior-independent
algorithm.  A consequence of these results is that there is no added
robustness from the prior-free framework over the prior-independent
framework.  In this section we observe, by an example of expert
learning, that this potential lack of robustness is serious and the
optimal prior-independent algorithm can perform much worse than the
standard algorithms that are known to approximate the standard
prior-free benchmark.  These observations are straightforward from the
perspective of the expert learning literature; we discuss them in
detail so as to map them onto the framework of \Cref{sec:prior-free}
and give formal proofs for completeness.  %\iffocs(An appendix section of our full version parallels \Cref{sec:gap} shows that the relaxed benchmark design program is not without loss of generality for expert learning.)
\iffocs(Our full version shows that the relaxed benchmark design program is not without loss of generality.)% for expert learning.)
\else(Moreover, \Cref{app:gap learning} parallels \Cref{sec:gap} shows that the relaxed
benchmark design program is not without loss of generality for expert learning.)\fi

We consider the binary-reward variant of the canonical online expert
learning problem.  A single player plays a repeated game against
Nature for $\horizon$ rounds.  In each round $\round$, each expert~$j$
from a discrete set $\{1,\ldots,\numexpert\}$ will receive a binary
reward $\vali[\round,j]\in\{0,1\}$.
Thus, the input space is $\feasibleVals =
[\{0,1\}^\numexpert]^{\horizon}$.  Before rewards are realized, the
player chooses to ``follow'' a (possibly randomized) expert for the
round, and receives a reward (possibly in expectation) equal to the
reward of the followed expert.  When the round concludes, the player
gets to observe the rewards of all experts, including those not
followed by the player.  The player's algorithm is $\mecha$ which
outputs distributions $\mecha_\round(\vals)$ over experts using only
the history $\left(\vali[1],\ldots,\vali[\round-1]\right)$ in each round
$\round$.  The class of all such {\em online} algorithms is denoted by
$\OLA$ and the performance of an online algorithm $\mecha\in\OLA$ on
input $\vals$ is:
$$
\mecha(\vals) = 
\sum\nolimits_{\round=1}^{\horizon} \expect[j\sim \mecha_\round(\vals)]{\vali[\round,j]}.
$$

As described in \Cref{sec:prior-free} we can define Bayesian,
prior-independent, and prior-free versions of the expert learning
problem.
We summarize as follows:
\begin{itemize}
\item In the Bayesian model, the optimal algorithm is $\OPT_{\dist}
  = \argmax_{\mecha \in \OLA} \mecha(\dist)$.

Consider the following family of {\em binary independent stationary}
distributions $\ISD$ for the Bayesian variant of the expert learning
problem.  For a distribution $\dist \in \ISD$, each expert
$j$'s reward in each round is a Bernoulli random variable with mean
$\mean_j$.  The class $\ISD$ is composed of all possible means $\mean_j\in[0,1]$.  At each round $\round$, the rewards are drawn independently
from each other and from other rounds.  Importantly the distribution
of each expert's reward is identical across rounds.
  For binary independent stationary distributions $\dist \in \ISD$, the
  optimal algorithm picks the expert with the highest ex ante
  probability $j^* = \argmax_j \mean_j$ and follows expert $j^*$ in each
  round; its expected performance is $$\optf{\dist} = \horizon\,\max\nolimits_j \mean_j.$$

\item In the prior-free model, the optimal algorithm is the one that
  minimizes regret in worst-case over inputs $\vals \in \feasibleVals$
  against a given benchmark $\benchmark$ defined
  as $$\pfratio^{\benchmark} = \min_{\mecha \in \OLA}\max_{\vals \in
    \feasibleVals} [\benchmark(\vals) - \mecha(\vals)].$$

  The \emph{best-in-hindsight} benchmark for any reward profile $\vals
  \in \feasibleVals$ is
  $$\BIH(\vals) = \max\nolimits_{j=1}^{\numexpert}
  \sum\nolimits_{\round=1}^{\horizon} \vali[\round, j].$$
  A typical
  online analysis measures performance in terms of worst-case regret
  with respect to the best-in-hindsight benchmark $\BIH$.

\item In the prior-independent model, the optimal algorithm is the one
  that minimizes regret in worst-case over distributions $\dist \in
  \feasibleDist$ against the optimal algorithm for the distribution
  $$\piratio = \min_{\mecha \in \OLA}\max_{\dist \in \feasibleDist}
  [\optf{\dist} - \mecha(\dist)].$$

  We will be considering this question for binary independent stationary
  distributions $\ISD$ where $\optf{\dist}$ is as described above.
\end{itemize}

We observe next that the best-in-hindsight benchmark is normalized.
(In fact, it is analogous to the normalized benchmark described by
\citet{HR-08} for evaluating prior-free mechanisms.)  Thus, an
algorithm that is a prior-free approximation of the benchmark is also
a prior-independent approximation algorithm (with the same bound on
regret, cf.\ \Cref{prop:npf=>pi}).

\begin{lemma}
  \label{prop:bihisnorm}
  For inputs from binary independent stationary distributions $\ISD$, the
  best-in-hindsight benchmark $\BIH$ is normalized, i.e., $\BIH(\dist)
  \geq \optf{\dist}$, $\forall \dist \in \ISD$.
  \end{lemma}
\begin{proof}
  Given $\dist \in \ISD$, the Bayesian optimal algorithm
  $\OPT_{\dist}$ selects the same expert in each round.  The
  best-in-hindsight benchmark selects the single expert that is best
  for the realized input $\vals$.  Thus, for all $\vals \in
  \feasibleVals$, $\BIH(\vals) \geq \OPT_{\dist}(\vals)$.  Taking
  expectations we have the lemma.
\end{proof}

We now show that the natural {\em follow-the-leader} algorithm, which
in round $\round$ chooses a uniform random expert from the set
of experts with highest total reward from the first $\round-1$ rounds, is the
prior-independent optimal algorithm.

\begin{definition}
  The \emph{follow-the-leader} algorithm selects an expert uniformly
  at random from the set of experts with highest total
  reward from previous rounds:
$$\FTL_t(\vals) = U[\argmax\nolimits_j \sum\nolimits_{\round' < \round} \vali[\round',j]].$$
\end{definition}

\begin{theorem}\label{thm:pi learning}
For binary independent stationary distributions, the follow-the-leader
algorithm is the prior-independent optimal online learning algorithm.
\end{theorem}

\begin{proof} 
  Consider the Bayesian optimal online algorithm for the {\em uniform
    permutation prior} defined by probabilities
  $\{\mean^{j}\}_{j=1}^n$ that are assigned to experts via a uniform
  random permutation $\permute$ (i.e., the reward of expert $j$ is
  Bernoulli with mean $\mean_j = \mean^{\permute(j)}$).  The theorem
    follows from the optimality of follow-the-leader for any
  uniform permutation prior.
  
  We first argue that the follow-the-leader algorithm is optimal for
  the any uniform permutation prior.  The optimal algorithm for the
  uniform permutation prior forms a posterior from the reward
  history at any time $\round$ and chooses the expert with the highest
  expectation under this posterior.  Naturally, the experts with the
  highest expected reward under the posterior are the ones with the
  highest historical reward \iffocs(proof of this claim is shown in the appendix of the full version).  \else(formally, \Cref{lem:posterior expert and FTL} in \Cref{app:learning}).  \fi In other words, follow-the-leader
  is the Bayesian optimal algorithm for the uniform permutation prior.

  As before denote by $\piratio$ the prior-independent optimal regret.
  To complete the proof, consider the probabilities
  $\{\mean^{j}\}_{j=1}^n$ for which the Bayesian optimal algorithm for
  the uniform permutation prior obtains the largest regret.  Observe
  that the regret of the Bayesian optimal algorithm for this uniform
  permutation prior lower bounds $\piratio$.  On the other hand, the
  prior-independent regret of the follow-the-leader algorithm upper
  bounds $\piratio$.  The follow-the-leader algorithm obtains the same
  regret on all permutations and this regret equals the Bayesian
  optimal regret for the uniform permutation prior; i.e., the upper
  bound and the lower bound are equal.
\end{proof}

As we have proved in \Cref{sec:prior-free}, benchmark
optimization~\eqref{eq:resolution} and prior-independent
optimization~\eqref{eq:pi} are the same.  Thus, the optimal prior-free
benchmark is the performance of the follow-the-leader algorithm scaled up
by its prior-independent approximation factor.  Moreover, the optimal
mechanism for the optimal benchmark is the follow-the-leader algorithm
itself.  While we may have hoped for the prior-free analysis to lead
to more robust algorithms than the prior-independent analysis, by
optimizing benchmarks in the framework provided in
\Cref{sec:prior-free}, we have lost all of this potential robustness.
Specifically, the standard expert learning algorithms that have low
worst-case regret against the best-in-hindsight benchmark exhibit
robustness that the follow-the-leader lacks.  This observation is
formalized in the following lemma which contrasts with the optimal
regret of standard algorithms like randomized weighted-majority
\citep{LW-94}.  The optimal worst-case regret against
best-in-hindsight is $\Theta(\sqrt{\horizon \ln \numexpert})$ for
$\numexpert$ experts, $\horizon$ rounds, and binary
rewards  \citep{HKW-95}.

\begin{lemma}
  The prior-free regret of follow-the-leader against the
  best-in-hindsight benchmark with $\horizon$ rounds is $\Theta(n)$.
\end{lemma}

\begin{proof}
  Consider the input with an even number of rounds and rounds alternating as:
  \begin{itemize}
  \item odd round payoffs: $(1,0,\ldots,0) \in \{0,1\}^\numexpert$,
  \item even round payoffs: $(0,1,\ldots,1)\in \{0,1\}^\numexpert$.
  \end{itemize}

  The follow-the-leader algorithm chooses a uniform random expert for odd rounds
  and obtains expected payoff $\frac{1}{\numexpert}$ and chooses expert 1 for
  even rounds and obtains expected payoff of $0$.  The total expected
  payoff of follow the leader is $\FTL(\vals) = \frac{\horizon}{2\numexpert}$. 
  On the other hand, the best-in-hindsight benchmark is $\BIH(\vals) = \frac{n}{2}$.  The additive regret is $\BIH(\vals) - \FTL(\vals) = \frac{n}{2}(1-\frac{1}{k}) \in \Theta(n)$.
\end{proof}

The observations of this section suggest that further study of the
formulation of the benchmark optimization problem is necessary to
better understand the trade-offs between prior-free and
prior-independent robustness.

%GIVE LOWER BOUND INSTANCE WHERE FOLLOW-THE-LEADER HAS HIGH REGRET WRT BEST-IN-HINDSIGHT.

%I think the worst input is

%(1 0 ... 0) => 1/n
%(0 1 ... 1) => 0
%(1 0 ... 0) => 1/n
%(0 1 ... 1)
%...

%Where FTL gets T/2n and BIH gets T/2 for

%Despite the poor performance of follow-the-leader in worst case, from \Cref{thm:equiv} we have the following positive corollary.

%\begin{corollary}
%\label{cor:ftlsolvesgamma}
%For the expert learning setting with $\numexpert$ experts, finite horizon $\horizon$, and payoffs from $\{0,1\}$, the algorithm with smallest approximation fa%ctor to any normalized benchmark is follow-the-leader.
%\end{corollary}

\section{Historical Context: Online Algorithms and Mechanism Design}
\label{sec:related-work}

This section describes the historical development of online algorithms and
mechanism design and compares approaches across these fields.

Online algorithms have been analyzed via a worst-case competitive
analysis since \citet{ST-85} with textbooks on the subject, e.g.,
\citet{BE-05}.  In competitive analysis, the performance of an online
algorithm is measured as its worst case ratio to the optimal offline
algorithm.  Good online algorithms are known with respect to this
measure for many problems.  For some problems this measure is too
pessimistic, occurring when no good ratio is achievable by any
algorithm and therefore good algorithms are not meaningfully separated
from bad algorithms.  The two approaches for resolving this issue are
to either (a) restrict the offline algorithm to which the performance
of the online algorithm is compared or (b) restrict the family of
inputs that are considered.  For this paper, the most relevant example
of (a) comes from online learning where a learning algorithm's regret
is measured with respect to the best fixed action in hindsight, i.e.,
to the optimal offline algorithm that is restricted to choose the same
action in each time period \citep{LW-94,FS-97}.  The most relevant
example of (b) for this paper is the diffuse adversary model of
\citet{KP-00} which evaluates an algorithm as the ratio between its
expected performance and the optimal offline performance in worst-case
over a family of distributions on inputs.

Less immediately related to our work is a systematic study of analysis
frameworks for online algorithms conducted by \citet{BIL-15}.  They
focus on a canonical {\em two server problem} and show how specific
analysis frameworks result in different rankings of standard
algorithms for the problem.  These frameworks tend to vary the inputs
compared between the algorithm and the benchmark rather than varying
the benchmark itself (as we do).  For example, the competitive ratio
is the worst case over inputs of the optimal offline performance to
the online algorithm's performance; while the max-max ratio compares
the worst optimal offline performance to the worst algorithm
performance, with both worst-cases taken over normalized inputs.  They
also consider analyses that directly compare algorithms 
rather than comparing the algorithms indirectly via a benchmark.

To place in context the development of robust analyses of mechanisms
it is helpful to consider the predominant method for the analysis of
mechanisms in economics.  In economics, the preferences of the agents
are assumed to be drawn from a known distribution and the
design framework asks for the mechanism that maximizes performance
(and satisfies incentive constraints) in expectation over this prior
distribution over preferences.  This area is known as {\em Bayesian
  mechanism design}, e.g., see the survey by \citet{har-13}.

Competitive analysis was introduced to the design of mechanisms by
\citet{GHKSW-06}.  While in online algorithms, the information
theoretic barrier to good performance is the lack of information that
comes from the online arrival of the input, for dominant strategy incentive compatible mechanism design, the
information theoretic barrier to good performance comes from the need
to satisfy the incentive constraints of the agents, namely that
truthtelling is an equilibrium.  The equivalent of the optimal offline
algorithm, namely the optimal performance without incentives, is
rarely an interesting benchmark as no good ratio is achievable.  Both
approaches (a) and (b) described above for online algorithms have been
taken.  For an example of (a), and with parallels to online learning,
\citet{GHKSW-06} compares mechanisms for selling a digital good to the
optimal revenue from posting a single price to all agents (i.e., the
price-posting-revenue benchmark discussed in the introduction).  They and the subsequent literature developed the area of prior-free
mechanism design.  In this area, the performance of a mechanism is
compared to a benchmark performance in worst-case over inputs.

Lower bounds are important in the study of good mechanisms.
\citet{GHKSW-06} introduced the following approach for establishing a
lower bound on the prior-free approximation of any mechanism to a
given benchmark.  The approach considers a distribution over inputs
for which all (undominated) mechanisms perform the same.  For the
objective of revenue maximization this distribution is known as the
{\em equal revenue distribution} and has cumulative distribution
function $\dist(z) = 1-1/z$ on support $[1,\infty)$.  An agent with value drawn from this distribution offered any price $p \geq 1$ accepts with
  probability $1/p$ and yields expected revenue 1.  \citet{GHKSW-06}
  show a lower bound on the prior-free approximation to a
  posted-price-revenue based benchmark of 2.42 in the limit with the
  number of agents going to infinity.  For the special case of $n=3$
  agents, the lower bound is 3.25; \citet{HM-05} proved that this
  lower bound in the $n=3$ case is indeed tight by giving a mechanism
  that achieved it.  \citet{CGL-14} proved that the lower bound of
  2.42 is tight with a non-constructive proof and, moreover, that the
  lower-bounding method gives a tight bound for a large family of
  benchmarks that, like the one of \citet{GHKSW-06}, are constant in the
  value of the highest agent.

\citet{HR-08} revisited the choice of benchmark of \citet{GHKSW-06}
and identified the normalization constraint as presented in \Cref{s:intro}.  %Recall, normalization of a benchmark is with respect to a family of distributions and requires that for all distributions in the family, the expected value of the benchmark is at least the expected performance of the Bayesian optimal mechanism for the distribution.
%\citet{HR-08} 
Recall, they observe that if a benchmark satisfies the normalization
constraint then mechanisms that approximate the benchmark in worst
case also approximate the Bayesian optimal mechanism for any
distribution in the family (with no worse an approximation factor).
This consequence is known as the {\em prior-independent corollary of
approximation of a normalized benchmark}.  A specific normalized benchmark that they give is defined by the supremum of the
performance of the Bayesian optimal mechanisms for each distribution
in the family.
% by the function (for any input) considering the supremum of the
% performance of the Bayesian optimal mechanisms for each distribution
% in the family. 
For expert learning, this benchmark is in fact the best-in-hindsight benchmark (see \Cref{sec:online-learning}). 
For auction settings, 
\citet{DHY-15} give a simpler normalized benchmark
based on relaxing the incentive constraints to constraints of
envy-freedom. 
% Connecting this idea to online algorithms, the
% best-in-hindsight benchmark of expert learning is an example of the
% \citet{HR-08} benchmark for the family of independent stationary
% distributions described in \Cref{sec:online-learning}.

Here we outline the origins of prior-independent analysis and compare it to the diffuse adversary approach in online algorithms.
The prior-independent corollary of prior-free approximation of the
benchmarks of \citet{HR-08} motivated the consideration of relaxing
the assumption of worst-case inputs in a similar fashion to approach
(b) above.  A key difference, however, between the diffuse adversary
model for online algorithms \citep{KP-00} and the prior-independent
model is that the diffuse adversary model compares algorithms against
the optimal offline benchmark (which relaxes the information theoretic
constraints of the online optimization problem) where as, in
prior-independent mechanism design, mechanisms are compared to the
optimal mechanism for the distribution (that satisfies the incentive
constraints).  The advantage of considering distributions in the
diffuse adversary model is more in the spirit of smoothed analysis.
\citet{DRY-15} considered prior-independent mechanism design as a
first-order goal and since then it has been the subject of a
flourishing area of research. 

The prior-independent mechanism design framework gives a natural
question of identifying the optimal mechanism.  This question is
framed by a restriction to a family of distributions, but is not
subject to an ad hoc performance benchmark as is prior-free mechanism
design.  Previous literature has only identified optimal
prior-independent mechanisms in environments that are special cases of the fully general problem.  % and are thus, in a sense, trivial. 
%\markup{Maybe a word of change for trivial? Aleck is concern about this may not give enough credit to other results. }
\citet{HR-14} gave the prior-independent optimal mechanism for revenue
maximization in the sale of a single item to a single agent with value
from a bounded support where the prior-independent optimal mechanism
posts a randomized price.  For revenue maximization in the sale of an item to
one of two agents with values drawn from an i.i.d.\ regular distribution, \citet{DRY-15} show that the second price auction
is a 2-approximation.  \citet{FILS-15} gave a randomized mechanism showing that %ed by example that this
this factor of 2 is not tight.  Upper and lower bounds on this canonical
problem were improved by \citet{AB-18} to be within $[1.80,1.95]$.%\markup{I think we need to point out here that AB'18 is with the assumption of scale invariance, which we also adopt}
A main result of our paper is to identify the prior-independent optimal
mechanism for this environment with a factor of about $1.91$.\footnote{The \citet{AB-18} lower bound of 1.80 was proved under the same additional assumption of scale invariance as our lower bound of 1.91.}  For this two agent problem with i.i.d.\ values from a
distribution in the subset of regular distributions that further satisfy a monotone hazard rate condition, \citet{AB-18} show that the
second-price auction is optimal.  As we see now, the optimality of
the second price auction results from constraints binding that are loose when the class of distributions includes Triangles.

Upper and lower bounds on this canonical
problem were improved by \citet{AB-18} to be within $[1.80,1.95]$.\footnote{Their lower bound of 1.80 holds under the additional assumption of scale invariance.}  %\markup{I think we need to point out here that AB'18 is with the assumption of scale invariance, which we also adopt}
A main result of ourthis paper is to identify the prior-independent optimal
mechanism for this environment with a factor of about $1.91$.  For this two agent problem with i.i.d.\ values from a
distribution in the subset of regular distributions that further satisfy a monotone hazard rate condition, \citet{AB-18} show that the
second-price auction is optimal.  
%\markup{my suggestion was to move this sentence to be the first sentence of the {\em NEXT} paragraph, with the summary of our other main results.  Also that puts it closer to the sentence ``thus, we identify the optimal normalized benchmark ..." statement which follows from equivalence.  Equivalence to what result?  To the result in the last sentence of the preceding paragraph; seems the thread could have been lost.}

In the context above, this paper formalizes the problem of identifying
the normalized benchmark that yields the tightest approximation by a
mechanism.  We show that this question is equivalent to solving the
prior-independent mechanism design problem.  Thus, from our prior-independent optimal mechanism, we identify the
optimal normalized benchmark for 2-agent revenue maximization with values drawn i.i.d.\ from regular distributions.  Moreover, in contrast to the positive main result of \citet{CGL-14}, we show that relaxing the program to
find the optimal benchmark to the simpler program of finding the
normalized benchmark that minimizes the lower bound from
\citet{GHKSW-06} (i.e., from the equal revenue distribution) is not
tight.  Specifically, the lower bound for the benchmark that is
optimal for the relaxed benchmark optimization program is
not tight.

\bibliography{auctions}

\appendix

\section{Missing Proofs from \Cref{sec:prior-independent}}
\label{app:prior-independent}

To identify the prior-independent optimal mechanism for triangle
distributions, we evaluate the ratio of revenues of markup mechanisms
on triangle distributions to the optimal revenue.  For distribution
$\tri_{\qo}$ the optimal revenue is $2-\qo$ (\Cref{lem:opt rev}).  The
revenue for $\ratio$-markup mechanism is calculated by \Cref{clm:rev
  of r}.  In this appendix, we drive the formula of \Cref{clm:rev of
  r} and show that it has bounded partial derivatives in both markup
$\ratio$ and monopoly quantile $\qo$.  We then describe the details of
the hybrid numerical and analytical argument of \Cref{thm:triangle}.
Finally we give the proof of continuity of the adversary's best
response distribution to the probability the mechanism places on the
second-price auction.

\subsection{Derivation and smoothness of \Cref{clm:rev of r}}

\begin{proof}[Proof of \Cref{clm:rev of r}]
  Denote the quantile 
  corresponding to the price $\ratio\,\valf_{\tri_{\qo}}(\quant)$ for markup $\ratio > 1$ as
  $$\ScaledQuant(\quant, \ratio) =
  \quantf_{\tri_{\qo}}(\ratio\,\valf_{\tri_{\qo}}(\quant)) = 
  \begin{cases}
    \frac{\quant}{\ratio - \quant\ratio + \quant} &\text{if }\ratio \, \valf_{\tri_{\qo}}(\quant)  \leq \sfrac{1}{\qo},\\
    0 & \text{otherwise.}
  \end{cases}
  $$
  When the quantile of the second highest agent is smaller than
  $\ScaledQuant(\qo, \sfrac{1}{\ratio})$, 
  the price $\ratio \cdot \vsec$ is higher than the support of the valuation distribution. 
  Therefore, %for any $r > 1$ and $q \in [0, 1)$, 
  the revenue of posting price $\ratio \cdot \vsec$
  to the highest bidder is 
  \begin{align*}
    \mecha_{\ratio}(\tri_{\qo}) 
    &=
    2 \ratio 
    \int_{\ScaledQuant(\qo, \sfrac{1}{\ratio})}^{1}
    \valf_{\tri_{\qo}}(\quant) 
    \ScaledQuant(\quant, \ratio)
    \,d\quant \\
    &= 
    2 \ratio 
    \int_{\ScaledQuant(\qo, \sfrac{1}{\ratio})}^{1}
    \frac{1-\quant}{1-\qo} \cdot 
    \frac{1}{\ratio - \quant\ratio + \quant}
    \,d\quant = \frac{2\ratio}{1-\qo}\left[\frac{\quant}{\ratio-1}+\frac{\ln(\ratio-\quant\ratio+\quant)}{(\ratio-1)^2} \right]_{\frac{\qo}{1/\ratio-\qo/\ratio+\qo}}^{1} \\
    % &= 
    % \frac{2 \ratio}{(1-\qo)(\ratio-1)}
    % \int_{\ScaledQuant(\qo, \sfrac{1}{\ratio})}^{1}
    % \left(1- 
    % \frac{1}{\ratio - \quant\ratio + \quant}
    % \right)
    % \,d\quant \\
    &= 
    \frac{2\ratio}{(1-\qo)(\ratio-1)} 
    \left(
    \frac{1-\qo}{1-\qo+\qo\ratio}
    -
    \frac{\ln \left(
      \frac{\ratio}{1-\qo+\qo\ratio}
      \right)}{\ratio-1}
    \right),
  \end{align*}
where the second equality holds just by the definition of the distribution. 
\end{proof}

Consider the revenue of $\ratio$-markup mechanism on the triangle
distribution $\tri_{\qo}$ as a function of $\ratio \in (1,\infty)$ and $\qo \in [0,1]$.  The
formula for this revenue is given by \Cref{clm:rev of r}.  The
following two claims show that the ratio of revenues has bounded
partial derivative with respect to both $\ratio \in (1, \infty)$ and $\qo\in[0, 1]$ and, thus,
numerical evaluation of the revenue at selected parameters allows
large regions of parameter space to be ruled out.

\begin{claim}\label{clm:cont in r}
For any distribution $\dist$ and any constants $1 \leq \ratio_1 \leq \ratio_2$, 
we have $\mecha_{\ratio_1}(\dist) \geq 
\sfrac{\ratio_1}{\ratio_2} \, \mecha_{\ratio_2}(\dist)$. 
\end{claim}
\begin{claim}\label{clm:cont in q}
For any mechanism $\mecha_{\ratio}$ with $\ratio \geq 1$, 
and any constants $0 \leq \qo_1 \leq \qo_2 < 1$, 
we have 
$\sfrac{(1-\qo_2)}{(1-\qo_1)} \,
\mecha_{\ratio}(\tri_{\qo_2}) 
\leq \mecha_{\ratio}(\tri_{\qo_1}) 
\leq 2(\qo_2 - \qo_1) 
+ \mecha_{\ratio}(\tri_{\qo_2})$.
\end{claim}

\begin{proof}[Proof of \Cref{clm:cont in r}]
For any realized valuation profile, 
if the item is sold in mechanism $\mecha_{\ratio_2}$, 
then the item is sold in mechanism $\mecha_{\ratio_1}$
since the price posted to the highest agent is smaller in mechanism $\mecha_{\ratio_1}$. 
Moreover, 
when the item is sold in mechanism $\mecha_{\ratio_1}$, 
the payment from agent with highest value is at least 
$\sfrac{\ratio_1}{\ratio_2}$ fraction of 
the payment in mechanism $\mecha_{\ratio_2}$. 
Taking expectation over the valuation profiles, 
we have $\mecha_{\ratio_1}(\dist) \geq 
\sfrac{\ratio_1}{\ratio_2} \cdot \mecha_{\ratio_2}(\dist)$. 
\end{proof}

\begin{proof}[Proof of \Cref{clm:cont in q}]
  Consider $\ScaledQuant(\cdot,\cdot)$ as defined in the proof of
  \Cref{clm:rev of r}, above.  By directly comparing the revenue from
  two distributions,
\begin{eqnarray*}
\mecha_{\ratio}(\tri_{\qo_1}) 
&=& 
2 \ratio 
\int_{\ScaledQuant(\qo_1, \sfrac{1}{\ratio})}^{1}
\valf_{\tri_{\qo_1}}(\quant) 
\,\ScaledQuant(\quant, \ratio)
\,d\quant \\
&\leq& 2(- \ScaledQuant(\qo_1, \sfrac{1}{\ratio})
+ \ScaledQuant(\qo_2, \sfrac{1}{\ratio})) 
+ 2\ratio\int_{\ScaledQuant(\qo_2, \sfrac{1}{\ratio})}^{1}
\valf_{\tri_{\qo_1}}(\quant) 
\,\ScaledQuant(\quant, \ratio)
\,d\quant \\
&\leq& 2(\qo_2 - \qo_1) 
+ 2\ratio\int_{\ScaledQuant(\qo_2, \sfrac{1}{\ratio})}^{1}
\valf_{\tri_{\qo_2}}(\quant) 
\,\ScaledQuant(\quant, \ratio)
\,d\quant \\
&=& 
2(\qo_2 - \qo_1) 
+ \mecha_{\ratio}(\tri_{\qo_2}).
\end{eqnarray*}
The first equality holds because the quantile of $\valf_{\tri_{\qo_1}}(\quant) \cdot \ratio$
is 0 for $\quant < \ScaledQuant(\qo_1, \sfrac{1}{\ratio})$. 
The first inequality holds because 
$\ratio \cdot 
\valf_{\tri_{\qo_1}}(\quant) 
\ScaledQuant(\quant, \ratio) \leq 1$ 
for any quantile $\quant$.
The second inequality holds because 
$\valf_{\tri_{\qo_1}}(\quant) \leq \valf_{\tri_{\qo_2}}(\quant) $
for $\qo_1 \leq \qo_2$ 
and $\quant \geq \qo_2$
by the definition of distributions $\tri_{\qo_1}$ and $\tri_{\qo_2}$, 
and 
$\ScaledQuant(\qo_2, \sfrac{1}{\ratio}) 
- \ScaledQuant(\qo_1, \sfrac{1}{\ratio})
\leq \qo_2 - \qo_1$. 
Moreover, we have
\begin{eqnarray*}
\mecha_{\ratio}(\tri_{\qo_1}) 
&=& 
2 \ratio 
\int_{\ScaledQuant(\qo_1, \sfrac{1}{\ratio})}^{1}
\valf_{\tri_{\qo_1}}(\quant) 
\ScaledQuant(\quant, \ratio)
\,d\quant \\
&\geq& 2\ratio
\int_{\ScaledQuant(\qo_2, \sfrac{1}{\ratio})}^{1}
\valf_{\tri_{\qo_1}}(\quant) 
\ScaledQuant(\quant, \ratio)
\,d\quant \\
&\geq& 
\frac{2\ratio(1-\qo_2)}{1-\qo_1}
\int_{\ScaledQuant(\qo_2, \sfrac{1}{\ratio})}^{1}
\valf_{\tri_{\qo_2}}(\quant) 
\ScaledQuant(\quant, \ratio)
\,d\quant \\
&=& 
\frac{1-\qo_2}{1-\qo_1} \cdot
\mecha_{\ratio}(\tri_{\qo_2}), 
\end{eqnarray*}
where the first inequality holds because $\qo_1 \leq \qo_2$ 
and function $\ScaledQuant(\quant, \ratio)$ is increasing in $\quant$. 
The second inequality holds because 
$\valf_{\tri_{\qo_1}}(\quant) \geq
\sfrac{(1-\qo_2)}{(1-\qo_1)} \cdot
\valf_{\tri_{\qo_2}}(\quant) $.
\end{proof}

% \begin{proof}[Proof of \Cref{thm:triangle}]
% Combining \Cref{thm:lower} and \ref{thm:upper},
% we have 
% \begin{align*}
% \min_{\mecha \in \scalei}
% \max_{\dist \in \regular}
% \frac{\optf{\dist}}{\mecha(\dist)} 
% \leq \max_{\dist \in \TRIF}
% \min_{\mecha \in \scalei}
% \frac{\optf{\dist}}{\mecha(\dist)}.
% \end{align*}
% Also by weak duality, we have 
% \begin{align*}
% \min_{\mecha \in \scalei}
% \max_{\dist \in \regular}
% \frac{\optf{\dist}}{\mecha(\dist)} 
% \geq \max_{\dist \in \regular}
% \min_{\mecha \in \scalei}
% \frac{\optf{\dist}}{\mecha(\dist)}
% \geq \max_{\dist \in \TRIF}
% \min_{\mecha \in \scalei}
% \frac{\optf{\dist}}{\mecha(\dist)}.
% \end{align*}
% Combining the inequalities, we have
% \begin{align*}
% \min_{\mecha \in \scalei}
% \max_{\dist \in \regular}
% \frac{\optf{\dist}}{\mecha(\dist)} 
% = \max_{\dist \in \TRIF}
% \min_{\mecha \in \scalei}
% \frac{\optf{\dist}}{\mecha(\dist)}
% \approx \apxsimp.
% \end{align*}
% This implies that the mechanism $\mechaARO$ found in \Cref{thm:upper}
% is the optimal prior-independent scale-invariant mechanism, and
% \Cref{thm:pi optimal mechanism} holds.
% \end{proof}

\subsection{Numerical and Analytical Arguments of \Cref{thm:triangle}}

The proof of \Cref{thm:triangle} is based on a hybrid numerical and
analytical argument.  We can numerically calculate the revenue of a
mechanism $\mech_{\ratio}$ on a distribution $\tri_{\qo}$ via
\Cref{clm:rev of r} and then we can argue, via \Cref{clm:cont in q}
and \Cref{clm:cont in r}, that nearby mechanisms and distributions
have similar revenue.  This approach will both allow us to argue about
the structure of the solution and to identify the mechanism
$\mechaARO$ and distribution of the solution $\tri_{\optqo}$.  Our
subsequent discussion gives the details of these hybrid arguments.

%% The high level idea for numerically calculating the parameters and
%% proving their approximate optimality is to first discretize the space
%% for parameters $\ratio^*, \qo^*, \alpha^*$ and find the best
%% parameters on discretized points.  This is done by evaluating the
%% approximation ratio of given parameters on discretized points using
%% \Cref{clm:rev of r}.  Then we bound the approximation ratio of
%% parameters between discretized points using the continuity argument
%% formalized in \Cref{clm:cont in r} and \ref{clm:cont in q} to get the
%% error bound of the selected parameters.

% By the characterization in \Cref{thm:triangle},
% to numerically calculate the parameters $\ratio^*, \qo^*, \alpha^*$, 
We first approximate $\qo^*$ 
by showing that $\qo^* \in [0.09310569, 0.09310571]$. 
The parameters for this range are found by discretizing the space and finding the optimal choice of $\qo^*$. 
Note that the optimal choice of $\qo^*$ satisfies
$\mecha_1(\tri_{\qo^*}) = \mecha_{\ratio(\qo^*)}(\tri_{\qo^*})$.
Therefore, it is sufficient for us to show that for any quantile $\qo \not\in [0.09310569, 0.09310571]$, 
either $\mecha_1(\tri_{\qo}) > \mecha_{\ratio(\qo)}(\tri_{\qo})$
or $\mecha_1(\tri_{\qo}) < \mecha_{\ratio(\qo)}(\tri_{\qo})$. 
% This approach holds since the optimal choice of $\qo^*$ satisfies
% $\mecha_1(\tri_{\qo^*}) = \mecha_{\ratio(\qo^*)}(\tri_{\qo^*})$.
% Moreover, the optimal prior independent approximation ratio 
% is $2-\qo^* \approx 1.9068943$. 
% Finally we identify choices of $\ratio^*, \alpha^*$ 
% such that the prior independent approximation ratio of mechanism $\mechaARO$ is at most $1.906894309$. 

% in \Cref{thm:triangle}, 
% we want to discretize the space of ratio $\ratio \in (1, \infty]$ and quantile $\qo \in [0, 1]$ and numerically calculate the revenue for those discretized points. 
% Note that by the characterization in \Cref{thm:triangle},

First we show for any $\qo \in [0, 0.09310569]$, 
$\mecha_1(\tri_{\qo}) < \mecha_{\ratio(\qo)}(\tri_{\qo})$. 
% numerically evaluate the monopoly quantile $\qo^*$
Here we discretize the space $[0, 0.09310569]$ into $\discreteQ$
with precision $\precision = 10^{-9}$. 
By numerically calculation using \Cref{clm:rev of r}, 
we have 
\begin{align*}
\min_{\qo \in \discreteQ} \mecha_{2.446946}(\tri_{\qo})
= \mecha_{2.446946}(\tri_{0.09310569})
\geq 1+10^{-8}
\end{align*}
and for any $\qo \in [0, 0.09310569]$, 
letting $\qo_d$ be the largest quantile in $\discreteQ$ smaller than or equal to $\qo$,
the minimum revenue for mechanism $\mecha_{2.446946}$ is 
\begin{align*}
\mecha_{2.446946}(\tri_{\qo})
\geq \frac{1-\qo_d -\precision}{1-\qo_d} 
\cdot \mecha_{2.446946}(\tri_{\qo_d})
\geq 1 + 8\times 10^{-9} > \mecha_1(\tri_{\qo}),
\end{align*}
where the first inequality holds by \Cref{clm:cont in q}
and the second inequality holds because $\qo_d \leq 0.1$.

Then we show for any $\qo \in [0.09310571, 1]$, 
$\mecha_1(\tri_{\qo}) > \mecha_{\ratio(\qo)}(\tri_{\qo})$. 
We discretize the space $[0.09310571, 1]$ into $\hat{\discreteQ}$
with precision $\hat{\precision} = 10^{-9}$. 
First note that
$\mecha_{\ratio}(\tri_{\qo}) < 1$ for any $\qo \geq 0.093$
and $\ratio \geq 11$,
since the expected probability the highest type got allocated is less than $\frac{1}{2}$, 
and hence the expected virtual value for mechanism $\mecha_{\ratio}$
with distribution $\tri_{\qo}$ is less than $1$.
By \Cref{lem:opt rev}, the revenue in this case is less than $1$. 
With bounded range for optimal ratio $\ratio$, 
we discretize the space $(1, 11]$ into $\discreteR$
with precision $\precisionR = 10^{-9}$. 
By numerically calculation using \Cref{clm:rev of r}, 
we have 
\begin{align*}
\max_{\qo \in \hat{\discreteQ}, \ratio \in \discreteR} \mecha_{\ratio}(\tri_{\qo})
= \mecha_{2.446945061}(\tri_{0.09310571})
\leq 1-3\times10^{-8}
\end{align*}
and for any $\qo \in [0.09310571, 1]$ and any $\ratio \in (1, 11]$, 
letting $\qo_d$ be the largest quantile in $\hat{\discreteQ}$ smaller than or equal to $\qo$
and $\ratio_d$ be the smallest number in $\discreteR$ larger than or equal to $\ratio$,
the maximum revenue for distribution $\tri_{\qo}$ is 
\begin{align*}
\max_{\ratio\in (1, 11]}\mecha_{\ratio}(\tri_{\qo})
\leq \frac{\ratio_d}{\ratio_d-\precision_r}\cdot 
(2\hat{\precision} + \mecha_{\ratio_d}(\tri_{\qo_d}))
\leq 1 - 10^{-8} < \mecha_1(\tri_{\qo}),
\end{align*}
where the first inequality holds by \Cref{clm:cont in r} and \ref{clm:cont in q},
and the second inequality holds because $\ratio_d > 1$. 
Combining the numerical calculation, we have that $\qo^* \approx \qval$.

Note that both mechanism $\mecha_1$ and $\mecha_{\ratio^*}$ are the best responses for distribution $\tri_{\qo^*}$, 
achieving revenue 1, 
and hence the optimal prior independent approximation ratio 
is 
\begin{align*}
\piratio = \frac{\OPT_{\tri_{\qo^*}}(\tri_{\qo^*})}{\mechaARO(\tri_{\qo^*})}
= 2-\qo^* \approx 1.9068943.
\end{align*}
% Finally we identify choices of $\ratio^*, \alpha^*$ 
% such that the prior independent approximation ratio of mechanism $\mechaARO$ is at most $1.906894309$. 

Next we show that by choosing ratio $\ratio^* \approx \rval$ and probability $\alpha^* \approx \aval$, 
the prior independent approximation ratio of mechanism $\mechaARO$ approximates $\beta$. 
Here we discretize the quantile space $[0, 1]$ into $\discreteQ'$ with precision $\precision' = 10^{-9}$, 
using the formula in \Cref{lem:opt rev} and \Cref{clm:rev of r},
the triangle distribution that maximizes the approximation ratio
for mechanism $\mechaARO$
is $\tri_{0.093105694}$ with 
approximation ratio at most $1.9068943044$. 
For any $\qo \in [0, \frac{1}{2}]$, 
letting $\qo_d$ be the largest quantile in $\discreteQ'$ smaller than or equal to $\qo$,
the minimum revenue for mechanism $\mechaARO$ is 
\begin{align*}
\mechaARO(\tri_{\qo})
\geq \frac{1-\qo_d -\precision'}{1-\qo_d} 
\cdot \mechaARO(\tri_{\qo_d})
\geq \frac{1}{1.906894309} \OPT_{\qo_d}(\tri_{\qo_d}) 
\geq \frac{1}{1.906894309} \OPT_{\qo}(\tri_{\qo}),
\end{align*}
where the second inequality holds because $\qo_d \leq \frac{1}{2}$
and the last inequality holds because $\qo_d \leq \qo$. 
For any $\qo \in [\frac{1}{2}, 1]$, 
the minimum revenue for mechanism $\mechaARO$ is 
\begin{align*}
\mechaARO(\tri_{\qo})
\geq \alpha^* \cdot \mecha_1(\tri_{\qo})
\geq 0.8
\geq \frac{1}{1.875} \OPT_{\qo}(\tri_{\qo}),
\end{align*}
since for any $\qo \in [\frac{1}{2}, 1]$, $\mecha_1(\tri_{\qo_d}) = 1$ and $\OPT_{\qo}(\tri_{\qo}) = 2-\qo \leq 1.5$. 
Therefore, $\ratio^* \approx \rval$ and probability $\alpha^* \approx \aval$ are the desirable parameters, 
with error at most $2\times 10^{-8}$ in prior independent approximation ratio. 
By our characterization, the error solely comes from numerical calculation, 
finishing the numerical analysis for \Cref{thm:triangle}.

\subsection{Continuity of Distribution in Probability of Second-price Auction}

Recall the function $\qo_{\ratio}(\alpha)$ which gives the adversary's
best-response triangle distribution the mechanism $\mechaAR$.  The
continuity of the function $\qo_{\ratio}(\alpha)$ is used to prove the
existence of equilibrium between the randomized markup mechanism and
the triangle distribution in \Cref{thm:triangle}.  The following claim
proves the continuity of the function $\qo_{\ratio}(\alpha)$, by
numerically bounding the second derivative of the revenue ratio of the
stochastic markup mechanism $\mechaAR$ on distribution $\tri_{\qo}$
with respect to $\secprob$, the probability that the markup mechanism
runs the second-price auction.
\begin{claim}\label{clm:continuous}
Given any $\ratio \in [\ratioS, \ratioL]$, 
function $\qo_{\ratio}(\alpha)$ is continuous in $\alpha$ for $\alpha \in [\alphaS, \alphaL]$. 
\end{claim}
\begin{proof}[Proof of \Cref{clm:continuous}]
By \Cref{clm:rev of r} and \Cref{thm:myerson}, 
the approximation ratio of mechanism $\mechaAR$ for triangle distribution $\tri_{\qo}$ is 
\begin{align*}
\APX(\alpha, \ratio, \qo) 
&= \frac{\OPT_{\tri_{\qo}}(\tri_{\qo})}{\alpha\cdot \mech_1(\tri_{\qo}) + (1-\alpha)\mech_{\ratio}(\tri_{\qo})} \\
&= \frac{2-\qo}{\alpha 
+ 
\frac{2\ratio(1-\alpha)}{(1-\qo)(\ratio-1)} 
\left(
\frac{1-\qo}{1-\qo+\qo\ratio}
+
\frac{\ln \left(
\frac{\ratio}{1-\qo+\qo\ratio}
\right)}{1-\ratio}
\right)}
\end{align*}
The approximation ratio is a continuous function of $\alpha, \qo$. 
Therefore, to show that fixing $\ratio$,  
function $\qo_{\ratio}(\alpha)$ is continuous in $\alpha$, 
it is sufficient to show that there is a unique $\qo$ that maximizes $\APX(\alpha, \ratio, \qo)$ 
for $\ratio \in [\ratioS, \ratioL]$ and $\alpha \in [\alphaS, \alphaL]$, 
or equivalently, we show that there is a unique $\qo$ that minimizes
$\sfrac{1}{\APX(\alpha, \ratio, \qo)}$. 
By \Cref{clm:cont in r} and \ref{clm:cont in q}, 
we can discretize the quantile space and numerically verify that 
distributions with monopoly quantiles $\qo \not\in [\qoS, \qoL]$ are suboptimal. 
Therefore, we prove the uniqueness of the maximizer by showing that the second order derivative of 
$\sfrac{1}{\APX(\alpha, \ratio, \qo)}$
is strictly positive for $\qo \in [\qoS, \qoL]$. 
\begin{align*}
\frac{\partial^2 \frac{1}{\APX(\alpha, \ratio, \qo)}}{(\partial \qo)^2}
= &
\frac{4(1-\alpha)\ratio
\left(-\frac{\ratio-1}{(1-\qo+\qo\ratio)^2}
+ \frac{1}{(1-\qo)(1-\qo+\qo\ratio)}
- \frac{\log(\frac{\ratio}{1-\qo+\qo\ratio})}{(\ratio-1)(1-\qo)^2}\right)}
{(\ratio-1)(2-\qo)^2}\\
& + \frac{2(1-\alpha)\ratio
\left(-\frac{(\ratio-1)^2}{(1-\qo+\qo\ratio)^3}
- \frac{\ratio-1}{(1-\qo)(1-\qo+\qo\ratio)^2}
+ \frac{2}{(1-\qo)^2(1-\qo+\qo\ratio)}
+ \frac{2\log(\frac{\ratio}{1-\qo+\qo\ratio})}{(\ratio-1)(1-\qo)^3}\right)}
{(\ratio-1)(2-\qo)}\\
& + \frac{4(1-\alpha)\ratio
\left(-\frac{1}{1-\qo+\qo\ratio}
- \frac{\log(\frac{\ratio}{1-\qo+\qo\ratio})}{(\ratio-1)(1-\qo)}\right) 
+ 2\alpha(\ratio - 1)}
{(\ratio-1)(2-\qo)^3}
% \\
% \geq& \frac{4(1-\ratioS)\ratioL
% \left(-\frac{\ratioL-1}{(1-\qoL+\qoS\ratioS)^2}
% + \frac{1}{(1-\qoL)(1-\qoS+\qoS\ratioS)}
% - \frac{\log(\frac{\ratioL}{1-\qoL+\qoS\ratioS})}{(\ratioS-1)(1-\qoL)^2}\right)}
% {(\ratioS-1)(2-\qoS)^2}\\
% & + \frac{2(1-\ratioS)\ratioS
% \left(-\frac{(\ratio-1)^2}{(1-\qo+\qo\ratio)^3}
% - \frac{\ratio-1}{(1-\qo)(1-\qo+\qo\ratio)^2}
% + \frac{2}{(1-\qo)^2(1-\qo+\qo\ratio)}
% + \frac{2\log(\frac{\ratio}{1-\qo+\qo\ratio})}{(\ratio-1)(1-\qo)^3}\right)}
% {(\ratio-1)(2-\qo)}\\
% & + \frac{4(1-\alpha)\ratio
% \left(-\frac{1}{1-\qo+\qo\ratio}
% - \frac{\log(\frac{\ratio}{1-\qo+\qo\ratio})}{(\ratio-1)(1-\qo)}\right) 
% + 2\alpha(\ratio - 1)}
% {(\ratio-1)(2-\qo)^3}
\end{align*}
By substituting the upper and lower bounds of $\alpha, \ratio, \qo$, 
we know that 
\begin{align*}
\frac{\partial^2 \frac{1}{\APX(\alpha, \ratio, \qo)}}{(\partial \qo)^2} > 0.7
\end{align*}
for $\ratio \in [\ratioS, \ratioL], \alpha \in [\alphaS, \alphaL]$ and $\qo \in [\qoS, \qoL]$,
which concludes the uniqueness of the maximizer and the continuity of function 
$\qo_{\ratio}(\alpha)$. 
\end{proof}

\section{Missing Proof from \Cref{sec:gap}}
\label{app:gap}

We first establish the proofs of \Cref{lem:lower bound ratio for m1}
and \Cref{lem:upper bound ratio for opt mech}.  We then give the
proof of \Cref{thm:gap} and auxiliary lemmas that establish a gap
between the relaxed benchmark design program and the benchmark design
program for regular distributions.

\begin{proof}[Proof of \Cref{lem:lower bound ratio for m1} and \Cref{lem:upper bound ratio for opt mech}]
The proof for these two lemmas is similar to the numerical calculation
in \Cref{thm:triangle}.  For \Cref{lem:lower bound ratio for m1}, we
discretize the quantile space $[0.05, 0.25]$ into $\discreteQ$ with
precision $\precision = 10^{-5}$, and compute the minimum
approximation ratio of mechanism $\mecha_{1.1}$ for any $\qo$ in the
discretized points in $\discreteQ$ using \Cref{clm:rev of r} (\Cref{sec:prior-independent}).  By
numerical calculation, the minimum approximation ratio is
\begin{align*}
\min_{\qo \in \discreteQ} \frac{\OPT_{\tri_{\qo}}(\tri_{\qo})}{\mecha_{1.1}(\tri_{\qo})}
= \frac{\OPT_{\tri_{0.05}}(\tri_{0.05})}{\mecha_{1.1}(\tri_{0.05})}
\geq 2.0026 
\end{align*}
and the minimum revenue for mechanism $\mecha_{1.1}$ is 
\begin{align*}
\min_{\quant \in \discreteQ} \mecha_{1.1}(\tri_{\qo})
= \mecha_{1.1}(\tri_{0.25})
\geq 0.74.
\end{align*}
For any $\qo \in [0.05, 0.25]$, 
let $\qo_d$ be the largest quantile in $\discreteQ$ smaller than or equal to $\qo$, 
we have 
\begin{align*}
\frac{\OPT_{\tri_{\qo}}(\tri_{\qo})}{\mecha_{1.1}(\tri_{\qo})}
&\geq \frac{\OPT_{\tri_{\qo_d+\precision}}(\tri_{\qo_d+\precision})}{2\precision + \mecha_{1.1}(\tri_{\qo_d})}
= \frac{2 - \qo_d-\precision}{2-\qo_d} \cdot 
\frac{\OPT_{\tri_{\qo_d}}(\tri_{\qo_d})}{2\precision + \mecha_{1.1}(\tri_{\qo_d})}\\
& \geq \frac{2 - \qo_d-\precision}{(2-\qo_d)(1+3\precision)} \cdot
\frac{\OPT_{\tri_{\qo_d}}(\tri_{\qo_d})}{\mecha_{1.1}(\tri_{\qo_d})}
\geq 2, 
\end{align*}
where the first inequality holds by applying \Cref{lem:opt rev} (\Cref{sec:prior-independent}) and \Cref{clm:cont in q} (\Cref{app:prior-independent}),
and the equality holds by applying \Cref{lem:opt rev}. 
The second inequality holds by noticing that $3\mecha_{1.1}(\tri_{\qo_d}) \geq 3\times0.74 \geq 2.$ 

We apply the similar analysis to the proof of \Cref{lem:upper bound ratio for opt mech}. 
We discretize the quantile space $[0, 0.05] \cup [0.25, 1]$ into $\hatdiscreteQ$ with precision $\hatprecision = 10^{-5}$, 
and compute the minimum approximation ratio of mechanism $\mechaARO$ for any $\qo$ 
in the discretized points in $\hatdiscreteQ$ using \Cref{clm:rev of r}. 
By numerical calculation, the minimum approximation ratio in $\hatdiscreteQ$ is 
\begin{align*}
\min_{\qo \in \hatdiscreteQ} \frac{\OPT_{\tri_{\qo}}(\tri_{\qo})}{\mechaARO(\tri_{\qo})}
= \frac{\OPT_{\tri_{0.05}}(\tri_{0.05})}{\mechaARO(\tri_{0.05})}
\leq 1.90406 
\end{align*}
and for any $\qo \in [0, 0.05] \cup [0.25, 0.5]$, we have 
\begin{align*}
\frac{\OPT_{\tri_{\qo}}(\tri_{\qo})}{\mechaARO(\tri_{\qo})}
\leq \frac{1 - \qo_d}{1-\qo_d-\precision} \cdot 
\frac{\OPT_{\tri_{\qo_d}}(\tri_{\qo_d})}{\mechaARO(\tri_{\qo_d})}
\leq 1.9041,  
\end{align*}
where the first inequality holds by applying \Cref{lem:opt rev} and \Cref{clm:cont in q}
and the last inequality holds because $\qo \leq 0.5$.
For any $\qo \in [0.5, 1]$, we have 
\begin{align*}
\frac{\OPT_{\tri_{\qo}}(\tri_{\qo})}{\mechaARO(\tri_{\qo})}
\leq \frac{\OPT_{\tri_{0.5}}(\tri_{0.5})}{\alpha^* \cdot \mecha_1(\tri_{\qo})}
\leq 1.88,  
\end{align*}
and combining the inequalities gives the desired result.
% and apply
% \Cref{clm:rev of r}, and \Cref{clm:cont in q} to
% get the bounds on the approximation ratios.
\end{proof}

The main difference for proving the gap for distributions comparing to
triangle distributions is that the space for regular distributions is
larger and we need to define more carefully what it means to have a
distribution close to the worst case triangle distribution
$\tri_{\qo}$.  By \Cref{thm:pi optimal mechanism}, let $\mechaARO$ be
the approximately optimal prior-independent mechanism with numerical parameters $\optweight =
\aval, \optratio = \rval$. 
We first define a pseudo-mechanism
$\pseudomech$ as the affine combination of mechanisms with revenue
$\pseudomech(\vals) = (1+\delta) \mechaARO(\vals) - \delta\cdot
\mecha_{1.18}(\vals)$, where $\delta > 0$ is a small constant.  Recall
that $\regular$ is the family of regular distributions.  We show that
\begin{enumerate}
\item $\forall \dist \in \regular,\ \pseudomech(\dist) \leq \optf{\dist}$;

\item $\forall \dist \in \regular,\ \piratio'\,\pseudomech(\dist) \geq
  \optf{\dist}$, where $\piratio' < \piratio$.
\end{enumerate}
Letting $\gapbm(\vals) = \piratio' \cdot \pseudomech(\vals)$, 
benchmark $\gapbm$ is normalized for all regular distributions
and the approximation of optimal payoff against benchmark $\gapbm$ is $\piratio' < \piratio = \resolution$. 
Therefore, the approximation of the optimal heuristic benchmark 
is $\heuristic \leq \piratio' < \resolution$
and \Cref{thm:gap} holds. 

The first statement holds by \Cref{lem:upper bound benchmark}. 
To show the second statement holds, 
since we know the worst case distribution $\tri_{\quant^*}$ for the optimal prior-independent mechanism $\mechaARO$, 
intuitively, it is sufficient to show that mechanism $\mecha_{1.18}$
has lower revenue for any distribution close to distribution $\tri_{\quant^*}$. 
Formally, for any $0 < \quant_1 < \quant_2 < 1, 0 < \delta_1, \delta_2 < 1$, 
let $\bar{\dist}$ be the distribution 
such that 
\begin{equation*}
\rev_{\bar{\dist}}(\quant) = 
\begin{cases}
\frac{\quant(1-\delta_1)}{\quant_1} + \delta_1 & 0 \leq \quant \leq \quant_1,\\
1 & \quant_1 < \quant < \quant_2, \\
\frac{\quant(\delta_2 - 1)}{1-\quant_2} + \delta_2 + \frac{1-\delta_2}{1-\quant_2} & \quant_2 \leq \quant \leq 1. 
\end{cases}
\end{equation*}
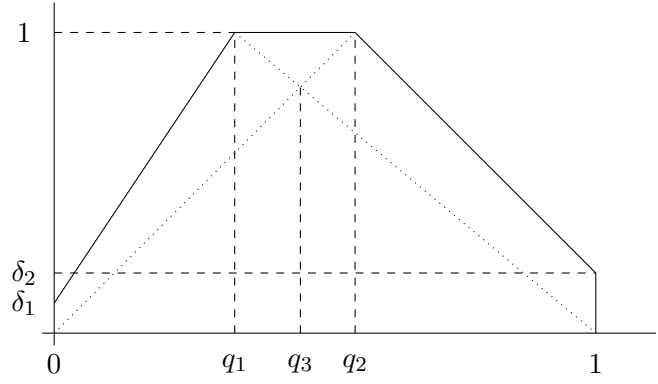
\begin{figure}[t]
\centering
\begin{tikzpicture}[scale = 0.8]

\draw (-0.2,0) -- (10, 0);
\draw (0, -0.2) -- (0, 5.5);

\draw (0, 0.5) -- (3, 5);
\draw (3, 5) -- (5, 5);
\draw (5, 5) -- (9, 1);
\draw (9, 1) -- (9, 0);

\draw [dashed] (0, 1) -- (9, 1);
\draw [dashed] (0, 5) -- (3, 5);

\draw [dotted] (0, 0) -- (5, 5);
\draw [dotted] (3, 5) -- (9, 0);
\draw [dashed] (4.09091, 0) -- (4.09091, 4.09091);
\draw [dashed] (3, 0) -- (3, 5);
\draw [dashed] (5, 0) -- (5, 5);

% \draw (3, 0) -- (3, 2);
% \draw (3, 2) -- (6, 2);
% \draw (6, 2) -- (6, 4);
% \draw (6, 4) -- (11, 4);

% \begin{scope}[ultra  thick]

% \draw [dashed] (0, 0) -- (3, 0);
% \draw [dashed] (3, 0) -- (3, 2);
% \draw [dashed] (3, 2) -- (11, 2);

% \end{scope}

\draw (0, -0.5) node {$0$};
% \draw (-0.2, 4.8) -- (0.2, 4.8);
\draw (-0.5, 5) node {$1$};

\draw (-0.5, 0.5) node {$\delta_1$};
\draw (-0.5, 1) node {$\delta_2$};

\draw (9, -0.5) node {$1$};
\draw (3, -0.5) node {$\quant_1$};
\draw (5, -0.5) node {$\quant_2$};
\draw (4.09091, -0.5) node {$\quant_3$};

\end{tikzpicture}
\caption{\label{f:dist bound}
The revenue curve for distribution $\bar{\dist}$.
}
\end{figure}
The revenue curve for distribution $\bar{\dist}$ is shown in \Cref{f:dist bound}. 
Note that if distribution $\dist$ is close to distribution $\tri_{\quant^*}$, 
then naturally we have $\dist$ 
stochastically dominated by distribution $\bar{\dist}$. 
Before proving \Cref{thm:gap}, 
we first propose a formula for upper bounding the revenue of $\mecha_{\ratio}$
for any distribution $\dist$ stochastically dominated by $\bar{\dist}$, 
which will be used later to upper bound the revenue of mechanism $\mecha_{1.18}$ for any distribution close to $\tri_{\qo^*}$. 
According to \Cref{f:dist bound}, 
letting $\quant_3 = \sfrac{\quant_2}{(1-\quant_1 + \quant_2)}$
and $\quant_4 = \quantf_{\bar{\dist}}(\sfrac{1}{\quant_2})= \frac{\ratio \quant_2 (1-\quant_2\delta_2)}{1 - \quant_2 + \ratio \quant_2(1-\delta_2)} $, 
we upper bound the revenue of mechanism $\mech_r$
for distribution $\dist$ by 
\begin{align}
&\frac{1}{2} \cdot \mecha_{\ratio}(\dist) 
= \int_0^1 \revv_{\dist}(\valf_{\dist}(\quant) \cdot \ratio) \, d\quant \nonumber\\
=& 
\int_0^{\quant_3} \revv_{\dist}(\valf_{\dist}(\quant) \cdot \ratio) \, d\quant
+ \int_{\quant_3}^{\ScaledQuant(\quant_1, \sfrac{1}{\ratio})} 
\revv_{\dist}(\valf_{\dist}(\quant) \cdot \ratio) \, d\quant \nonumber\\
&+
\int_{\ScaledQuant(\quant_1, \sfrac{1}{\ratio})}^{\quant_4} 
\revv_{\dist}(\valf_{\dist}(\quant) \cdot \ratio) \, d\quant
+ \int_{\quant_4}^{1} 
\revv_{\dist}(\valf_{\dist}(\quant) \cdot \ratio) \, d\quant \nonumber\\
\leq& 
\frac{\ratio \quant_1 \quant_3 \delta_1 }{\quant_1\ratio - (1-\delta_1)\quant_2}
+ \int_{\quant_3}^{\ScaledQuant(\quant_1, \sfrac{1}{\ratio})} 
\revv_{\dist}\left(\frac{\ratio(1-\quant)}{1-\quant_1}\right) \, d\quant
+ 
\quant_4
- \frac{\ratio \quant_1}{\ratio - \quant_1 + \ratio \quant_1}
+ \int_{\quant_4}^{1} 
\revv_{\dist}(\valf_{\dist}(\quant) \cdot \ratio) \, d\quant \nonumber\\
\leq&
\frac{\ratio \quant_1 \quant_3 \delta_1 }{\quant_1\ratio - (1-\delta_1)\quant_2}
+ \int_{\quant_3}^{\ScaledQuant(\quant_1, \sfrac{1}{\ratio})} 
\frac{r\quant_1\delta_1(1-\quant)}
{r\quant_1 (1-\quant) 
- \quant(1-\delta_1)(1-\quant_1)} \, d\quant \nonumber\\
&+ \quant_4 
- \frac{\ratio \quant_1}{\ratio - \quant_1 + \ratio \quant_1}
+ \int_{\quant_4}^{1} 
\frac{1}{1-\quant_2}\left( 1-\quant_2\delta_2
- \frac{(1-\quant_2\delta_2)(1-\delta_2)}
{-\ratio(1-\delta_2) + \sfrac{\ratio(1-\quant_2\delta_2)}{\quant} + (1-\delta_2)} 
\right)\, d\quant. \label{eq:upper bound rev mr}
\end{align}
The first inequality holds because for any $\quant \leq \quant_3$, 
the price $\valf_{\dist}(\quant) \cdot \ratio$ is higher than $\sfrac{1}{\quant_1}$ and the revenue curve has positive slope. 
Therefore, the revenue is higher if the price is lower, which is at least $\sfrac{1}{\quant_2}$. 
Moreover, $\revv_{\dist}(\valf_{\dist}(\quant) \cdot \ratio)  \leq 1$ for any quantile $\quant$. 
The second inequality holds similarly since the slope on revenue curve is negative for $\quant \geq \quant_4$. 
With this upper bound on the revenue, 
we can continue to the numerical calculation for the gap between two programs. 

\begin{proof}[Proof of \Cref{thm:gap}]
We prove the gap by upper bounding 
the revenue of mechanism $\mecha_{1.18}$ 
and lower bounding the revenue of mechanism $\mechaARO$ 
for regular distributions parameterized by their monopoly quantile $\qo$. 

First we upper bounding 
the revenue of mechanism $\mecha_{1.18}$. 
By substituting $\ratio = 1.18, \quant_1 = 0.09, \quant_2 = 0.098, \delta_1 = \delta_2 = 0.01$ 
as parameter for distribution $\bar{\dist}$, 
we have $\mecha_{1.18}(\dist) \leq 0.98444$
for any distribution $\dist$ stochastically dominated by $\bar{\dist}$
by Inequality \eqref{eq:upper bound rev mr}, 
and $\mecha_{1.18}(\dist) \leq \optf{\dist}$ otherwise. 

Next we bound the revenue of mechanism $\mechaARO$ for different sets of regular distributions. 
If the distribution $\dist$ is stochastically dominated by $\bar{\dist}$,
by numerical result in \Cref{app:prior-independent}, 
the revenue of mechanism $\mechaARO$ is at least $\sfrac{1}{1.90689431} \cdot \optf{\dist}$. 
If distribution $\dist$ is not stochastically dominated by $\bar{\dist}$,
we can divide the analysis for mechanism $\mechaARO$ into two cases. 
\begin{enumerate}
\item When $\qo \not\in [0.0905, 0.096]$. 
In this case, the approximation ratio of mechanism $\mechaARO$ for distribution $\dist$ is 
at most the approximation ratio for triangle distribution $\tri_{\qo}$, 
which is at least $1.9068845$ 
by discretizing the quantile space and applying numerical computation using \Cref{clm:rev of r} and \Cref{clm:cont in q}, 
i.e., the revenue of mechanism $\mechaARO$ in this case is at least $0.5244156 \cdot \optf{\dist}$. 

\item When $\qo \in [0.0905, 0.096]$, 
since the distribution $\dist$ is not stochastically dominated by $\bar{\dist}$, 
distribution $\dist$ is not a triangle distribution. 
Here we consider two subcases. 
\begin{enumerate}
\item There exists $\quant > 0.096$ such that $\rev_{\dist}(\quant) > \rev_{\bar{\dist}}$. 
In this case, the distribution $\dist$ maximizes the approximation ratio is the truncated distribution
and the optimal revenue for distribution $\dist$ is $2-\qo$. 
Since distribution $\dist$ is not a triangle distribution, 
and by \Cref{clm:continuous}, 
the increase of revenue for mechanism $\mechaARO$ is at least $0.0005$, 
and thus
\begin{align*}
\mechaARO(\dist) \geq \frac{1}{1.90689431}\cdot \OPT_{\dist}(\dist) + 0.0005
\geq 0.5246 \OPT_{\dist}(\dist)
\end{align*}
where the last inequality holds because $\OPT_{\dist}(\dist) \leq 2$. 

\item There exists $\quant < 0.0905$ such that $\rev_{\dist}(\quant) > \rev_{\bar{\dist}}$. 
In this case, the revenue of the optimal mechanism 
for distribution $\dist$
increases by at least $0.0009$ comparing to triangle distribution $\tri_{\qo}$, 
and the second price auction increases by at least this much. 
Therefore, 
\begin{align*}
\frac{\mechaARO(\dist)}{\OPT_{\dist}(\dist)} \geq
\frac{\mechaARO(\tri_{\qo}) + 0.0009\alpha^*}{\OPT_{\tri_{\qo}}(\tri_{\qo}) + 0.0009}
\geq 0.5245
\end{align*}
where the last inequality holds because $\mechaARO(\tri_{\qo}) \geq 0.5244156\cdot \OPT_{\dist}(\tri_{\qo})$, 
and $\OPT_{\tri_{\qo}}(\tri_{\qo}) \geq 2-\qo \geq 1.905$ for $\qo \leq 0.095$. 
\end{enumerate}
Combining the two subcases, we have that 
% So there exists quantile $\quant$ such that $\rev_{\dist}(\quant) > \rev_{\bar{\dist}}(\quant)$. 
% We can discretize the space of distribution
% and since the revenue of the mechanism changes continuously
% when the distribution changes in the similar sense of \Cref{clm:cont in q}, 
% by numerical computation, we have that in this case 
the revenue of mechanism $\mechaARO$
is at least $0.5245 \cdot \optf{\dist}$
for $\dist$ not dominated by $\bar{\dist}$
and $\qo \in [0.091, 0.095]$.
\end{enumerate}

By setting $\delta = 5.21\times 10^{-6}$, for both cases, 
we have $\piratio' \leq 1.90689356 \leq \piratio - 7\times 10^{-7}$, 
where the last inequality holds because 
as we illustrated in the numerical result in \Cref{app:prior-independent}, 
the optimal prior independent approximation ratio is 
$\beta = 2-\qo^* \geq 2-0.09310571 = 1.90689429$.
By setting the benchmark $\gapbm(\vals) = 
1.90689422[(1+\delta) \mechaARO(\vals) - \delta\cdot \mecha_{1.18}(\vals)]$, 
benchmark $\gapbm$ is normalized 
and the approximation ratio of the optimal revenue against benchmark $\gapbm$ 
is at most $1.90689422$. 
Therefore, the approximation of optimal revenue against the optimal heuristic benchmark 
is 
at most the approximation against 
benchmark $\gapbm$, 
which implies 
$\heuristic \leq \piratio' < \piratio = \resolution$.
\end{proof}

\section{Missing Proof from \Cref{sec:online-learning}}
\label{app:learning}

Recall from the proof of \Cref{thm:pi learning} that permutation
distributions are defined by probabilities $\{\mean^{j}\}_{j=1}^n$
assigned to experts via a uniform random permutation $\permute$ (the
reward of expert $j$ is Bernoulli with mean $\mean_j =
\mean^{\permute(j)}$).  Here we formally prove that, for permutation
distributions, follow-the-leader $\FTL$ chooses the optimal expert
according to the posterior distribution induced by the realized reward
history.

At any time $t$, 
let the cumulative reward $\cumval_{t, j}= \sum\nolimits_{t' < t} \val_{t', j}$ 
be the total number of 1's realized for expert $j$ until time $t-1$
and let 
$$\hsact_t = \argmax\nolimits_{j} \cumval_{t, j}$$
be the best expert chosen by $\FTL$. 
Let 
$$\postact_t = \argmax\nolimits_{j} \expect{\val_{t,j} \given \{\val_{t'}\}_{t'<t}}$$
be expert that maximizes the expected reward for the posterior distribution. 
\citet{HKM-15} show that when there are two experts, 
% The following lemma implies that 
% $\hsact_t = \postact_t$ for any time $t$. 
the expert $\postact_t$ that maximizes the expected payoff according to 
the posterior 
is the expert $\hsact_t$
that maximizes the cumulative reward $\cumval_{t, j}$. 
That is, $\postact_t = \hsact_t$ for any time $t$. 
In \Cref{lem:posterior expert and FTL}, 
we generalize the result to $m$ experts for any $m\geq 2$.

\begin{definition}\label{def:non-degenerate}
Let $\mean_1 \geq \mean_2 \geq \cdots \geq \mean_m$ be a sequence of probabilities for binary independent stationary distributions $\dist$. 
Distribution $\dist$ is non-degenerate if 
$\mean_1 > \mean_m$ and there exists $k$ such that $\mean_{k} \in (0, 1)$.
\end{definition}

We prove the results in this section for non-degenerate distributions, 
and the degenerate case is trivial.

\begin{lemma}[Lemma 4.3 in \citealp{HKM-15}]\label{lem:2 expert}
At any time $t$, when there are two experts $1, 2$ with cumulative
reward $\cumval_{t, 1} > \cumval_{t, 2}$ and non-degenerate
distribution with probability $\mean_1 > \mean_2$,
$$\Pr[\permute(1) = 1 \text{ and }\permute(2) = 2 \given
  \{\val_{t'}\}_{t'<t}] > \Pr[\permute(1) = 2 \text{ and }\permute(2)
  = 1 \given \{\val_{t'}\}_{t'<t}];$$ the probabilities are equal if
$\cumval_{t, 1} = \cumval_{t, 2}$ or $\mean_1 = \mean_2$.
\end{lemma}
\begin{proof}
The proof can be found in \citet{HKM-15}, and we include it here for completeness. 
The lemma trivially holds if $\mean_1 = 1$ or $\mean_2 = 0$. 
Otherwise, we have $\mean_1, \mean_2 \in (0,1)$ and 
\begin{align*}
\frac{\Pr[\permute(1) = 1 \text{ and }\permute(2) = 2 \given \{\val_{t'}\}_{t'<t}]}
{\Pr[\permute(1) = 2 \text{ and }\permute(2) = 1 \given \{\val_{t'}\}_{t'<t}]}
&= \frac{\mean_2^{\cumval_{t,2}} \cdot (1-\mean_2)^{t-\cumval_{t,2}} 
\cdot \mean_1^{\cumval_{t,1}} \cdot (1-\mean_1)^{t-\cumval_{t,1}}}
{\mean_2^{\cumval_{t,1}} \cdot (1-\mean_2)^{t-\cumval_{t,1}} 
\cdot \mean_1^{\cumval_{t,2}} \cdot (1-\mean_1)^{t-\cumval_{t,2}}}\\
&= \left(1+\frac{\mean_1-\mean_2}{\mean_2(1-\mean_1)}\right)^{\cumval_{t,1}-\cumval_{t,2}}
\end{align*}
which is strictly larger than 1; 
the inequality is an equality if $\cumval_{t, 1} = \cumval_{t, 2}$ or $\mean_1 = \mean_2$. 
\end{proof}

\begin{lemma}\label{lem:posterior expert and FTL}
For any $m\geq 2$, any non-degenerate distribution $\dist$, any pair
of experts $j, j'$ and for any time $t$, $\expect{\val_{t,j} \given
  \{\val_{t'}\}_{t'<t}} > \expect{\val_{t,j'} \given
  \{\val_{t'}\}_{t'<t}}$ if $\cumval_{t, j} > \cumval_{t, j'}$, and
the expectations are equal if $\cumval_{t, j} = \cumval_{t, j'}$.
% and distribution $\dist$ is non-degenerate.
\end{lemma}

\begin{proof}%[Proof of \Cref{lem:posterior expert and FTL}]
Let $\mean_1 \geq \mean_2 \geq \cdots \geq \mean_m$ be a sequence of probabilities for non-degenerate distribution $\dist$. 
Suppose there is a uniform random permutation $\permute$ 
such that the reward distribution for expert $j$ is $\mean_{\permute(j)}$. 
% The following lemma shows that experts with higher cumulative reward $\cumval_{t, j}$
% have higher probability being the expert with high expectation when there are two experts.  
% Note that the posterior reward for expert $j$ is 
% \begin{align*}
% \expect{\val_{t,j} \given \{\val_{t'}\}_{t'<t}}
% = \sum\nolimits_{s} \mean_{s} \Pr[\permute(j) = s \given \{\val_{t'}\}_{t'<t}].
% \end{align*}
For any expert $j'$ with $\cumval_{t, j} > \cumval_{t, j'}$, 
the expected reward conditional on the reward from time $1$ to $t-1$ is 
\begin{align*}
\expect{\val_{t,j} \given \{\val_{t'}\}_{t'<t}}
&= \sum\nolimits_{s} \mean_{s} \Pr[\permute(j) = s \given \{\val_{t'}\}_{t'<t}]\\
&= \sum\nolimits_{s, s'} \mean_{s} \Pr[\permute(j) = s \text{ and } \permute(j') = s' \given \{\val_{t'}\}_{t'<t}]\\
&> \sum\nolimits_{s, s'} \mean_{s} \Pr[\permute(j') = s \text{ and } \permute(j) = s' \given \{\val_{t'}\}_{t'<t}]\\
&= \sum\nolimits_{s} \mean_{s} \Pr[\permute(j') = s \given \{\val_{t'}\}_{t'<t}]\\
&= \expect{\val_{t,j'} \given \{\val_{t'}\}_{t'<t}}, 
\end{align*}
where the inequality holds by \Cref{lem:2 expert}; it becomes an equality
if $\cumval_{t, j} = \cumval_{t, j'}$.
% and distribution $\dist$ is non-degenerate. 
\end{proof}

From \Cref{lem:posterior expert and FTL}, we immediately have the following theorem. 
\begin{theorem}\label{thm:FTL and bayesian updating are strict opt}
For the uniform permutation of any non-degenerate binary independent
stationary distribution $\dist$, the follow-the-leader algorithm
$\FTL$ is the strict optimal algorithm that maximizes the expected
reward given the posterior distribution in each round $\round$.
\end{theorem}

\section{Gap for Expert Learning}
\label{app:gap learning}

In this section, we show that the gap between program
\eqref{eq:heuristic} and \eqref{eq:resolution} exists for expert
learning as well.
First we make the following observation for follow-the-leader
with non-degenerate distributions. 
\begin{lemma}\label{lem:FTL increasing}
For any non-degenerate binary independent stationary distribution
$\dist$, the expected performance of
follow-the-leader at any round $\round$ is less than its expected
performance at round $\round+1$, i.e., $\expect{\FTL_\round(\vals)} <
\expect{\FTL_{\round+1}(\vals)}$.
\end{lemma}
\begin{proof}
  Consider the uniform permutation environment under any
  non-degenerate distribution $\dist$.  \Cref{thm:FTL and bayesian
    updating are strict opt} shows that follow-the-leader is strictly
  optimal in every round.  Specifically, any algorithm that selects an
  expert that does not have the highest historical reward obtains
  strictly smaller expected reward in this round.

  Now consider two algorithms for selecting an expert in round
  $\round+1$:
  \begin{enumerate}
  \item follow the leader from rounds $\{1,\ldots,\round-1\}$, or
  \item follow the leader from rounds $\{1,\ldots,\round\}$.
  \end{enumerate}
  The expected round $\round+1$ payoffs of these algorithms are
  $\expect{\FTL_\round(\vals)}$ and $\expect{\FTL_{\round+1}(\vals)}$,
  respectively.  Under the assumption that the distribution is
  non-degenerate, there is a positive probability that the sets of
  experts with the highest historical rewards before round $\round$
  and $\round+1$ are distinct; thus, by \Cref{thm:FTL and bayesian
    updating are strict opt} the payoff of the former is strictly less
  than the payoff of the latter.

  Of course, the follow-the-leader algorithm obtains the same payoff
  regardless of the permutation of probabilities to experts; thus, the
  strict inequality of the lemma holds for any non-degenerate binary
  independent stationary distribution $\dist$.
\end{proof}

\begin{theorem}\label{thm:gap learning}
  For binary independent stationary reward distributions in expert
  learning problem with time horizon $\horizon \geq 2$, the heuristic
  benchmark optimization program~\eqref{eq:heuristic} has a strictly
  smaller objective value than the benchmark optimization
  program~\eqref{eq:resolution}, i.e., $\heuristic < \resolution$.
\end{theorem}
\begin{proof}
For benchmark optimization program~\eqref{eq:resolution}, 
the optimal benchmark is scaled-up follow-the-leader $\piratio\cdot \FTL$, 
where $\piratio$ is the prior-independent approximation ratio for follow-the-leader $\FTL$,
and $\resolution = \piratio$. 

Now consider another benchmark $\benchmark(\vals) = \piratio' \pseudomech(\vals)$, 
where $\pseudomech(\vals) = \horizon \cdot \FTL_n(\vals)$
% and $\bestn = \max\nolimits_{j} \sum\nolimits_{t< n} \val_{t, j}$ 
% with tie-breaking uniform randomly, 
% for $n\geq 2$, 
% $\bestn = 1$ with probability $\frac{1}{2}$ for $n = 1$, 
and $\piratio'$ is the prior-independent approximation ratio for $\pseudomech$. 
It is easy to verify that for any binary independent stationary $\dist$ 
with probabilities $\mean_1,\ldots,\mean_\numexpert$,
\begin{align*}
\pseudomech(\dist) 
= \horizon \cdot \expect{\FTL_{\horizon}(\vals)}
\leq \horizon \cdot \max\nolimits_j \mean_j
= \OPT_{\dist}(\dist).
\end{align*}
Therefore, benchmark $\benchmark$ is normalized and 
the resolution of benchmark $\benchmark$ with respect to program \eqref{eq:heuristic} is at most $\piratio'$. 
It is sufficient to prove $\piratio' < \piratio$ to prove \Cref{thm:gap learning}. 

First we show that degenerate distributions are not worst case for prior-independent approximation 
for either follow-the-leader $\FTL$
or pseudo-mechanism $\pseudomech$.
Note that for any distribution such that the expected reward for all experts are the same, 
we have $\FTL(\dist) = \pseudomech(\dist) = \OPT_{\dist}(\dist)$, 
and hence such a distribution is not the worst case. 
% for prior-independent approximation for 
% either $\FTL$ or $\pseudomech$. 
Moreover, for degenerate distribution $\dist$ 
such that the expected reward for experts are not all the same,
% such that 
% the probability $p_j$ of receiving reward 1 for each expert $j$ is either 1 or 0, 
we can regenerate a non-degenerate distribution $\dist'$ such that 
the prior-independent approximation ratio for both follow-the-leader $\FTL$
and pseudo-mechanism $\pseudomech$ are worse, 
by switching the distribution of receiving reward 1 with probability 0, 
to the distribution of receiving reward 1 with probability $\frac{1}{2}$. 
This does not affect the expected performance of the optimal mechanism, 
but strictly decreases the expected performance of follow-the-leader $\FTL$
and pseudo-mechanism $\pseudomech$.

% Next we note that for any binary distribution $\dist$ such that 
% the probability $p_j$ of receiving reward 1 for each expert $j$ is either 1 or 0, 
% distribution $\dist'$ such that $\dist'_j = \dist_j$ if $p_j = 1$ 
% and $\dist'_j = U\{0,1\}$ is the uniform distribution over support $\{0,1\}$ if $p_j = 1$. 
% In this case, $\OPT_{\dist}(\dist) = \OPT_{\dist'}(\dist')$
% while $\FTL(\dist) > \FTL(\dist')$ and $\pseudomech(\dist) > \pseudomech(\dist')$ 
% for any $n\geq 2$. 
% Hence such a distribution $\dist$ is also not the worst case 
% for prior-independent approximation for 
% either $\FTL$ or $\pseudomech$. 

Now we focus on distribution $\dist$ that is non-degenerate.
% where there exists $p_k, p_{k'}$
% such that $p_k > p_{k'}$, 
% and there exists $k''$ such that $p_{k''} \in (0, 1)$. 
We show that for any such distribution, 
$\FTL(\dist) < \pseudomech(\dist)$, 
which implies $\piratio' < \piratio$. 
Formally, we have that 
\begin{align*}
\FTL(\dist) &= \expect[\vals]{\sum\nolimits_{t\leq n}\FTL_t(\vals)}
< \expect[\vals]{\sum\nolimits_{t\leq n}\FTL_n(\vals)}
= \pseudomech(\dist)
\end{align*}
where the inequality holds by \Cref{lem:FTL increasing} and $n\geq 2$. 
\end{proof}

\end{document}